\documentclass[a4paper,11pt]{article}
\usepackage[utf8]{inputenc}
\usepackage[T1]{fontenc}
\usepackage{keyval}
\usepackage{footmisc}
\usepackage{ragged2e}
\usepackage{xspace}
\usepackage[usenames,dvipsnames,svgnames,table]{xcolor}
\usepackage[final,stretch=10,shrink=10]{microtype}
\usepackage[showlabels,sections,floats,textmath,displaymath]{preview}
\usepackage{gitinfo}
\usepackage{ifdraft}
\usepackage{enumitem}
\usepackage{footnote}
\usepackage{tikz}
\usetikzlibrary{backgrounds,snakes,arrows,decorations.markings,calc}
\usepackage{algpseudocode}
\usepackage{algorithm}

\usepackage{caption}
\DeclareCaptionLabelSeparator{dot}{.\enspace}
\clearcaptionsetup{figure}
\usepackage{macros}
\usepackage{cite}
\usepackage{amsmath}
\usepackage{amssymb}
\usepackage{amsthm}
\usepackage{thmtools}
\newtheorem{definition}{Definition}
\newtheorem{theorem}{Theorem}
\newtheorem{lemma}{Lemma}
\newtheorem{example}{Example}
\setlistdepth{9}
\setlist[itemize,1]{label=$\bullet$}
\setlist[itemize,2]{label=$\bullet$}
\setlist[itemize,3]{label=$\bullet$}
\setlist[itemize,4]{label=$\bullet$}
\setlist[itemize,5]{label=$\bullet$}
\setlist[itemize,6]{label=$\bullet$}
\setlist[itemize,7]{label=$\bullet$}
\setlist[itemize,8]{label=$\bullet$}
\setlist[itemize,9]{label=$\bullet$}

\renewlist{itemize}{itemize}{9}

\usepackage[english]{babel}
\usepackage[%
  final,
  pdftitle={An Expressive Model for the Web
  Infrastructure: Definition and Application to the
  BrowserID SSO System},%
  pdfauthor={Daniel Fett, Ralf Kuesters, Guido Schmitz},%
  pdfcreator={none},%
  pdfproducer={infsec.uni-trier.de},%
  pdfkeywords={}]{hyperref}
\usepackage{color}
\definecolor{linkcolor}{rgb}{0,0,0.5}
\hypersetup{
 colorlinks=true,
  linkcolor=linkcolor,
  anchorcolor=linkcolor,
  citecolor=linkcolor,
  filecolor=linkcolor,
  menucolor=linkcolor,
  runcolor=linkcolor,
  urlcolor=linkcolor,
}

\usepackage{fancybox}
\usepackage[headings]{fullpage}

\begin{document}
\overfullrule=5pt
\hyphenation{Brow-serID}
\hyphenation{in-fra-struc-ture}
\hyphenation{brow-ser}
\hyphenation{doc-u-ment}
\hyphenation{Chro-mi-um}
\hyphenation{meth-od}

\title{An Expressive Model for the Web
  Infrastructure:\\Definition and Application to the
  BrowserID SSO System\footnote{An abridged version appears
    in S\&P 2014 \cite{FettKuestersSchmitz-SP-2014}.}}
\author{Daniel Fett, Ralf Küsters, Guido Schmitz\\
        University of Trier, Germany\\
        \{fett,kuesters,schmitzg\}@uni-trier.de}
\date{}

\maketitle

\begin{abstract}
  The web constitutes a complex infrastructure and as
demonstrated by numerous attacks, rigorous analysis of
standards and web applications is indispensable.

Inspired by successful prior work, in particular the work
by Akhawe et al.~as well as Bansal et al., in this work we
propose a formal model for the web infrastructure. While
unlike prior works, which aim at automatic analysis, our
model so far is not directly amenable to automation, it is
much more comprehensive and accurate with respect to the
standards and specifications. As such, it can serve as a
solid basis for the analysis of a broad range of standards
and applications.

As a case study and another important contribution of our
work, we use our model to carry out the first rigorous
analysis of the BrowserID system (a.k.a.~Mozilla Persona),
a recently developed complex real-world single sign-on
system that employs technologies such as AJAX,
cross-document messaging, and HTML5 web storage. Our
analysis revealed a number of very critical flaws that
could not have been captured in prior models. We propose
fixes for the flaws, formally state relevant security
properties, and prove that the fixed system in a setting
with a so-called secondary identity provider satisfies
these security properties in our model. The fixes for the
most critical flaws have already been adopted by Mozilla
and our findings have been rewarded by the Mozilla Security
Bug Bounty Program.

\end{abstract}

\vfill
\pagebreak{}
\pagestyle{headings}

\tableofcontents

\newpage{}

\section{Introduction}
The World Wide Web is a complex infrastructure, with a rich
set of security requirements and entities, such as DNS
servers, web servers, and web browsers, interacting using
diverse technologies. New technologies and standards (for
example, HTML5 and related technologies) introduce even
more complexity and security issues. As illustrated by
numerous attacks (see, e.g.,
\cite{WangCheckWangQuadeer-SP-2011-shop-free,SantsaiBeznosov-CCS-2012-OAuth,KarlofShankarTygarWagner-CCS-2007-dynamic-pharming,AkhawBarthLamMitchellSong-CSF-2010,BansalBhargavanMaffeis-CSF-2012}),
rigorous analysis of the web infrastructure and web
applications is indispensable.

Inspired by successful prior work, in particular the work
by Akhawe et al.\cite{AkhawBarthLamMitchellSong-CSF-2010}
and Bansal et
al.\cite{BansalBhargavanetal-POST-2013-WebSpi,BansalBhargavanMaffeis-CSF-2012},
one goal of our work is to develop an expressive formal
model that precisely captures core security aspects of the
web infrastructure, where we intend to stay as close to the
standards as possible, with a level of abstraction that is
suitable for precise formal analysis. As further discussed
in Section~\ref{sec:relatedwork}, while prior work aimed at
automatic analysis, here our main focus is to obtain a
comprehensive and more accurate model with respect to the
standards and specifications. As such, our model
constitutes a solid basis for the analysis of a broad range
of standards and applications.

The standards and specifications that define the web are
spread across many documents, including the HTTP standard
RFC2616 (with its successor HTTPbis) and the HTML5
specification \cite{html5}, with certain aspects covered in
related documents, such as RFC6265, RFC6797, RFC6454, the
WHATWG Fetch living standard \cite{whatwg/fetch}, the W3C
Web Storage specification \cite{w3c/webstorage}, and the
W3C Cross-Origin Resource Sharing specification
\cite{w3c/cors}, to name just a few. Specifications for the
DNS system and communication protocols, such as TCP, are
relevant as well.  The documents often build upon each
other, replace older versions or other documents, and
sometimes different versions coexist. Some details or
behaviors are not specified at all and are only documented
in the form of the source code of a web browser.

Coming up with an accurate formal model is, hence, very
valuable not only because it is required as a basis to
precisely state security properties and perform formal
analysis, but also because it summarizes and condenses
important aspects in several specifications that are
otherwise spread across different documents.

Another goal and important contribution of our work is to
apply our model to the BrowserID system (also known under
the marketing name \emph{Mozilla Persona}), a complex
real-world single sign-on system developed by
Mozilla. BrowserID makes heavy use of several web
technologies, including AJAX, cross-document messaging
(postMessages), and HTML5 web storage, and as such, is a
very suitable and practically relevant target to
demonstrate the importance of a comprehensive and accurate
model.

More precisely, the main contributions of our work can be
summarized as follows.

\myparagraph{Web model.} We propose a formal model of the
web infrastructure and web applications. Our model is based
on a general Dolev-Yao-style communication model, in which
processes have addresses (modeling IP addresses) and, as
usual in Dolev-Yao-style models for cryptographic protocols
(see, e.g., \cite{AbadiFournet-POPL-2001}), messages are
modeled as formal terms, with properties of cryptographic
primitives, such as encryption and digital signatures,
expressed as equational theories on terms.

As mentioned before, our model is intended to be expressive
and close to the standards and specifications, while
providing a suitable level of abstraction. Our model
includes web servers, web browsers, and DNS servers. We
model HTTP(S) requests and responses, including several
headers, such as host, cookie, location,
strict-transport-security (STS), and origin headers. Our
model of web browsers captures the concepts of windows,
documents, and iframes as well as new technologies, such as
web storage and cross-document messaging. It takes into
account the complex security restrictions that are applied
when accessing or navigating other windows.  JavaScript is
modeled in an abstract way by what we call scripting
processes. These processes can be sent around and, among
others, they can create iframes and initiate
XMLHTTPRequests
(XHRs). %
We also consider two ways of dynamically corrupting
browsers. Altogether, our model is the most comprehensive
model for the web infrastructure to date (see also
Section~\ref{sec:relatedwork}).

\myparagraph{Analysis of the BrowserID system.}
We use our model to perform the first rigorous security
analysis of the BrowserID system, which supports both
so-called primary and secondary identity providers. Our
security analysis reveals a number of very critical
and previously unknown flaws, most of which cannot be
captured by previous models (see
Section~\ref{sec:relatedwork}). The most severe attack
allows an adversary to login to any service that supports
authentication via BrowserID with the email address of
\emph{any} Gmail and Yahoo user (without knowing the
Gmail/Yahoo credentials of these users), hence, breaking
the system completely. Another critical attack allows an
attacker to force a user to login with the attacker's
identity.  We confirmed that the attacks work on the actual
BrowserID implementation.  We propose fixes and formulate
relevant security properties. For the BrowserID system with
a secondary identity provider, we prove that
the fixed system satisfies these properties in our
model. By this, we provide the first rigorous formal
analysis of the BrowserID system.  Our attacks have been
acknowledged by Mozilla, with the fixes for the most severe
problems having been adopted by Mozilla already and other
fixes being under discussion. Our findings have been
rewarded by the Mozilla Security Bug Bounty Program.

\myparagraph{Structure of this Paper.} 
In Section~\ref{sec:commmunicationmodel}, we present the
basic communication model. Our web model is introduced in
Section~\ref{sec:webmodel}. For our case study, we first,
in Section~\ref{sec:browserid}, provide a description of
the BrowserID system. We then, in
Section~\ref{sec:analysisbrowserid}, present the analysis
of BrowserID using our model. Related work is discussed in
Section~\ref{sec:relatedwork}. We conclude in
Section~\ref{sec:conclusion}. Full details are provided in
the appendix, including some notational conventions.

\section{Communication
  Model}\label{sec:commmunicationmodel}

We now present a generic Dolev-Yao-style communication
model on which our web model (see
Section~\ref{sec:webmodel}) is based. While the model is
stated in a concise mathematical fashion, instantiations,
for example, using the applied pi-calculus
\cite{AbadiFournet-POPL-2001} or multi-set rewriting
\cite{DurginLincolnMitchellScedrov-JCS-2004}, are
conceivable. 

The main entities in the communication model are what we
call \emph{atomic processes}, which in
Section~\ref{sec:webmodel} are used to model web browsers,
web servers, DNS servers as well as web and network
attackers. Each \ap has a list of addresses (representing
IP addresses) it listens to.  A set of \aps forms what we
call a \emph{system}. The different atomic processes in
such a system can communicate via events, which consist of
a message as well as a receiver and a sender address.  In
every step of a run one event is chosen
non-deterministically from the current ``pool'' of events
and is delivered to an \ap that listens to the receiver
address of that event; if different atomic processes can
listen to the same address, the atomic process to which the
event is delivered is chosen non-deterministically among
the possible processes. The (chosen) atomic process can
then process the event and output new events, which are
added to the pool of events, and so on. (In our web model,
presented in Section~\ref{sec:webmodel}, only network
attackers may listen to addresses of other atomic
processes.)

\subsection{Terms, Messages and Events} 
To define the communication model just sketched, we first
define, as usual in Dolev-Yao models, messages, such as
HTTP messages, as formal terms over a signature, and based
on this notion of messages, we introduce events.

The signature $\Sigma$ for the terms and messages
considered in this work is the union of the following
pairwise disjoint sets of function symbols: 
\begin{itemize}
\item constants $C = \addresses\,\cup\, \mathbb{S}\cup
  \{\True,\bot,\notdef\}$ where the three sets are pairwise
  disjoint, $\mathbb{S}$ is interpreted to be the set of
  ASCII strings (including the empty string $\varepsilon$),
  and $\addresses$ is interpreted to be a set of (IP)
  addresses,
\item function symbols for public keys,
  asymmetric/symmetric encryption/decryption, and
  signatures: $\mathsf{pub}(\cdot)$, $\enc{\cdot}{\cdot}$,
  $\dec{\cdot}{\cdot}$, $\encs{\cdot}{\cdot}$,
  $\decs{\cdot}{\cdot}$, $\sig{\cdot}{\cdot}$,
  $\checksig{\cdot}{\cdot}$, $\unsig{\cdot}$,
\item $n$-ary
  sequences $\an{}, \an{\cdot}, \an{\cdot,\cdot},
  \an{\cdot,\cdot,\cdot},$ etc., and
\item projection symbols
  $\pi_i(\cdot)$ for all $i \in \mathbb{N}$.
\end{itemize}

Let $X=\{x_0,x_1,\dots\}$ be a set of variables and
$\nonces$ be an infinite set of constants (\emph{nonces})
such that $\Sigma$, $X$, and $\nonces$ are pairwise
disjoint. For $N\subseteq\nonces$, we define the set
$\gterms_N(X)$ of \emph{terms} over $\Sigma\cup N\cup X$
inductively as usual: (1) If $t\in N\cup X$, then $t$ is a
term. (2) If $f\in \Sigma$ is an $n$-ary function symbol in
$\Sigma$ for some $n\ge 0$ and $t_1,\ldots,t_n$ are terms,
then $f(t_1,\ldots,t_n)$ is a term. By
$\gterms_N=\gterms_N(\emptyset)$, we denote the set of all
terms over $\Sigma\cup N$ without variables, called
\emph{ground terms}. The set $\messages$ of messages (over
$\nonces$) is defined to be the set of ground terms
$\gterms_{\nonces}$. For example, $k\in \nonces$ and
$\pub(k)$ are messages, where $k$ typically models a
private key and $\pub(k)$ the corresponding public key. For
constants $a$, $b$, $c$ and the nonce $k\in \nonces$, the
message $\enc{\an{a,b,c}}{\pub(k)}$ is interpreted to be
the message $\an{a,b,c}$ (the sequence of constants $a$,
$b$, $c$) encrypted by the public key $\pub(k)$.

For strings, i.e., elements in $\mathbb{S}$, we use a
specific font. For example, $\cHttpReq$ and $\cHttpResp$
are strings. We denote by $\dns\subseteq \mathbb{S}$ the
set of domains, e.g., $\str{www.example.com}\in \dns$.  We
denote by $\methods\subseteq \mathbb{S}$ the set of methods
used in HTTP requests, e.g., $\mGet$, $\mPost\in \methods$.

The equational theory associated with the signature
$\Sigma$ is given in Figure~\ref{fig:equational-theory}.

\begin{figure}
\begin{align}
\dec{\enc x{\pub(y)}}{y} &= x\\
\decs{\encs x{y}}{y} &= x\\
\unsig{\sig{x}{y}} &= x\\
\checksig{\sig{x}{y}}{\pub(y)} &= \True\\
\pi_i(\an{x_1,\dots,x_n}) &= x_i \text{\;\;if\ } 1 \leq i \leq n \\
\proj{j}{\an{x_1,\dots,x_n}} &= \notdef \text{\;\;if\ } j
\not\in \{1,\dots,n\}\\
\proj{j}{t} &= \notdef \text{\;\;if $t$ is not a sequence}
\end{align}
\caption{Equational theory for $\Sigma$.}\label{fig:equational-theory}
\end{figure}

 By $\equiv$ we denote the congruence relation on
$\terms(X)$ induced by this theory. For example, we have
that
$\pi_1(\dec{\enc{\an{\str{a},\str{b}}}{\pub(k)}}{k})\equiv
\str{a}$.

An \emph{event (over $\addresses$ and $\messages$)} is of
the form $(a{:}f{:}m)$, for $a, f\in \addresses$ and $m \in
\messages$, where $a$ is interpreted to be the receiver
address and $f$ is the sender address.  We denote by
$\events$ the set of all events.

\subsection{Atomic Processes, Systems and Runs} We now
define atomic processes, systems, and runs of systems.  

An atomic process takes its current state and an
event as input, and then (non-deterministi\-cally) outputs a new state
and a set of events.
\begin{definition}\label{def:atomic-process-and-process}
  A \emph{(generic) \ap} is a tuple $p = (I^p, Z^p, R^p,
  s^p_0)$ where $I^p \subseteq \addresses$, $Z^p$ is a set
  of states, $R^p\subseteq (\events \times Z^p) \times
  (2^\events \times Z^p)$, and $s^p_0\in Z^p$ is the
  initial state of $p$.  We write $(e,z)R(E,z')$ instead of
  $((e,z),(E,z'))\in R$.

  A \emph{system} $\process$ is a (possibly infinite) set
  of \aps.
\end{definition}
In order to define a run of a system, we first define
configurations and processing steps.

A \emph{configuration of a system $\process$} is a
tuple $(S, E)$ where $S$ maps every atomic process $p\in
\process$ to its current state $S(p)\in Z^p$ and $E$ is a
(possibly infinite) multi-set of events waiting to be
delivered.

A \emph{processing step of the system $\process$} is of the
form $(S,E) \xrightarrow[p \rightarrow E_{\text{out}}]{e
  \rightarrow p} (S', E')$ such that
\begin{itemize}
\item there exist $e = (a{:}f{:}m) \in E$, $E_\text{out}
  \subseteq E'$, and $p \in \process$ with $(e,
  S(p))R^p(E_\text{out}, S'(p))$ and $a \in I^p$,
\item $S'(p') = S(p')$ for all $p' \neq p$, and
\item $E' = (E\setminus \{e\}) \cup E_\text{out}$
  (multi-set operations).
\end{itemize}
 We may omit the superscript
  and/or subscript of the arrow.
\begin{definition}
  Let $\process$ be a system and $E_0$ be a multi-set of
  events. A \emph{run $\rho$ of a system $\process$
    initiated by $E_0$} is a finite sequence of
  configurations $(S_0, E_0),\dots,(S_n, E_n)$ or an
  infinite sequence of configurations $(S_0, E_0),\dots$
  such that $S_0(p) = s_0^p$ for all $p \in \process$ and
  $(S_i, E_i) \xrightarrow{} (S_{i+1}, E_{i+1})$ for all $0
  \leq i < n$ (finite run) or for all $i \geq 0$ (infinite
  run).
\end{definition}
\subsection{Atomic Dolev-Yao Processes}  We next define
atomic Dolev-Yao processes, for which we require that the
messages and states that they output can be computed (more
formally, derived) from the current input event and
state. For this purpose, we first define what it means to
derive a message from given messages.

Let $N\subseteq \nonces$, $\tau \in
\gterms_N(\{x_1,\ldots,x_n\})$, and $t_1,\ldots,t_n\in
\gterms_N$. Then, by $\tau[t_1/x_1,\ldots,t_n/x_n]$ we
denote the (ground) term obtained from $\tau$ by replacing
all occurrences of $x_i$ in $\tau$ by $t_i$, for all $i\in
\{1,\ldots,n\}$.  Let $M\subseteq \messages$ be a set of
messages. We say that \emph{a message $m$ can be derived
  from $M$ with nonces $N$} if there exist $n\ge 0$,
$m_1,\ldots,m_n\in M$, and $\tau\in
\gterms_N(\{x_1,\ldots,x_n\})$ such that $m\equiv
\tau[m_1/x_1,\ldots,m_n/x_n]$. We denote by $d_N(M)$ the
set of all messages that can be derived from $M$ with
nonces $N$. For example, $a\in
d_{\{k\}}(\{\enc{\an{a,b,c}}{\pub(k)}\})$.

\begin{definition} \label{def:adyp} An \emph{atomic
    Dolev-Yao process (or simply, a DY process)} is a tuple $p = (I^p, Z^p, R^p,$
  $s^p_0, N^p)$ such that $(I^p, Z^p, R^p, s^p_0)$ is an
  atomic process and (1) $N^p\subseteq\nonces$ is an
  (initial) set of nonces, (2) $Z^p \subseteq
  \gterms_{\nonces}$ (and hence, $s^p_0\in
  \gterms_{\nonces}$), and (3) for all $a, a', f, f'\in
  \addresses$, $m, m', s, s'\in \gterms_{\nonces}$, set of
  events $E$ with $((a{:}f{:}m), s)R(E, s')$ and $(a'{:}f'{:}m')
  \in E$ it holds true that $m',s' \in d_N(\{m,s\})$. (Note that
  $a',f'\in d_N(\{m,s\})$.)
\end{definition}
In the rest of this paper, we will only consider 
DY processes and assume different DY processes to
have disjoint initial sets of nonces. 

We define a specific  DY process, called an attacker
process, which records all messages it receives and outputs
all messages it can possibly derive from its recorded
messages. Hence, an attacker process is the maximally
powerful DY process. It can carry out all attacks any DY
process could possibly perform. The attacker process is
parametrized by the set of sender addresses it may use.
\begin{definition}\label{def:atomicattacker}
  An \emph{(atomic) attacker process for a set of sender
    addresses $A\subseteq \addresses$} is an atomic DY
  process $p = (I, Z, R, s_0, N)$ such that for all $a,
  f\in \addresses$, $m\in \gterms_{\nonces}$, and $s\in
  Z$ we have that $((a{:}f{:}m), s)R(E,s')$ iff $s'=\an{\an{a,
      f, m}, s}$ and $E=\{(a'{:}f'{:}m')\mid a'\in \addresses$,
  $f'\in A$, $m'\in d_N(\{m,s\})\}$.
\end{definition}

\section{Our Web Model}\label{sec:webmodel}

We now present our web model. We formalize the web
infrastructure and web applications by what we call a web
system. A web system, among others, contains a (possibly
infinite) set of  DY processes, which model web
browsers, web servers, DNS servers as well as web and
network attackers.

As already mentioned in the introduction, the model has
been carefully designed, closely following published
(de-facto) standards, for instance, the HTTP/1.1 standard,
associated (proposed) standards (main\-ly RFCs), and the
HTML5 W3C candidate recommendation. We also checked these
standards against the actual implementations (primarily,
Chromium and Firefox).

\subsection{Web System}\label{sec:websystem}

Before we can define a web system, we define scripting
processes, which model client-side scripting technologies,
such as JavaScript, in our browser model. Scripting
processes are defined similarly to DY processes.

\begin{definition} \label{def:sp} A \emph{scripting
    process} (or simply, a \emph{script}) is a relation
  $R\subseteq (\terms \times 2^{\nonces})\times \terms$
  such that for all $s, s' \in \terms$ and $N\subseteq
  \nonces$ with $(s,N)\,R\,s'$ it follows that $s'\in
  d_N(s)$.
\end{definition}
A script is called by the browser which provides it with a
(fresh, infinite) set $N$ of nonces and state information
$s$. The script then outputs a term $s'$, which represents
the new internal state and some command which
is interpreted by the browser (see
Section~\ref{sec:web-browsers} for details).

Similarly to an attacker process, we define the
\emph{attacker script} $\Rasp$. This script outputs
everything that is derivable from the input, i.e.,
$\Rasp=\{((s,N),s')\mid s\in \terms, N\subseteq \nonces,
s'\in d_N(s)\}$.

We can now define web systems, where we distinguish between
web and network attackers. Unlike web attackers, network
attackers can listen to addresses of other parties and can
spoof the sender address, i.e., they can control the
network. Typically, a web system has either one network
attacker or one or more web attackers, as network attackers
subsume all web attackers. As we will see later, web and
network attacks may corrupt other entities, such as
browsers.

\begin{definition}\label{def:websystem}
  A \emph{web system $\completewebsystem=(\websystem,
    \scriptset,\mathsf{script}, E_0)$} is a tuple with its
  components defined as follows:

  The first component, $\websystem$, denotes a system
  (a set of DY processes) and is partitioned into the
  sets $\mathsf{Hon}$, $\mathsf{Web}$, and $\mathsf{Net}$
  of honest, web attacker, and network attacker processes,
  respectively.  We require that all DY processes in
  $\websystem$ have disjoint sets of nonces, i.e., $N^p\cap
  N^{p'}=\emptyset$ for every distinct $p,p'\in
  \websystem$.

  Every $p\in \mathsf{Web} \cup \mathsf{Net}$ is an
  attacker process for some set of sender addresses
  $A\subseteq \addresses$. For a web attacker $p\in
  \mathsf{Web}$, we require its set of addresses $I^p$ to
  be disjoint from the set of addresses of all other web
  attackers and honest processes, i.e., $I^p\cap I^{p'} =
  \emptyset$ for all $p' \in \mathsf{Hon} \cup
  \mathsf{Web}$. Hence, a web attacker cannot listen to
  traffic intended for other processes. Also, we require
  that $A=I^p$, i.e., a web attacker can only use sender
  addresses it owns. Conversely, a network attacker may
  listen to all addresses (i.e., no restrictions on $I^p$)
  and may spoof all addresses (i.e., the set $A$ may be
  $\addresses$).

  Every $p \in \mathsf{Hon}$ is a DY process which
  models either a \emph{web server}, a \emph{web browser},
  or a \emph{DNS server}, as further described in the
  following subsections. Just as for web attackers, we
  require that $p$ does not spoof sender addresses and that
  its set of addresses $I^p$ is disjoint from those of
  other honest processes and the web attackers. 

  The second component, $\scriptset$, is a finite set of
  scripts such that $\Rasp\in \scriptset$. The third
  component, $\mathsf{script}$, is an injective mapping
  from $\scriptset$ to $\mathbb{S}$, i.e., by
  $\mathsf{script}$ every $s\in \scriptset$ is assigned its
  string representation $\mathsf{script}(s)$. 

  Finally, $E_0$ is a multi-set of events, containing an
  infinite number of events of the form $(a{:}a{:}\trigger)$
  for every $a \in \bigcup_{p\in \websystem} I^p$.

  A \emph{run} of $\completewebsystem$ is a run of
  $\websystem$ initiated by $E_0$.
\end{definition}
In the definition above, the multi-set $E_0$ of initial
events contains for every process and address an infinite
number of $\trigger$ messages in order to make sure that
every process in $\websystem$ can be triggered arbitrarily
often. In particular, by this it is guaranteed that an
adversary (a dishonest server/browser) can be triggered
arbitrarily often. Also, we use trigger events to model
that an honest browser takes an action triggered by a user,
who might, for example, enter a URL or click on some link.

The set $\scriptset\setminus\{\Rasp\}$ specified in a web
system as defined above is meant to describe the set of
honest scripts used in the considered web
application. These scripts are those sent out by an honest
web server to a browser as part of a web application. In
real web applications, possibly several dynamically loaded
scripts may run in one document. However, if these scripts
originate from honest sites, their composition can be
considered to be one honest script (which is loaded right
from the start into the document). In this sense, every
script in $\scriptset\setminus\{\Rasp\}$ models an honest
script or a combination of such scripts in a web
application. (In our case study, the combination is
illustrated by the script running in RP-Doc.)

We model the situation where some malicious script was
loaded into a document by the ``worst-case'' scenario,
i.e., we allow such a script to be the script $\Rasp$. This
script subsumes \emph{everything} any malicious (and
honest) script can do.

We emphasize that script representations being modeled as
strings are public information, i.e., any server or
attacker is free to send out the string representation for
any script.

Since we do not model client-side or server-side language
details, and hence details such as correct escaping of user
input, we cannot analyze whether a server application (say,
written in PHP) is vulnerable to
Cross-Site-Scripting. However, we can model the effects of
Cross-Site-Scripting by letting the (model of the) server
output the script $\Rasp$, say, if it receives certain
malicious input.

In the following subsections, (honest) DNS servers and web
browsers are modeled as  DY processes, including the
modeling of HTTP messages. We also discuss the modeling of
web servers.

\subsection{DNS Servers}\label{sec:DNSservers}

For the sake of brevity, in this paper we consider a flat
DNS model in which DNS queries are answered directly by one
DNS server and always with the same address for a domain. A
full (hierarchical) DNS system with recursive DNS
resolution, DNS caches, etc.~could also be modeled to cover
certain attacks on the DNS system itself.

A \emph{DNS server} $d$ (in a flat DNS model) is modeled in
a straightforward way as a DY process $(I^d,
\{s^d_0\}, R^d, s^d_0, N^d)$. It has a finite set of 
addresses $I^d$ and its initial (and only) state $s^d_0$
encodes a mapping from domain names to addresses of the
form
$$s^d_0=\langle\an{\str{domain}_1,a_1},\an{\str{domain}_2, a_2}, \ldots\rangle \ .$$ DNS
queries are answered according to this table. DNS queries
have the following form, illustrated by an example:
\begin{align*}
  \an{\cDNSresolve,\, \str{example.com},\, n}
\end{align*}
where $\str{example.com}$ is the domain name to be resolved
and $n$ is a nonce representing the random query ID and UDP
source port number selected by the sender of the query. The
corresponding response is of the form 
\begin{align*}
\an{\cDNSresolved,  a, n}
\end{align*}
where $a\in \addresses$ is the IP address of the queried
domain name and $n$ is the nonce from the query.

The precise message format of DNS messages is provided in
Appendix~\ref{sec:dns-messages}, the full formal
specification for DNS servers can be found in
Appendix~\ref{app:dns-servers}. 

\subsection{HTTP Messages}\label{sec:http-messages}

In order to model web browsers and servers, we first need
to model HTTP requests and responses. The formal
specification of HTTP messages can be found in
Appendix~\ref{sec:http-messages-full}. Here we provide a
more informal presentation.

HTTP requests and responses are encoded as messages (ground
terms).  An HTTP request (modeled as a message) contains a
nonce, a method (e.g., $\mGet$ or $\mPost$), a domain name,
a path, URL parameters, request headers (such as
$\str{Cookie}$ or $\str{Origin}$), and a message body. For
example, an HTTP $\mGet$ request for the URL
\url{http://example.com/show?page=1} is modeled as the
term 
\begin{align*}
  \mi{r} := \hreq{ nonce=\linebreak[2]n_1, method=\mGet,
    host=example.com, path=/show,
    parameters=\an{\an{\str{page},1}}, headers=\an{},
    body=\an{}}
\end{align*}
where body and headers are empty.  A web server that
responds to this request is supposed to include the nonce
$n_1$ contained in $r$ in the response so that the browser
can match the request to the corresponding response.  More
specifically, an HTTP response (modeled as a message)
contains a nonce (matching the request), a status code
(e.g., $\str{200}$ for a normal successful response),
response headers (such as $\str{Set{\mhyphen}Cookie}$ and
$\str{Location}$), and a body. For example, a response to
$r$ could be 
\begin{align*}
  \mi{s} := \hresp{ nonce=n_1, status=200,
    headers=\an{\an{\str{Set{\mhyphen}Cookie},
        \an{\str{SID},\an{n_2,\bot,\True,\bot}}}},xbody=\an{\str{script1},n_3}}
\end{align*}
where $s$ contains (1) in the headers section, a cookie with
the name $\str{SID}$, the value $n_2$, and the attributes
$\str{secure}$ and $\str{httpOnly}$ not set but the
attribute $\str{session}$ set (see
Section~\ref{sec:web-browsers} for details on cookies) and
(2) in the body section, the string representation
$\str{script1}$ of the scripting process
$\mathsf{script}^{-1}(\str{script1})$ (which should be an
element of $\scriptset$) and its initial state $n_3$.

For the HTTP request and response in the above examples,
the corresponding HTTPS request would be of the form
$\ehreqWithVariable{r}{k'}{\pub(k_\text{example.com})}$ and
the response of the form $\ehrespWithVariable{s}{k'}$ where
$k'$ is a fresh symmetric key (a nonce) which is typically
generated by the sender of the request. The responder is
supposed to use this key to encrypt the response.

\subsection{Web Browsers}\label{sec:web-browsers}

We think of an honest browser to be used by one honest
user. However, we also allow browsers to be taken over by
attackers. The honest user is modeled as part of the web
browser model. Actions a user takes are modeled as
non-deterministic actions of the web browser. For example,
the web browser itself can non-deterministically follow the
links provided by a web page. Secrets, such as passwords,
typically provided by the user are stored in the initial
state of a browser and are given to a web page when needed,
similar to the AutoFill function in browsers (see below).

A web browser $p$ is modeled as a DY process $(I^p, Z^p,
R^p, s^p_0, N^p)$ where $I^p\subseteq \addresses$ is a
finite set and $N^p\subseteq \nonces$ is an infinite
set. The set of states $Z^p$, the initial state $s^p_0$,
and the relation $R^p$ are defined below
(Sections~\ref{sec:browserstate}
and~\ref{sec:browserrelation}), with a full formal
specification provided in
Appendix~\ref{sec:deta-descr-brows}.

\subsubsection{Browser State: $Z^p$ and
  $s^p_0$}\label{sec:browserstate}
The set $Z^p$ of states of a browser 
consists of terms of the form
\begin{align*}
\an{
  \mi{windows}, 
  \mi{secrets}, 
  \mi{cookies}, 
  \mi{localStorage}, 
  \mi{sessionStorage},
  \mi{keyMapping}, \\
  \mi{sts},  
  \mi{DNSaddress},
  \mi{nonces}, 
  \mi{pendingDNS},\allowbreak
  \mi{pendingRequests}, \allowbreak
  \mi{isCorrupted}
}
\end{align*}
\myparagraph{Windows and documents.} The most important
part of the state are windows and documents, both stored in
the subterm $\mathit{windows}$.  A browser may have a
number of windows open at any time (resembling the tabs in
a real browser). Each window contains a list of documents
of which one is ``active''. Being active means that this
document is currently presented to the user and is
available for interaction, similarly to the definition of
active documents in the HTML5 specification~\cite{html5}. The document list of a window represents the
history of visited web pages in that window. A window may
be navigated forward and backward (modeling forward and
back buttons). This deactivates one document and activates
its successor or predecessor.

A document is specified by a term which essentially
contains (the string representing) a script, the current
state of the script, the input that the script obtained so
far (from \xhrs and \pms), the origin (domain name plus
HTTP or HTTPS, see Appendix~\ref{sec:origins} for details)
of the document, and a list of windows (called
\emph{subwindows}), which correspond to iframes embedded in
the document, resulting in a tree of windows and
documents.  The (single) script is
meant to model the static HTML code, including, for
example, links and forms, and possibly multiple JavaScript
code parts. When called by the browser, a script
essentially outputs a command which is then interpreted by
the browser, such as following a link, creating an iframe,
or issuing an \xhr. In particular, a script can represent a
plain HTML document consisting merely of links, say: when
called by the browser such a script would
non-deterministically choose such a link and output it to
the browser, which would then load the corresponding web
page (see below for details).

We use the terms \emph{top-level window} (a window which is
not a subwindow itself), \emph{parent window} (the window
of which the current window is a direct subwindow) and
\emph{ancestor window} (some window of which the current
window is a not necessarily direct subwindow) to describe
the relationships in a tree of windows and documents.

A term describing a window or a document also contains a
unique nonce, which we refer to by \emph{reference}. This
reference is used to match HTTP responses to the
corresponding windows and documents from which they
originate (see below). 

Top-level windows may have been opened by
another window. In this case, the term of the opened window
contains a reference to the window by which it was
opened (the \emph{opener}). Following the HTML5 standard, we call such a window
an \emph{auxiliary window}. Note that auxiliary windows are
always top-level windows.

We call a window \emph{active} if it is a top-level window
or if it is a subwindow of an active document in an active
window. Note that the active documents in all active
windows are exactly those documents a user can currently
see/interact with in the browser.

\begin{example}\label{ex:window}
  The following is an example of a window term with reference
  $n_1$, two documents, and an opener ($n_4$):
  \vspace{-1.3ex}
  \begin{align*}
    \an{n_1, &\an{\an{n_2,\! \an{\str{example.com}, \http},
          \str{script1}, \an{}, \an{}, \an{}, \bot},
        \\ &\,\,\an{n_3,\! \an{\str{example.com}, \https},
          \str{script2}, \an{}, \an{}, \an{}, \True}}, n_4}
  \end{align*}\vspace{-3.6ex}\\
  The first document has reference $n_2$. It was loaded
  from the origin $\an{\str{example.com}, \http}$, which
  translates into \url{http://example.com}. Its scripting
  process has the string representation $\str{script1}$,
  the last state and the input history of this process are
  empty. The document does not have subwindows and is
  inactive ($\bot$). The second document has the reference
  $n_3$, its origin corresponds to
  \url{https://example.com}, the scripting process is
  represented by $\str{script2}$, and the document is
  active ($\top$).  All other components are empty.
\end{example}

\myparagraph{Secrets.} This subterm of the state term of a
browser holds the secrets of the user of the web
browser. Secrets (such as passwords) are modeled as nonces
and they are indexed by origins. Secrets are only released
to documents (scripts) with the corresponding origin,
similarly to the AutoFill mechanism in browsers.

\myparagraph{Cookies, localStorage, and sessionStorage.}
These subterms contain the cookies (indexed by domains),
\ls data (indexed by origins), and sessionStorage data
(indexed by origins and top-level window references) stored
in the browser. Cookies are stored together with their
$\str{secure}$, $\str{httpOnly}$, and $\str{session}$
attributes: If $\str{secure}$ is set, the cookie is only
delivered to HTTPS origins. If $\str{httpOnly}$ is set, the
cookie cannot be accessed by JavaScript (the
script). According to the proposed standard RFC6265 (which
we follow in our model) and the majority of the existing
implementations, cookies that neither have the (real)
``max-age'' nor the ``expires'' attribute should be deleted
by the browser when the session ends (usually when the
browser is closed). In our model, such cookies carry the
$\str{session}$ attribute.

\myparagraph{KeyMapping.} This term is our equivalent to a
certificate authority (CA) certificate store in the
browser. Since, for simplicity, we currently do not
formalize CAs in the model, this term simply encodes a
mapping assigning domains $d\in \dns$ to their respective
public keys $\pub(k_d)$.

\myparagraph{STS.} Domains that are listed in this term are
contacted by the web browser only over HTTPS\@. Connection
attempts over HTTP are transparently rewritten to HTTPS
requests. Web sites can issue the
$\str{Strict\mhyphen{}Transport\mhyphen{}Security}$ header
to clients in order to add their domain to this list, see
below.

\myparagraph{DNSaddress.} This term contains the address
of the DY process that is to be contacted for DNS requests;
typically a DNS server.

\myparagraph{Nonces, pendingDNS, and pendingRequests.}
These terms are used for bookkeeping purposes, recording
the nonces that have been used by the browser so far, the
HTTP requests that await successful DNS resolution, and
HTTP requests that await a response, respectively.

\myparagraph{IsCorrupted.} This term indicates
whether the browser is corrupted ($\not=\bot$) or not
($=\bot$). A corrupted browser behaves like a web attacker
(see Section~\ref{sec:browserrelation}).

\myparagraph{Initial state $s_0^p$ of a web browser.}  In
the initial state, $\mi{keyMapping}$, $\mi{DNSAddress}$, and
$\mi{secrets}$ are defined as needed, $\mi{isCorrupted}$ is
set to $\bot$, and all other subterms are $\an{}$.

\subsubsection{Web Browser Relation
  $R^p$}\label{sec:browserrelation}
Before we define the relation $R^p$, we first sketch the
processing of HTTP(S) requests and responses by a web
browser, and also provide some intuition about the
corruption of browsers.

\myparagraph{HTTP(S) Requests and Responses.}  An HTTP
request, contains, as mentioned before, a nonce created by
the browser. In the example in
Section~\ref{sec:http-messages}, this nonce is $n_1$.  A
server is supposed to include this nonce into its HTTP
response. By this, the browser can match the response to
the request (a real web browser would use the TCP sequence
number for this purpose). If a browser wants to send an
HTTP request, it first resolves the domain name to an IP
address. (For simplicity, we do not model DNS response
caching.) It therefore first records the HTTP request in
$\mi{pendingDNS}$ along with the reference of the window
(in the case of HTTP(S) requests) or the reference of the
document\footnote{As we will see later, in the case of
  \xhrs this reference is actually a sequence of two
  elements, a document reference and a nonce that was
  chosen by the script that issued the \xhr. For now, we
  will refer to this sequence simply as the document
  reference.} (in the case of \xhrs) from which the request
originated and then sends a DNS request. Upon receipt of
the corresponding DNS response it sends the HTTP request
and stores it (again along with the reference as well as
the server address) in $\mi{pendingRequests}$. Before
sending the HTTP request, the cookies stored in the browser
for the domain of the request are added as cookie headers
to the request. Cookies with attribute $\str{secure}$ are
only added for HTTPS requests. If an HTTP response arrives,
the browser uses the nonce in this response to match it
with the recorded corresponding HTTP request (if any) and
checks whether the address of the sender is as
expected. The reference recorded along with the request
then determines to which window/document the response
belongs. The further processing of a response is described
below.

We note that before HTTPS requests are sent out, a fresh
symmetric key (a nonce) is generated and added to the
request by the browser. The resulting message is then
encrypted using the public key corresponding to the domain
in the request (according to $\mathit{keyMapping}$). The
symmetric key is recorded along with the request in
$\mi{pendingRequests}$. The response is, as mentioned,
supposed to be encrypted with this symmetric key (see also
Appendix~\ref{sec:http-messages-full} for details).

\begin{figure}[t!]
  \small{ \underline{\textsc{Processing Input Message
        \hlExp{$m$}}}
    \begin{itemize}[noitemsep,nolistsep,label=,leftmargin=0pt]
    \item \textbf{\hlExp{$m = \fullcorrupt$}:}
      $\mi{isCorrupted} := \fullcorrupt$
    \item \textbf{\hlExp{$m = \closecorrupt$}:}
      $\mi{isCorrupted} := \closecorrupt$
    \item \textbf{\hlExp{$m = \trigger$}:}
      non-deterministically choose $\mi{action}$ from
      $\{1,2\}$
      \begin{itemize}[noitemsep,nolistsep,label=,leftmargin=1em]
      \item \textbf{\hlExp{$\mi{action} = 1$}:} Call script
        of some active document. Outputs new state and command $cmd$.
        \begin{itemize}[noitemsep,nolistsep,label=,leftmargin=1em]
        \item \textbf{\hlExp{$\mi{cmd} = \tHref$}:}
          $\rightarrow$ \emph{Initiate request}
        \item \textbf{\hlExp{$\mi{cmd} = \tIframe$}:}
          Create subwindow, $\rightarrow$ \emph{Initiate
            request}
        \item \textbf{\hlExp{$\mi{cmd} = \tForm$}:}
          $\rightarrow$ \emph{Initiate request}
        \item \textbf{\hlExp{$\mi{cmd} = \tSetScript$}:}
          Change script in given document.
        \item \textbf{\hlExp{$\mi{cmd} =
              \tSetScriptState$}:} Change state of script in given document.
        \item \textbf{\hlExp{$\mi{cmd} =
              \tXMLHTTPRequest$}:} $\rightarrow$
          \emph{Initiate request}
        \item \textbf{\hlExp{$\mi{cmd} = \tBack$ or
              $\tForward$}:} Navigate given window.
        \item \textbf{\hlExp{$\mi{cmd} = \tClose$}:} Close
          given window.
        \item \textbf{\hlExp{$\mi{cmd} = \tPostMessage$}:}
          Send \pm to specified document.
        \end{itemize}
      \item \textbf{\hlExp{$\mi{action} = 2$}:}
        $\rightarrow$
        \emph{Initiate request to some URL in new  window}
      \end{itemize}
    \item \textbf{\hlExp{$m=$} DNS response:} send
      corresponding HTTP request
    \item \textbf{\hlExp{$m=$} HTTP(S) response:}
      (decrypt,) find reference.
      \begin{itemize}[noitemsep,nolistsep,label=,leftmargin=1em]
      \item \textbf{reference to window:} create document
        in window
      \item \textbf{reference to
          document:} add response
          body to document's script input
      \end{itemize}
    \end{itemize}
  }
  \caption{The basic structure of the web browser relation
    $R^p$ with an extract of the most important processing
    steps, in the case that $\mi{isCorrupted}
    =\bot$.}\label{fig:browser-structure}
\end{figure}

\myparagraph{Corruption of Browsers.} We model two types of
corruption of browsers, namely \emph{full corruption} and
\emph{close-corruption}, which are triggered by special
network messages in our model. In the real world, an
attacker can exploit buffer overflows in web browsers,
compromise operating systems (e.g., using trojan horses),
and physically take control over shared terminals.

Full corruption models an attacker that gained full control
over a web browser and its user.  Besides modeling a
compromised system, full corruption can also serve as a
vehicle for the attacker to participate in a protocol using
secrets of honest browsers: In our case study
(Section~\ref{sec:analysisbrowserid}), the attacker starts
with no user secrets in its knowledge, but may fully
corrupt any number of browsers, so, in particular, he is
able to impersonate browsers/users.

Close-corruption models a browser that is taken over by the
attacker after a user finished her browsing session, i.e.,
after closing all windows of the browser.  This form of
corruption is relevant in situations where one browser can
be used by many people, e.g., in an Internet caf\'{e}.
Information left in the browser state after closing the
browser could be misused by malicious users.

\myparagraph{The Relation $R^p$.} To define $R^p$, we need
to specify, given the current state of the browser and an
input message $m$, the new state of the browser and the set
of events output by the
browser. Figure~\ref{fig:browser-structure} provides an
overview of the structure of the following definition of
$R^p$.  The input message $m$ is expected to be
$\fullcorrupt$, $\closecorrupt$, $\trigger$, a DNS
response, or an HTTP(S) response.

If $\mi{isCorrupted} \neq \bot$ (browser is corrupted), the
browser, just like an attacker process, simply adds $m$ to
its current state, and then outputs all events it can
derive from its state. Once corrupted, the browser stays
corrupted. Otherwise, if $\mi{isCorrupted} = \bot$, on
input $m$ the browser behaves as follows.

\smallskip\noindent \textbf{\hlExp{$m = \fullcorrupt$}:} If
the browser receives this message, it sets
$\mi{isCorrupted}$ to $\fullcorrupt$. From then on the
browser is corrupted as described above, with the attacker
having full access to the browser's internal state,
including all secrets. 

\smallskip\noindent \textbf{\hlExp{$m = \closecorrupt$}:}
If the browser receives this message, it first removes the
user secrets, open windows and documents, all
\emph{session} cookies, all sessionStorage data, and all
pending requests from its current state; nonces used so far
by the browser may not be used any longer. LocalStorage
data and persistent cookies are not deleted. The browser
then sets $\mi{isCorrupted}$ to $\closecorrupt$ (and hence,
from then on is corrupted). As already mentioned, this
models that the browser is closed by a user and that then
the browser is used by another, potentially malicious user
(an attacker), such as in an Internet caf\'{e}.

\smallskip\noindent \textbf{\hlExp{$m = \trigger$}:} Upon
receipt of this message, the browser non-deterministically
chooses one of two \emph{actions}: ($1$) trigger a script
or ($2$) request a new document.

\noindent \textbf{\hlExp{$m = \trigger$},
  \hlExp{$\mi{action} = 1$}:} Some active window (possibly
an \iframe) is chosen non-determin\-istically. Then the
script of the active document of that window is triggered
(see below).

\noindent \textbf{\hlExp{$m = \trigger$},
  \hlExp{$\mi{action} = 2$}:} A new HTTP(S) $\mGet$ request
(i.e., an HTTP(S) request with method $\mGet$) is created
where the URL is some message derivable from the current
state of the browser. However, nonces may not be used.
This models the user typing in a URL herself, but we do not
allow her to type in secrets, e.g., passwords or session
tokens.  A new window is created to show the
response. (HTTP requests to domains listed in $\mi{sts}$
are automatically rewritten to HTTPS requests).

\smallskip\noindent \textbf{\hlExp{$m=$} DNS response:}
DNS responses are processed as already described above,
resulting in sending the corresponding HTTP(S) request (if
any).

\smallskip\noindent \textbf{\hlExp{$m=$} HTTP(S) response:}
The browser performs the steps (I) to (IV) in this order. 

\myparagraph{\ (I)} The browser
identifies the corresponding HTTP(S) request (if
any), say $q$, and the window or document from which $q$
originated. (In case of HTTPS, the browser also decrypts
$m$ using the recorded symmetric key.)

\myparagraph{\ (II)} If there is a
$\str{Set{\mhyphen}Cookie}$ header in the response, its
content (name, value, and if present, the attributes
$\str{httpOnly}$, $\str{secure}$, $\str{session}$) is
evaluated: The cookie's name, value, and attributes are
saved in the browser's list of cookies. If a cookie with
the same name already exists, the old values and attributes
are overwritten, as specified in RFC6265.

\myparagraph{\ (III)} If there is a
$\str{Strict\mhyphen{}Transport\mhyphen{}Security}$ header
in the response, the domain of $q$ is added to the term
$\mi{sts}$. As defined in RFC6797, all future requests to
this domain, if not already HTTPS requests, are
automatically altered to use HTTPS.

\myparagraph{\ (IV)} If there is a $\str{Location}$ header
(with some URL $u$) in the response and the HTTP status
code is 303 or 307, the browser performs a redirection
(unless it is a non-same-origin redirect of an \xhr) by
issuing a new HTTP request to $u$, retaining the body of
the original request. Rewriting POST to GET requests for
303 redirects and extending the origin header value are
handled as defined in RFC2616 and in the W3C Cross-Origin
Resource Sharing specification\cite{w3c/cors}.

Otherwise, if no redirection is requested, the browser does
the following: If the request originated from a window, a
new document is created from the response body. For this,
the response body is expected to be a term of the form
$\an{\mi{sp}, \mi{stat}}$ where $\mi{sp}$ is a string such
that $\mathsf{script}^{-1}(\mi{sp})\in \scriptset$ is a
script and $\mi{stat}$ is a term used as its initial state.
The document is then added to the window the reference
points to, it becomes the active document, and the
successor of the currently active document. All previously
existing successors are removed. If the request originated
from a document (and hence, was the result of an \xhr{}),
the body of the response is appended to the script input
term of the document. When later the script of this
document is activated, it can read and process the
response.

\myparagraph{Triggering the Script of a Document
  (\textbf{\hlExp{$m = \trigger$}, \hlExp{$\mi{action} =
      1$}}).} First, the script of the document is called
with the following input:
\begin{itemize}
\item document and window references of all active
  documents and subwindows\footnote{Note that we
    overapproximate here: In real-world browsers, only a
    limited set of window handles are available to a
    script. Our approach is motivated by the fact that in some
    cases windows can be navigated by names (without a
  handle). However, as we will see, specific restrictions
  for navigating windows and accessing/changing their data
  apply.}
\item only for same-origin documents: information about the
  documents' origins, scripts, script states and script
  inputs
\item the last state and the input history (i.e., previous
  inputs from \pms and \xhrs) of the script as recorded in
  the document
\item cookies (names and values only) indexed with the
  document's domain, except for $\str{httpOnly}$ cookies
\item \ls data for the document's origin
\item sessionStorage data that is indexed with the
  document's origin and the reference of the document's
  top-level window
\item secrets indexed with the document's origin.
\end{itemize}
In addition, the script is given an infinite set of fresh
nonces from the browser's set of (unused) nonces.

Now, according to the definition of scripts, the script
outputs a term. The browser expects terms of the form
$\an{\mi{state}, \mi{cookies}, \mi{localStorage},
  \mi{sessionStorage}, \mi{cmd}}$ (and otherwise ignores
the output) where $\mi{state}$ is an arbitrary term
describing the new state of the script, $\mi{cookies}$ is a
sequence of name/value pairs, $\mi{localStorage}$ and
$\mi{sessionStorage}$ are arbitrary terms, and $\mi{cmd}$
is a term which is interpreted as a command which is to be
processed by the browser. The old state of the script
recorded in the document is replaced by the new one
($\mi{state}$), the local/session storage data recorded in
the browser for the document's origin (and top-level window
reference) is replaced by
$\mi{localStorage}$/$\mi{sessionStorage}$, and the old
cookie store of the document's origin is updated using
$\mi{cookies}$ similarly to the case of HTTP(S) responses
with cookie headers, except that now no httpOnly cookies
can be set or replaced, as defined by the HTML5 standard
\cite{html5} in combination with RFC6265. For details, see
Line~\ref{line:cookiemerge} of
Algorithm~\ref{alg:runscript} and
Definition~\ref{def:cookiemerge} in
Appendix~\ref{sec:deta-descr-brows}.

Subsequently, $\mi{cmd}$ (if not empty) is interpreted by
the browser, as described next. We note that commands may
contain parameters.

\smallskip

\noindent
\textbf{\hlExp{$\mi{cmd} = \tHref$}} (parameters: URL $u$,
window reference $w$): A new $\mGet$ request to $u$ is
initiated. If $w$ is $\wBlank$, the response to the request
will be shown in a new \emph{auxiliary} window.  This new
window will carry the reference to its opener, namely the
reference to the window in which the script was
running. Otherwise, if $w$ is not $\wBlank$, the window
with reference $w$ is navigated (upon receipt of the
response and only if it is active) to the given
URL\@. Navigation is subject to several
restrictions.\footnote{We follow the rules defined in
  \cite{html5}, Subsection~5.1.4: A window $A$ can navigate
  a window $B$ if the active documents of both are same
  origin, or $B$ is an ancestor window of $A$ and $B$ is a
  top-level window, or if there is an ancestor window of
  $B$ whose active document has the same origin as the
  active document of $A$ (including $A$
  itself). Additionally, $A$ may navigate $B$ if $B$ is an
  auxiliary window and $A$ is allowed to navigate the
  opener of $B$. We follow these rules closely, as can be
  seen in Algorithm~\ref{alg:getnavigablewindow} in
  Appendix~\ref{sec:deta-descr-brows}.}

\noindent \textbf{\hlExp{$\mi{cmd} = \tIframe$}}
(parameter: URL $u$, window reference $w$): Provided that
the active document of $w$ is same origin, create a new
subwindow in $w$ (with a new window reference) and initiate
an HTTP $\mGet$ request to $u$ for that
subwindow. 

\noindent \textbf{\hlExp{$\mi{cmd} = \tForm$}} (parameters:
URL $u$, method $\mi{md}$, data $d$, window reference $w$):
Initiate a new request for $u$ with method $\mi{md}$, and
body data $d$, where, just like in the case of $\tHref$,
$w$ determines in what window the response is shown.  Again
the same restrictions for navigating other windows as in
the case of $\tHref$ apply.  For this request an
$\str{Origin}$ header is set if $\mi{md}=\mPost$. Its value
is the origin of the document.

\noindent \textbf{\hlExp{$\mi{cmd} = \tSetScript$}}
(parameters: window reference $w$, string $s$): Change the
string representation of the script of the active document
in $w$ to $s\in \mathsf{script}^{-1}(\scriptset)$, provided
that the active document in $w$ is same origin.

\noindent \textbf{\hlExp{$\mi{cmd} = \tSetScriptState$}}
(parameters: window reference $w$, term $t$): Change the
state of the script of the active document in $w$ to $t$,
provided that $w$ is same origin.

\noindent \textbf{\hlExp{$\mi{cmd} = \tXMLHTTPRequest$}}
(parameters: URL $u$, method $\mi{md}$, data $d$, nonce
$\mi{xhrreference}$): Initiate a request with method
$\mi{md}$ and data $d$ for $u$, provided that $u$ is same
origin and $\mi{md}$ is not $\mConnect$, $\mTrace$, or
$\mTrack$ (according to \cite{xmlhttprequest}).  Instead of
the window reference, the reference used for this request
is $\an{r, \mi{xhrreference}}$, where $r$ is the reference
of the script's document and $\mi{xhrreference}$ is a nonce
chosen by the script (for later correlation). This combined
reference indicates that this request originated from an
\xhr and $\mi{xhrreference}$ is used by the script to
assign the response to the request. The $\str{Origin}$
header is set as in the case of $\tForm$.

\noindent \textbf{\hlExp{$\mi{cmd} = \tBack$ or
    $\tForward$}} (parameter: window reference $w$):
replace the active document in $w$ by its
predecessor/successor in $w$. Again, the same restrictions
for navigating windows as in the case of $\tHref$ apply.

\noindent \textbf{\hlExp{$\mi{cmd} = \tClose$}} (parameter:
window reference $w$): close the given window, i.e., remove
it from the list of windows in which it is contained. The
same restrictions for navigating windows as in the case of
$\tHref$ apply.

\noindent \textbf{\hlExp{$\mi{cmd} = \tPostMessage$}}
(parameter: message $\mi{msg}$, window reference $w$,
origin $o$): $\mi{msg}$, the origin of the sending
document, and a reference to its window are appended to the
input history of the active document in $w$ if that
document's origin matches $o$ or if $o=\bot$.

\subsection{Web Servers}

While the modeling of DNS servers and browsers is
independent of specific web applications, and hence, forms
the core of the model of the web infrastructure, the
modeling of a \emph{web server} heavily depends on the
specific web application under consideration. Conversely,
the model of a specific web application is determined by
the model of the web server.  We therefore do not and
cannot fix a model for web servers at this point.  Such a
model should be provided as part of the analysis of a
specific web application, as illustrated by our case study
(see Section~\ref{sec:browserid} and following).

\subsection{Limitations}\label{sec:limitations}

We now briefly discuss main limitations of the model. As
will be illustrated by our case study, our model is
formulated on a level of abstraction that is suitable to
capture many security relevant features of the web, and
hence, a relevant class of attacks. However, as with all
models, certain attacks are out of the scope of our
model. For example, as already mentioned, we currently
cannot reason about language details (e.g., how two
JavaScripts running in the same document interact). Also,
we currently do not model user interface details, such as
frames that may overlap in Clickjacking attacks. Being a
Dolev-Yao-style model, our model clearly does not aim at
lower-level cryptographic attacks. Also, byte-level
attacks, such as buffer overflows, are out of scope.

\section{The BrowserID System}\label{sec:browserid}

BrowserID \cite{mozilla/persona/mdn} is a new decentralized
single sign-on (SSO) system developed by Mozilla for user
authentication on web sites. It is a complex full-fledged
web application deployed in practice, with currently
{\raise.17ex\hbox{$\scriptstyle\mathtt{\sim}$}}47k LOC
(excluding code for Sideshow/BigTent, see below, and some
libraries).  It allows web sites to delegate user
authentication to email providers, where users use their
email addresses as identities. The BrowserID implementation
makes use of a broad variety of browser features, such as
\xhrs, \pm, local- and sessionStorage, cookies, etc.

We first, in Section~\ref{sec:browseridoverview}, provide a
high-level overview of the BrowserID system. A more
detailed description of the BrowserID implementation is
then given in Sections~\ref{sec:javascript-descr} to
\ref{sec:ident-prov-fallb} with further details provided in
Appendices~\ref{app:browserid-lowlevel}
to~\ref{app:browserid-sidp-lowlevel}.

\subsection{Overview}\label{sec:browseridoverview}

The BrowserID system knows three distinct parties: the
user, which wants to authenticate herself using a browser,
the relying party (RP) to which the user wants to
authenticate (log in) with one of her email addresses (say,
\nolinkurl{user@eyedee.me}), and the identity/email address
provider IdP\@. If the email provider (\nolinkurl{eyedee.me}) supports
BrowserID directly, it is called a \emph{primary
  IdP}. Otherwise, a Mozilla-provided service, a so-called
\emph{secondary IdP}, takes the role of the IdP\@. In what
follows, we describe the case of a primary IdP, with more
information on secondary IdPs given in
Section~\ref{sec:ident-prov-fallb}. 

A primary IdP provides information about its BrowserID
setup in a so-called \emph{support document}, which it
provides at a fixed URL derivable from the email domain,
e.g., \url{https://eyedee.me/.well-known/browserid}.

A user who wants to log in at an RP with an email address
for some IdP has to present two signed documents: A
\emph{user certificate} (UC) and an \emph{identity
  assertion} (IA). The UC contains the user's email address
and a public key. It is signed by the IdP\@.  The IA contains
the origin of the RP and is signed with the private key
corresponding to the user's public key. Both documents have
a limited validity period. A pair consisting of a UC and a
matching IA is called a \emph{certificate assertion pair}
(CAP) or a \emph{backed identity assertion}. Intuitively,
the UC in the CAP tells the RP that (the IdP certified
that) the owner of the email address is (or at least
claimed to be) the owner of the public key. By the IA
contained in the CAP, the RP is ensured that the owner of
the given public key wants to log in. Altogether, given a
valid CAP, RP would consider the user (with the email
address mentioned in the CAP) to be logged in.

\begin{figure}[bth!]\centering
  \scriptsize{
  \begin{tikzpicture}

    \matrix [row sep=5ex, column sep={6pc}]
    {
     \node[anchor=base,fill=Gainsboro,rounded corners](rp){RP}; & \node[anchor=base,fill=Gainsboro,rounded corners](b){Browser}; & \node[anchor=base,fill=Gainsboro,rounded corners](idp){IdP}; \\
     & \node(b-gen-key){}; & \\
     & \node(b-send-pubkey){}; & \node(idp-recv-pubkey){}; \\
     & & \node(idp-create-uc){}; \\
     & \node(b-recv-cert){}; & \node(idp-send-cert){}; \\
     \node(phase-1-2-left){}; & & \node(phase-1-2-right){}; \\ %
     & \node(b-gen-ia){}; & \\
     \node(rp-recv-cap){}; & \node(b-send-cap){}; & \\
     \node(phase-2-3-left){}; & & \node(phase-2-3-right){}; \\ %
     \node(rp-recv-pubkey){}; & & \node(idp-send-pubkey){}; \\
     \node(rp-vrfy-cap){}; & & \\
     \node(rp-end){}; & \node(b-end){}; & \node(idp-end){}; \\
    };

    \begin{pgfonlayer}{background}
    \draw [color=gray] (rp) -- (rp-end);
    \draw [color=gray] (b) -- (b-end);
    \draw [color=gray] (idp) -- (idp-end);
    \end{pgfonlayer}

    \node(phase-1-2-leftleft)[left of=phase-1-2-left]{};
    \node(phase-1-2-rightright)[right of=phase-1-2-right]{};
    \draw [dashed] (phase-1-2-leftleft) -- (phase-1-2-rightright);
    \node(phase-2-3-leftleft)[left of=phase-2-3-left]{};
    \node(phase-2-3-rightright)[right of=phase-2-3-right]{};
    \draw [dashed] (phase-2-3-leftleft) -- (phase-2-3-rightright);

    \node at (b-gen-key) [draw,fill=white,rounded corners] {\alphprotostep{gen-key-pair} gen.\ key pair};
    \draw [->] (b-send-pubkey) to node[draw,fill=white]{\alphprotostep{req-uc} $\text{pk}_\text{b}$, email} (idp-recv-pubkey);
    \node at (idp-create-uc) [draw,fill=white,rounded corners] {\alphprotostep{create-uc} create UC};
    \draw [->] (idp-send-cert) to node[draw,fill=white]{\alphprotostep{recv-uc} UC} (b-recv-cert);
    \node at (b-gen-ia) [draw,fill=white,rounded corners]{\alphprotostep{gen-ia} gen. IA};
    \draw [->] (b-send-cap) to node[draw,fill=white]{\alphprotostep{send-cap} CAP} (rp-recv-cap);
    \draw [->] (idp-send-pubkey) to node[draw,fill=white]{\alphprotostep{recv-idp-pubkey} $\text{pk}_{\text{IdP}}$} (rp-recv-pubkey);
    \node at (rp-vrfy-cap) [draw,fill=white,rounded corners]{\alphprotostep{verify-cap} verify CAP};

    \node(phase-1-label) at ($(phase-1-2-left)!0.5!(rp)$) [left=1ex] {\bigprotophase{provisioning}};
    \node(phase-2-label) at ($(phase-2-3-left)!0.5!(phase-1-2-left)$) [left=1ex] {\bigprotophase{authentication}};
    \node(phase-3-label) at ($(rp-end)!0.75!(phase-2-3-left)$) [left=1ex] {\bigprotophase{verification}};

  \end{tikzpicture}}
  \caption{BrowserID authentication: basic overview}
  \label{fig:browserid-highlevel}
\end{figure}
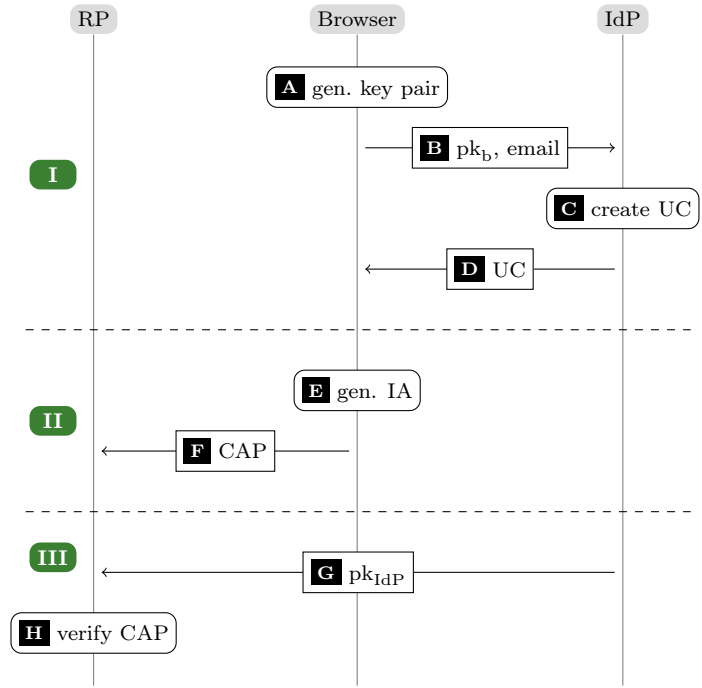

The BrowserID authentication process (with a primary IdP)
consists of three phases (see
Figure~\ref{fig:browserid-highlevel} for an overview):
\refbigprotophase{provisioning} provisioning of the UC,
\refbigprotophase{authentication} CAP creation, and
\refbigprotophase{verification} verification of the CAP.

In Phase \refbigprotophase{provisioning}, the (browser of
the) user creates a public/pri\-vate key
pair~\refalphprotostep{gen-key-pair}. She then sends her
public key as well as the email address she wants to use to
log in at some RP to IdP~\refalphprotostep{req-uc}. IdP now
creates the UC~\refalphprotostep{create-uc}, which is
then sent to the user~\refalphprotostep{recv-uc}. The above
requires the user to be logged in at IdP.

With the user having received the UC, Phase
\refbigprotophase{authentication} can start. The user wants
to authenticate to an RP, so she creates the
IA~\refalphprotostep{gen-ia}. The UC and the IA are
concatenated to a CAP, which is then sent to the
RP~\refalphprotostep{send-cap}.

In Phase \refbigprotophase{verification}, the RP checks the
authenticity of the CAP\@. For this purpose, the RP could use
an external verification service provided by Mozilla or
check the CAP itself as follows: First, the RP fetches the
public key of IdP \refalphprotostep{recv-idp-pubkey}, which
is contained in the support document. Afterwards, the RP
checks the signatures of the UC and the
IA~\refalphprotostep{verify-cap}. If this check is
successful, the RP can, as mentioned before, consider the
user to be logged in with the given email address and send
her some token (e.g., a session ID), which we refer to as
an \emph{RP service token}.

\subsection{Implementation Details}\label{sec:javascript-descr}
\begin{figure}[bth!]\centering

\begin{tikzpicture}
\scriptsize

\matrix[column sep={6pc,between origins}, row sep=1.5ex] {
\node[anchor=base,fill=Gainsboro,rounded corners](lpo){LPO}; & \node[anchor=base,fill=Gainsboro,rounded corners](idp){IdP}; & \node[draw,anchor=base](rpdoc){RP-Doc}; & \node(ld-top){};                 & \node(pif-top){}; \\
                                   &                                    & \node(rpdoc-ld-open){};                 & \node(ld)[draw,anchor=base]{LD}; & \\
\node(phase-ld-init-1-left){};     &                                    &                                         &                                  & \node(phase-ld-init-1-right){}; \\ [2pt]
\node(lpo-ld-init-1){};            &                                    &                                         & \node(ld-init-1){};              & \\ [1pt]
                                   &                                    & \node(rpdoc-ld-recv-ready){};           & \node(ld-rpdoc-send-ready){};    & \\
                                   &                                    & \node(rpdoc-ld-send-request){};         & \node(ld-rpdoc-recv-request){};  & \\
\node(lpo-ld-ctx-1){};             &                                    &                                         & \node(ld-lpo-ctx-1){};           & \\
                                   &                                    &                                         & \node(ld-user-email){};          & \\
\node(lpo-ld-addr-info-1-top){};   &                                    &                                         & \node(ld-lpo-addr-info-1-top){}; & \\ [1pt]
\node(lpo-idp-wk-1){};             & \node(idp-lpo-wk-1){};             &                                         &                                  & \\
\node(lpo-ld-addr-info-1-btm){};   &                                    &                                         & \node(ld-lpo-addr-info-1-btm){}; & \\
\node(phase-prov-left){};          &                                    &                                         &                                  & \node(phase-prov-right){}; \\
                                   &                                    &                                         & \node(ld-pif-open){};            & \node(pif)[draw,anchor=base]{PIF}; \\
                                   & \node(idp-pif-init-1){};           &                                         &                                  & \node(pif-init-1){}; \\ [1pt]
                                   &                                    &                                         & \node(ld-pif-pms){};             & \node(pif-ld-pms){}; \\
                                   &                                    &                                         & \node(ld-pif-close){};           & \node(pif-end)[draw,anchor=base]{/PIF}; \\
\node(phase-auth-left){};          &                                    &                                         &                                  & \node(phase-auth-right){}; \\
                                   & \node(idp-ld-auth){};              &                                         & \node(ld-auth){}; & \\
\node(phase-ld-init-2-left){};     &                                    &                                         &                                  & \node(phase-ld-init-2-right){}; \\
\node(lpo-ld-init-2-top){};        & \node(idp-ld-init-2-top){};        & \node(rpdoc-ld-init-2-top){};           & \node(ld-ld-init-2-top){};       & \node(pif-ld-init-2-top){}; \\
\node(lpo-ld-init-2-btm){};        & \node(idp-ld-init-2-btm){};        & \node(rpdoc-ld-init-2-btm){};           & \node(ld-ld-init-2-btm){};       & \node(pif-ld-init-2-top){}; \\
\node(phase-prov-cont-left){};     &                                    &                                         &                                  & \node(phase-prov-cont-right){}; \\
\node(lpo-prov-cont-top){};        & \node(idp-prov-cont-top){};        & \node(rpdoc-prov-cont-top){};           & \node(ld-prov-cont-top){};       & \node(pif-prov-cont-top){}; \\
                                   &                                    &                                         & \node(ld-prov-cont-key-pair){};  & \\ [3ex]
                                   &                                    &                                         & \node(ld-prov-cont-pkb){};       & \node(pif-prov-cont-pkb){};\\
                                   & \node(idp-prov-cont-req-uc){};     &                                         &                                  & \node(pif-prov-cont-req-uc){};\\
                                   & \node(idp-prov-cont-cert-uc){};    &                                         &                                  & \\
                                   & \node(idp-prov-cont-send-uc){};    &                                         &                                  & \node(pif-prov-cont-send-uc){};\\
                                   &                                    &                                         & \node(ld-prov-cont-recv-uc){};   & \node(pif-prov-cont-recv-uc){};\\
\node(lpo-prov-cont-btm){};        & \node(idp-prov-cont-btm){};        & \node(rpdoc-prov-cont-btm){};           & \node(ld-prov-cont-btm){};       & \node(pif-prov-cont-btm){}; \\
\node(phase-auth-lpo-left){};      &                                    &                                         &                                  & \node(phase-auth-lpo-right){}; \\
                                   &                                    &                                         & \node(ld-gen-cap-lpo){};         & \\ [3ex]
\node(lpo-ld-auth){};              &                                    &                                         & \node(ld-lpo-auth){};            & \\ [-2pt]
\node(phase-cap-left){};           &                                    &                                         &                                  & \node(phase-cap-right){}; \\ [2pt]
\node(lpo-ld-list-emails){};       &                                    &                                         & \node(ld-lpo-list-emails){};     & \\ [2pt]
\node(lpo-ld-addr-info-3){};       &                                    &                                         & \node(ld-lpo-addr-info-3){};     & \\
                                   &                                    &                                         & \node(ld-gen-cap){};             & \\ [3ex]
                                   &                                    & \node(rpdoc-ld-cap){};                  & \node(ld-rpdoc-cap){};           & \\
                                   &                                    & \node(rpdoc-ld-close){};                & \node[draw,anchor=base](ld-end){/LD}; & \\
\node(lpo-end){};                  & \node(idp-end){};                  & \node(rpdoc-end){};                     &                                  & \node(pif-col-end){}; \\
};

\node(phase-1-label) at ($(phase-prov-left)!0.5!(phase-ld-init-1-left)$) [left=1ex] {\protophase{ld-start-1}};
\node(phase-2-label) at ($(phase-auth-left)!0.5!(phase-prov-left)$) [left=1ex] {\protophase{ld-prov-1}};
\node(phase-3-label) at ($(phase-ld-init-2-left)!0.5!(phase-auth-left)$) [left=1ex] {\protophase{ld-auth}};
\node(phase-4-label) at ($(phase-prov-cont-left)!0.5!(phase-ld-init-2-left)$) [left=1ex] {\protophase{ld-start-2}};
\node(phase-5-label) at ($(phase-auth-lpo-left)!0.5!(phase-prov-cont-left)$) [left=1ex] {\protophase{ld-prov-2}};
\node(phase-6-label) at ($(phase-cap-left)!0.5!(phase-auth-lpo-left)$) [left=1ex] {\protophase{ld-lpo-auth}};
\node(phase-7-label) at ($(lpo-end)!0.5!(phase-cap-left)$) [left=1ex] {\protophase{ld-cap}};

\tikzstyle{xhrArrow} = [color=blue,decoration={markings, mark=at position 1 with {\arrow[color=blue]{triangle 45}}}, preaction = {decorate}]

\draw [->,snake=snake,segment amplitude=0.2ex] (rpdoc-ld-open.40) to node [above=-2pt] {\protostep{ld-open} open} (ld);

\draw [->] (ld-init-1.160) to node [above=-2pt]{\protostep{ld-init-1} GET LD} (lpo-ld-init-1.20);
\draw [->] (lpo-ld-init-1.340) -- (ld-init-1.200);

\draw [->,color=red,dashed] (ld-rpdoc-send-ready) to node [above=-2pt]{\protostep{ld-rpdoc-ready-1} ready} (rpdoc-ld-recv-ready);

\draw [->,color=red,dashed] (rpdoc-ld-send-request) to node [above=-2pt]{\protostep{rpdoc-ld-request-1} request} (ld-rpdoc-recv-request);

\draw [->,color=blue,>=latex] (ld-lpo-ctx-1.160) to node [above=-2pt]{\protostep{ld-ctx-1} GET session\_context} (lpo-ld-ctx-1.20);
\draw [->,color=blue,>=latex] (lpo-ld-ctx-1.340) -- (ld-lpo-ctx-1.200);

\node (ld-user-email-drawn) at (ld-user-email) [draw,rounded corners,fill=Gainsboro]{\protostep{ld-user-email} email address };

\draw [->,color=blue,>=latex] (ld-lpo-addr-info-1-top) to node [above=-2pt]{\protostep{ld-addrinfo-1} GET address\_info} (lpo-ld-addr-info-1-top);

\draw [->] (lpo-idp-wk-1.20) to node [above=-2pt]{\protostep{lpo-idp-wk-1} GET wk} (idp-lpo-wk-1.160);
\draw [->] (idp-lpo-wk-1.200) -- (lpo-idp-wk-1.340);

\draw [->,color=blue,>=latex] (lpo-ld-addr-info-1-btm) to node [above=-2pt]{\protostep{ld-addrinfo-1-resp}} (ld-lpo-addr-info-1-btm);

\draw [->,snake=snake,segment amplitude=0.2ex] (ld-pif-open.40) to node [above=-2pt] {\protostep{ld-pif-open-1} create} (pif);

\draw [->] (pif-init-1.160) to node [above=-2pt]{\protostep{pif-init-1} GET PIF} (idp-pif-init-1.20);
\draw [->] (idp-pif-init-1.340) -- (pif-init-1.200);

\draw [implies-implies,color=red,dashed,double=Gainsboro] (pif-ld-pms) to node [above=-2pt]{\protostep{pif-ld-pms-1} PMs} (ld-pif-pms);

\draw [->,snake=snake,segment amplitude=0.2ex] (ld-pif-close.40) to node [above=-2pt] {\protostep{ld-pif-close-1} close} (pif-end);

\node (ld-auth-drawn) at (ld-auth) [draw]{\protostep{idp-ld-auth} auth IdP};
\draw [implies-implies,double] (ld-auth-drawn) -- (idp-ld-auth);

\node[right of=phase-4-label,node distance=5em] {repeat \resizebox{!}{0.8\baselineskip}{\refprotophase{ld-start-1}}};

\node[right of=phase-5-label,node distance=3em,rotate=90] {repeat \resizebox{!}{0.8\baselineskip}{\refprotophase{ld-prov-1}}};

\node at (pif-prov-cont-top) [draw,fill=Gainsboro] {PIF};

\node at (ld-prov-cont-key-pair) [draw,rounded corners,fill=Gainsboro]{\protostep{gen-key-pair} gen. key pair};

\draw [->,color=red,dashed] (ld-prov-cont-pkb) to node [above=-2pt]{\protostep{pubkey-ld-pif} $\text{pk}_\text{b}$, email} (pif-prov-cont-pkb);

\draw [->,color=blue,>=latex] (pif-prov-cont-req-uc) to node [above=-2pt]{\protostep{req-uc} $\text{pk}_\text{b}$, email} (idp-prov-cont-req-uc);

\node (certify-uc) at (idp-prov-cont-cert-uc) [draw,rounded corners,fill=white]{\protostep{certify-uc} create UC};

\draw [->,color=blue,>=latex] (idp-prov-cont-send-uc) to node [above=-2pt]{\protostep{send-uc} UC} (pif-prov-cont-send-uc);

\draw [->,color=red,dashed] (pif-prov-cont-recv-uc) to node [above=-2pt]{\protostep{recv-uc} UC} (ld-prov-cont-recv-uc);

\node at (pif-prov-cont-btm) [draw,fill=Gainsboro] {/PIF};

\node at (ld-gen-cap-lpo) [draw,rounded corners,fill=Gainsboro]{\protostep{ld-gen-cap-lpo} gen. $\text{IA}_\text{LPO}$};

\draw [->,color=blue,>=latex] (ld-lpo-auth.160) to node [above=-2pt]{\protostep{ld-lpo-auth} POST auth\_with\_assertion ($\text{CAP}_\text{LPO}$)} (lpo-ld-auth.20);
\draw [->,color=blue,>=latex] (lpo-ld-auth.340) -- (ld-lpo-auth.200);

\draw [->,color=blue,>=latex] (ld-lpo-list-emails.160) to node [above=-2pt]{\protostep{ld-lpo-list-emails} GET list\_emails} (lpo-ld-list-emails.20);
\draw [->,color=blue,>=latex] (lpo-ld-list-emails.340) -- (ld-lpo-list-emails.200);

\draw [->,color=blue,>=latex] (ld-lpo-addr-info-3.160) to node [above=-2pt]{\protostep{ld-addrinfo-3} GET address\_info} (lpo-ld-addr-info-3.20);
\draw [->,color=blue,>=latex] (lpo-ld-addr-info-3.340) -- (ld-lpo-addr-info-3.200);

\node at (ld-gen-cap) [draw,rounded corners,fill=Gainsboro]{\protostep{ld-gen-cap} gen. $\text{IA}_\text{RP}$};

\draw [->,color=red,dashed] (ld-rpdoc-cap) to node [above=-2pt]{\protostep{ld-rpdoc-cap} response ($\text{CAP}_\text{RP}$)} (rpdoc-ld-cap);

\draw [->,snake=snake,segment amplitude=0.2ex] (rpdoc-ld-close.40) to node [above=-2pt] {\protostep{ld-close} close} (ld-end);

\begin{pgfonlayer}{background}
 \node (rpdoc-a) [above of=rpdoc, node distance=2ex]{};
 \node (rpdoc-al) [left of=rpdoc-a, node distance=6ex]{};
 \node (pif-col-end-b) [below of=pif-col-end, node distance=2ex]{};
 \node (pif-col-end-br) [right of=pif-col-end-b, node distance=4ex]{};
 \filldraw [color=Gainsboro,rounded corners] (rpdoc-al) rectangle (pif-col-end-br);

 \draw [color=gray] (lpo.270) -- (lpo-ld-init-2-top);
 \draw [color=gray,dotted,thick] (lpo-ld-init-2-top) -- (lpo-ld-init-2-btm);
 \draw [color=gray] (lpo-ld-init-2-btm) -- (lpo-prov-cont-top);
 \draw [color=gray,dotted,thick] (lpo-prov-cont-top) -- (lpo-prov-cont-btm);
 \draw [color=gray] (lpo-prov-cont-btm) -- (lpo-end);
 \draw [color=gray] (idp.270) -- (idp-ld-init-2-top);
 \draw [color=gray,dotted,thick] (idp-ld-init-2-top) -- (idp-ld-init-2-btm);
 \draw [color=gray] (idp-ld-init-2-btm) -- (idp-prov-cont-top);
 \draw [color=gray,dotted,thick] (idp-prov-cont-top) -- (idp-prov-cont-btm);
 \draw [color=gray] (idp-prov-cont-btm) -- (idp-end);
 \draw [color=gray] (rpdoc.270) -- (rpdoc-ld-init-2-top);
 \draw [color=gray,dotted,thick] (rpdoc-ld-init-2-top) -- (rpdoc-ld-init-2-btm);
 \draw [color=gray] (rpdoc-ld-init-2-btm) -- (rpdoc-prov-cont-top);
 \draw [color=gray,dotted,thick] (rpdoc-prov-cont-top) -- (rpdoc-prov-cont-btm);
 \draw [color=gray] (rpdoc-prov-cont-btm) -- (rpdoc-end);
 \draw [color=gray] (ld.270) -- (ld-auth-drawn);
 \draw [color=gray] (ld-auth-drawn) -- (ld-ld-init-2-top);
 \draw [color=gray,dotted,thick] (ld-ld-init-2-top) -- (ld-ld-init-2-btm);
 \draw [color=gray] (ld-ld-init-2-btm) -- (ld-prov-cont-top);
 \draw [color=gray,dotted,thick] (ld-prov-cont-top) -- (ld-prov-cont-btm);
 \draw [color=gray] (ld-prov-cont-btm) -- (ld-end);
 \draw [color=gray] (pif.270) -- (pif-end);
 \draw [color=gray,dotted,thick] (pif-prov-cont-top) -- (pif-prov-cont-btm);
\end{pgfonlayer}

\draw [dashed] (phase-ld-init-1-left.180) -- (phase-ld-init-1-right.0);
\draw [dashed] (phase-prov-left.180) -- (phase-prov-right.0);
\draw [dashed] (phase-auth-left.180) -- (phase-auth-right.0);
\draw [dashed] (phase-ld-init-2-left.180) -- (phase-ld-init-2-right.0);
\draw [dashed] (phase-prov-cont-left.180) -- (phase-prov-cont-right.0);
\draw [dashed] (phase-auth-lpo-left.180) -- (phase-auth-lpo-right.0);
\draw [dashed] (phase-cap-left.180) -- (phase-cap-right.0);

\node at ($(ld-top)!0.7!(pif-top)$) {Browser};

\end{tikzpicture}

\caption{BrowserID implementation overview. Black arrows (open
  tips) denote HTTPS messages, blue arrows (filled tips) denote \xhrs
  (over HTTPS), red (dashed) arrows are \pms, snake lines are commands to the
  browser.}
\label{fig:browserid-lowlevel-ld}
\end{figure}
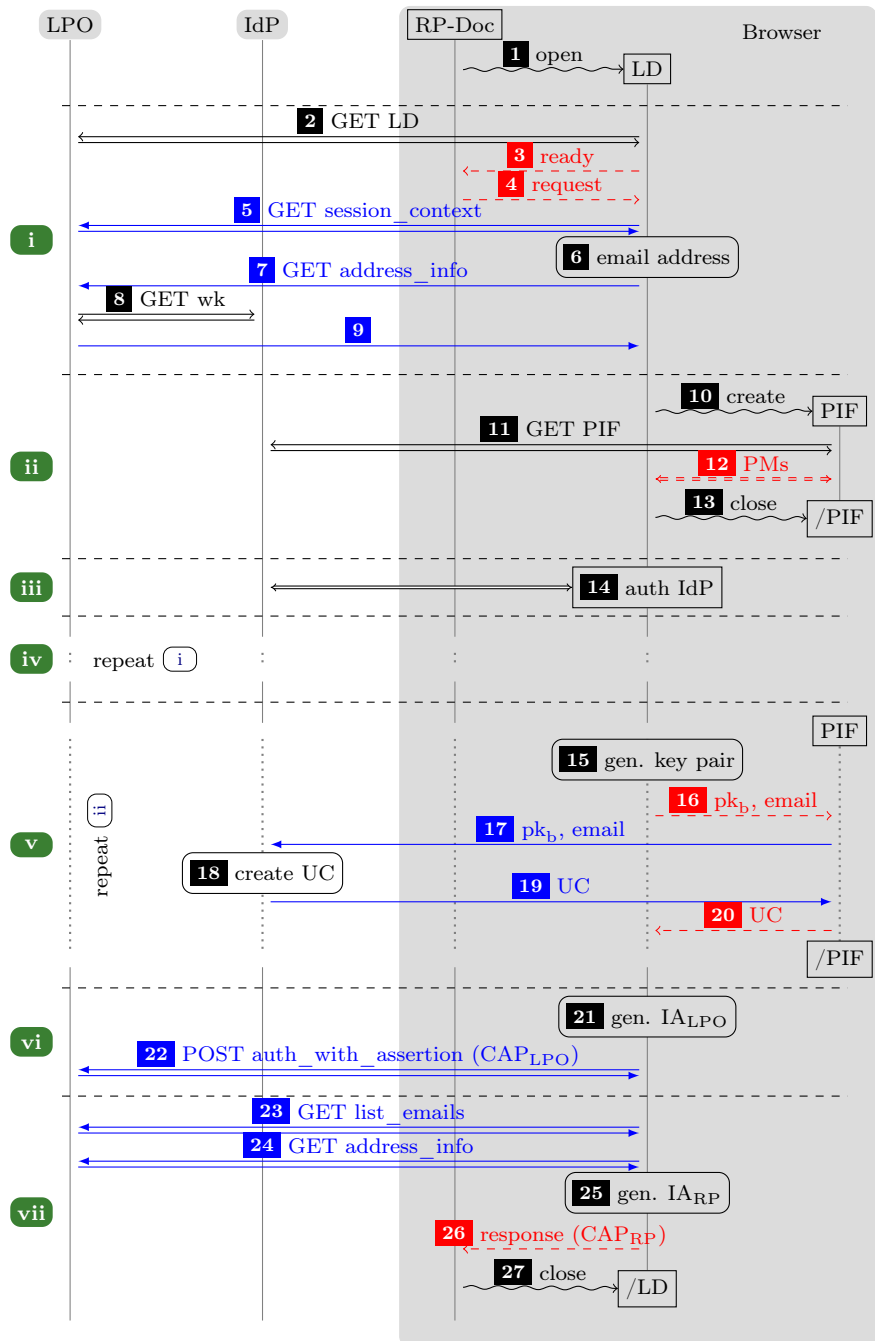

We now provide a more detailed description of the BrowserID
implementation (see also
Figure~\ref{fig:browserid-lowlevel-ld}). Since the system
is very complex, with many HTTPS requests, \xhrs, and
postMessages sent between different entities (servers as
well as windows and iframes within the browser), we here
describe mainly the phases of the login process without
explaining every single message exchange done in the
implementation. A more detailed step-by-step description
can be found in Appendix~\ref{app:browserid-lowlevel}.

In addition to the parties mentioned in the rough overview
in Section~\ref{sec:browseridoverview}, the actual
implementation uses another party,
\nolinkurl{login.persona.org} (LPO). The role of LPO is as
follows: First, LPO provides the HTML and JavaScript files
of the implementation. Thus, the BrowserID implementation
mainly runs under the origin of LPO.\footnote{It is
  envisioned by Mozilla to integrate the part of LPO
  directly into the browser in the future.} When the
JavaScript implementation running in the browser under the
origin of LPO needs to retrieve information from the IdP
(support document), LPO acts as a proxy to circumvent
cross-origin restrictions.

Before explaining the login process, we provide a quick
overview of the windows and iframes in the browser. By
RP-Doc we denote the window (see
Figure~\ref{fig:browserid-lowlevel-ld}) containing the
document loaded from some RP, a web page on which the user
wants to log in with an email address of some IdP\@. This
document typically includes JavaScript from LPO and
contains a button ``Login with BrowserID''. (Loading of
\emph{RP-Doc} from the RP and the JavaScript from LPO is
not depicted in
Figure~\ref{fig:browserid-lowlevel-ld}). The LPO JavaScript
running in \emph{RP-Doc} opens an auxiliary window called
the \emph{login dialog} (LD). Its content is provided by
LPO and it handles the interaction with the user. During
the login process, a temporary invisible iframe called the
\emph{provisioning iframe} (PIF) can be created in the
LD\@. The PIF is loaded from IdP\@. It is used by LD to
communicate (cross-origin) with IdP\@. Temporarily, the LD
may navigate itself to a web page at IdP to allow for
direct user interaction with the IdP.

Now, in order to describe the login process, for the time
being we assume that the user uses a ``fresh'' browser,
i.e., the user has not been logged in before. As mentioned,
the process starts by the user visiting a web site of some
RP\@.  After the user has clicked on the login button in
RP-Doc, the LD is opened and the interactive login flow is
started. We can divide this login flow into seven phases:
In Phase~\refprotophase{ld-start-1}, the LD is initialized
and the user is prompted to provide her email address. Then
LD fetches the support document (see
Section~\ref{sec:browseridoverview}) of IdP via LPO\@. In
Phase~\refprotophase{ld-prov-1}, LD creates the PIF from
the \emph{provisioning URL} provided in the support
document. As (by our assumption) the user is not logged in
yet, the PIF notifies LD that the user is not authenticated
to IdP yet. In Phase~\refprotophase{ld-auth}, LD navigates
itself away to the \emph{authentication URL} which is also
provided in the support document and links to IdP\@. Usually,
this document will show a login form in which the user
enters her password to authenticate to the IdP\@.  After the
user has been authenticated to IdP (which typically implies
that IdP sets a session cookie in the browser), the window
is navigated to LPO again. (This is done by JavaScript
loaded from LPO that the IdP document is supposed to
include.)

Now, the login flow continues in Phase
\refprotophase{ld-start-2}, which basically repeats Phase
\refprotophase{ld-start-1}. However, the user is not
prompted for her email address (it has previously been
saved in the localStorage under the origin of LPO along
with a nonce, where the nonce is stored in the
sessionStorage). In Phase~\refprotophase{ld-prov-2}, which
basically repeats Phase~\refprotophase{ld-prov-1}, the PIF
detects that the user is now authenticated to IdP and the
provisioning phase is started
(\refbigprotophase{provisioning} in
Figure~\ref{fig:browserid-highlevel}): The user's keys are
created by LD and stored in the localStorage under the
origin of LPO\@. The PIF forwards the certification request
to IdP, which then creates the UC and sends it back to the
PIF\@. The PIF in turn forwards it to the LD, which stores it
in the localStorage under the origin of LPO.

In Phases~\refprotophase{ld-lpo-auth} and
\refprotophase{ld-cap}, mainly the IA is generated by LD
for the origin of RP-Doc and sent (together with the UC) to
RP-Doc (\refbigprotophase{authentication} in
Figure~\ref{fig:browserid-highlevel}). In the localStorage,
LD stores that the user's email is logged in at
RP\@. Moreover, the user's email is recorded at LPO (see the
explanation on LPO Sessions below). For this purpose, LD
generates an IA for the origin of LPO and sends the UC and
IA to LPO.

\myparagraph{LPO Session.} LPO establishes a session with
the browser by setting a cookie \texttt{browserid\_state}
(in Step~\refprotostep{ld-ctx-1} in
Figure~\ref{fig:browserid-lowlevel-ld}) on the
client-side. LPO considers such a session authenticated
after having received a valid CAP
(in Step~\refprotostep{ld-lpo-auth} in
Figure~\ref{fig:browserid-lowlevel-ld}). In future
runs, the user is presented a list of her email addresses
(which is fetched from LPO) in order to choose one
address. Then, she is asked if she trusts the computer she
is using and is given the option to be logged in for one
month or ``for this session only'' (\emph{ephemeral}
session). In order to use any of the email addresses, the
user is required to authenticate to the IdP responsible for
that address to get an UC issued. If the localStorage
(under the origin LPO) already contains a valid UC, then,
however, authentication at the IdP is not necessary.

\myparagraph{Automatic CAP Creation.}  In addition to the
interactive login presented above, BrowserID also contains
an automatic, non-interactive way for RPs to obtain a
freshly generated CAP: During initialization of the
BrowserID code included by RP-Doc, an invisible iframe
called the \emph{communication iframe} (CIF) is created
inside RP-Doc. The CIF's JavaScript is loaded from LPO and
behaves similar to LD, but without user interaction. The
CIF automatically issues a fresh CAP and sends it to RP-Doc
under specific conditions: among others, the email address
must be marked as logged in at RP in the \ls. If necessary,
a new key pair is created and a corresponding new UC is
requested at IdP.

In Figure~\ref{fig:browserid-lowlevel} in
Appendix~\ref{app:browserid-sidp-lowlevel}, we show the
usage of the CIF in detail for the case of the \emph{secondary
  IdP}.

\myparagraph{Logout.} We have to differentiate between
three ways of logging out: an RP logout, an LPO logout, and
an IdP logout. An RP logout is handled by the CIF after it
has received a \emph{logout} \pm from RP-Doc. The CIF then
changes the \ls such that no email address is recorded to
be logged in at RP\@. 

An LPO logout essentially requires to logout at the web
site of LPO\@. The LPO logout removes all key pairs and
certificates from the \ls and invalidates the session on
the LPO server.

An IdP logout depends on the IdP implementation and usually
cancels the user's session with IdP\@. This entails that IdP
will not issue new UCs for the user without
re-authentication.

\subsection{Sideshow and BigTent}\label{sec:sideshow}

Since several email providers, such as \nolinkurl{gmail.com} and
\nolinkurl{yahoo.com}, already use OpenID \cite{openid}, a widely
employed SSO system, Mozilla implemented IdPs called Sideshow and
BigTent which use an OpenID backend for user authentication:
Sideshow/BigTent are put between BrowserID and an email provider
running OpenID\@. That is, BrowserID uses Sideshow/BigTent as an
IdP\@. Sideshow/BigTent translate requests from BrowserID to requests
to the email provider's OpenID interface.  Currently, Sideshow and
BigTent are used to provide BrowserID support for
\nolinkurl{gmail.com} and \nolinkurl{yahoo.com}, respectively. In what
follows, we describe Sideshow in more detail; BigTent is
similar. Technical details on the communication between OpenID and
Sideshow/BigTent can be found in
Appendix~\ref{app:sideshow-openid-flow}.

All BrowserID protocol steps that would normally be carried out by the
IdP are now handled by Sideshow (i.e., the Sideshow server). For this
purpose, Sideshow serves the provisioning URL (for the PIF) and the
authentication URL used in \refprotophase{ld-auth}.  It maintains a
session with the user's browser. This session is considered to be
authenticated if the user has successfully authenticated to Sideshow
using OpenID\@. In this case, Sideshow's PIF document may send public
keys to Sideshow. Sideshow then creates a UC for the identity it
believes to be logged in. If the session at Sideshow is not
authenticated, the user will first be redirected to the Sideshow
authentication URL\@. Sideshow's authentication document will redirect
the user further to the OpenID URL at Gmail. This URL contains an
authentication request encoding that Sideshow requests an \emph{OpenID
  assertion} that contains an email address. In general, such an
assertion is a list of attribute name/value pairs (partially) MACed by
Gmail with a temporary symmetric key known only to Gmail; an
additional attribute, \texttt{openid.signed}, in such an assertion
encodes which attribute name/value pairs have actually been MACed and
in which order.  The user now authenticates to Gmail. Then, Gmail
issues the requested OpenID assertion and redirects the browser to
Sideshow with the assertion in the URL parameters. Sideshow then sends
the OpenID assertion to Gmail in order to check its validity. If the
OpenID assertion is valid, i.e. the MAC over the attributes listed in
\texttt{openid.signed} verifies, Sideshow considers its session with
the user's browser to be authenticated for the email address contained
in the OpenID assertion.

\subsection{Secondary Identity Provider}\label{sec:ident-prov-fallb}

If an email provider (IdP) does not directly support BrowserID, LPO
can be used as a so-called \emph{secondary IdP} (sIdP), i.e., it
replaces the IdP completely.  For this, the user has to register at
LPO\@. That is, she creates an account at LPO where she can register one
or more email addresses to be used as identities. She has to prove
ownership of all email addresses she registers. (LPO sends URLs 
to each email address, which then have to be opened by the user.)

When the sIdP is used, the phases \refprotophase{ld-prov-1} --
\refprotophase{ld-lpo-auth} are not needed as now LPO replaces the IdP
and the actions previously performed by IdP and LPO are now carried
out by LPO alone. The user is prompted to enter her password directly
into LD\@. If the password is correct, LPO now considers the session
with the browser to be authenticated. LPO will then issue UCs on
behalf of the email provider. We note that, for automatic CAP
creation, the CIF (see Section~\ref{sec:javascript-descr}) is still
used.

A detailed step-by-step description of a BrowserID flow with an sIdP
can be found in Appendix~\ref{app:browserid-sidp-lowlevel}.

\section{Analysis of BrowserID}\label{sec:analysisbrowserid}

In this section, we present the analysis of the BrowserID
system. We first formulate fundamental security properties
for the BrowserID system. We then present attacks that show
that these properties are not satisfied and propose
fixes. For the case of BrowserID with sIdP and the fixes
applied, we then prove that the security properties are
satisfied in our web model. We note that we also
incorporate the automated CAP creation with the CIF in our
model of BrowserID (see
Section~\ref{sec:javascript-descr}).  Our web model is
expressive enough to also formally model the BrowserID
system with primary IdPs (and Sideshow/BigTent) in a
straightforward way. However, we leave the detailed
formulation of such a model and the proof of the security
of the fixed system with \emph{primary} IdPs to future
work.

\subsection{Security Properties for BrowserID}\label{sec:securitypropsBrowserID}

While the documentation of BrowserID does not contain
explicit security goals, we deduce two fundamental security
properties that can be informally described as follows (see
Appendix~\ref{app:securitypropertiesbrowserid} for a formal
description): \textbf{(A)} \emph{The attacker should not be
  able to use a service of RP as an honest user.} In other
words, the attacker should not get hold of (be able to
derive from his current knowledge) an RP service token for
an ID of an honest user (browser), even if the browser was
closed and then later used by a malicious user (i.e., after
a $\closecorrupt$). \textbf{(B)} \emph{The attacker should
  not be able to authenticate an honest browser to an RP
  with an ID that is not owned by the browser.}

\subsection{Attacks on BrowserID}

\label{sec:attacks-browserid}
Our analysis of BrowserID w.r.t.~the above security
properties revealed several attacks (as sketched next). We
confirmed the attacks on the actual implementation and also
reported them to Mozilla. The
first two fixes proposed below have been adopted by Mozilla
already and the others are currently under discussion at
Mozilla.

\subsubsection{Identity Forgery} 

There are two problems in
Sideshow that lead to identity forgery attacks for Gmail
addresses; analogously in BigTent with Yahoo email
addresses.\footnote{\label{footnote:identity-forgery}See
  \url{https://bugzilla.mozilla.org/show_bug.cgi?id=920030}
  and \\
  \url{https://bugzilla.mozilla.org/show_bug.cgi?id=920301}.}

\myparagraph{a)} It is not checked if all requested
attributes in the OpenID assertion are MACed, which allows
for the following attack: A (web) attacker may choose any
Gmail address to impersonate, say
\nolinkurl{victim@gmail.com}. He starts a BrowserID login
with this email address. When he is then redirected to the
OpenID URL at Gmail, he removes the email attribute from
Sideshow's authentication request. The attacker
authenticates himself at Gmail with his own account (say,
\nolinkurl{attacker@gmail.com}). Upon receipt of the OpenID
assertion, he appends the email attribute with value
\nolinkurl{victim@gmail.com} and forwards it to
Sideshow. The assertion is declared valid by Gmail since
the MAC is correct (the email attribute is not listed in
\texttt{openid.signed}). Since Sideshow does not require
the email attribute to be in \texttt{openid.signed}, it
accepts the OpenID assertion, considers the attacker's
session to be authenticated for
\nolinkurl{victim@gmail.com}, and issues UCs for this
address to the attacker. This violates Condition
\textbf{(A)}.

\myparagraph{b)} Sideshow uses the first email address in
the OpenID assertion (based on the attribute type
information), which is not necessarily the MACed email
address. This allows for an attack similar to the above,
except that the attacker does not need to change Sideshow's
authentication request but only prepends the victim's email
address to the OpenID assertion in an additional attribute.

\myparagraph{Proposed fix.} Sideshow/BigTent must ensure to
use the correct and MACed attribute for the email address.

\subsubsection{Login Injection Attack} 
During the login process, the origin of the response \pm
(\refprotostep{ld-rpdoc-cap} in
Figure~\ref{fig:browserid-lowlevel-ld}), which contains the
CAP, is not checked. An attacker (e.g., in a malicious
advertisement \iframe within RP-Doc or in a parent window
of RP-Doc), can continuously send \pms to the RP-Doc with
his own CAP in order to log the user into his own account.
The outcome of this attack is similar to session
swapping. For example, if the attacker is able to log the
user into a search engine, the attacker might be able to
read the search terms the user enters.
This attack violates Condition \textbf{(B)}.\footnote{See
  \url{https://bugzilla.mozilla.org/show_bug.cgi?id=868967}}

\myparagraph{Proposed fix.} To fix the problem, the
sender's origin of the \pm~\refprotostep{ld-rpdoc-cap} must be checked to
match LPO\@.

\subsubsection{Key Cleanup Failure Attack} When LD creates
a key pair (\refprotostep{gen-key-pair} in
Figure~\ref{fig:browserid-lowlevel-ld}), it stores the keys
in the localStorage (even in ephemeral sessions). When a
user quits a session (e.g, by clicking on RP's logout
button and closing the browser) the key pair (and the UC)
remain in the localStorage, unlike session cookies. Hence,
users of shared terminals can read the \ls (in our model, a
$\closecorrupt$ allows an attacker to do this) and then,
using the key pair and the UC, create valid CAPs to log in
at any RP under the identity of the previous user, which
violates Condition~\textbf{(A)}.\footnote{See
  \url{https://github.com/mozilla/browserid/issues/3770}}

\myparagraph{Proposed fix.} We propose to use the \ls for
this data only in non-ephemeral sessions. 

\subsubsection{Cookie Cleanup Failure Attack (for the case
  of secondary IdP only)} The LPO session cookie is not
deleted when the browser is closed, even in ephemeral
sessions and even if a user logged out at RP
beforehand. (In our model, if the attacker issues a
$\closecorrupt$, he can therefore still access the LPO
session cookie.)  Hence, another user of the same browser
could request new UCs for \emph{any} ID registered at LPO
for that user, and hence, log in at any RP under this ID,
which violates Condition \textbf{(A)}.\footnote{See
  \url{https://github.com/mozilla/browserid/issues/3769}}

\myparagraph{Proposed fix.} In ephemeral sessions, LPO
should limit the cookie lifetime to the browser
session.

\subsection{Analysis of BrowserID with sIdP}\label{sec:analysisbrowseridsidp}\label{sec:browerIDmodel}

We now present an overview of our formal model and analysis
of BrowserID with sIdP. More details are presented in
Appendix~\ref{app:browseridmodel}. We consider ephemeral
sessions (the default), which are supposed to last until
the browser is closed. We assume that users are already
registered at LPO, i.e., they have accounts at LPO with one
or more email addresses registered in each account.

More specifically, we first model the BrowserID system as a
web system (in the sense of Section~\ref{sec:webmodel}),
then precisely formalize the security properties already
sketched in Section~\ref{sec:securitypropsBrowserID} in
this model, and finally prove, for the BrowserID model with
the fixes proposed in Section~\ref{sec:attacks-browserid}
applied (otherwise the proof would not go through), that
these security properties are satisfied. 

We call a web system $\bidwebsystem=(\bidsystem,
\scriptset, \mathsf{script}, E_0)$ a \emph{BrowserID web
  system} if it is of the form described in what follows.

The system $\bidsystem=\mathsf{Hon}\cup \mathsf{Web} \cup
\mathsf{Net}$ consists of the (network) attacker process
$\fAP{attacker}$, the web server for $\fAP{LPO}$, a finite
set $\fAP{B}$ of web browsers, and a finite set $\fAP{RP}$
of web servers for the relying parties, with $\mathsf{Hon}
:= \fAP{B} \cup \fAP{RP} \cup \{\fAP{LPO}\}$, $\mathsf{Web}
:= \emptyset$, and $\mathsf{Net} :=
\{\fAP{attacker}\}$. DNS servers are assumed to be
dishonest, and hence, are subsumed by
$\fAP{attacker}$. More details on the processes in
$\bidsystem$ are provided below.

The set $\nonces$ of nonces is partitioned into three sets,
an infinite set $N^\bidsystem$, an infinite set
$K_\text{private}$, and a finite set $\PLISecrets$. The set
$N^\bidsystem$ is further partitioned into infinite sets of
nonces, one set $N^p\subseteq N^\bidsystem$ for every $p\in
\bidsystem$.

The set $\addresses$ contains for $\fAP{LPO}$,
$\fAP{attacker}$, every relying party in $\fAP{RP}$, and
every browser in $\fAP{B}$ one address each. By
$\mapAddresstoAP$ we denote the corresponding assignment
from a process to its address.  The set $\dns$ contains one
domain for $\fAP{LPO}$, one for every relying party in
$\fAP{RP}$, and a finite set of domains for
$\fAP{attacker}$. Each domain is assigned a fresh private
key (a nonce). Additionally, $\fAP{LPO}$ has a fresh
signing key $k^\fAP{LPO}$, which it uses to create UCs.

Each browser $b\in \fAP{B}$ owns a finite set of secrets
($\subseteq \PLISecrets$) for $\fAP{LPO}$ and each secret
is assigned a finite set of email addresses (IDs) of the
form $\an{\mi{name},d}$, with $\mi{name}\in \mathbb{S}$ and
$d\in \dns$, such that browsers have disjoint sets of
secrets and secrets have disjoint sets of IDs. An \emph{ID
  $i$ is owned by a browser $b$} if the secret associated
with $i$ belongs to $b$.

The set $\scriptset$ contains four scripts, with their
string representations defined by $\mathsf{script}$: the
honest scripts running in RP-Doc, CIF, and LD,
respectively, and the malicious script $\Rasp$ (see below
for more details). 

The set $E_0$ contains only the trigger events specified in
Definition~\ref{def:websystem}.

Before we specify the processes in $\websystem$, we first
note that a UC $\mi{uc}$ for a user $u$ with email address
$i$ and public key (verification key) $\pub(k_u)$, where
$k_u$ is the private key (signing key) of $u$, is modeled
to be a message of the form $\mi{uc}=\sig{\an{i,
    \pub(k_u)}}{k^\fAP{LPO}}$, with $k^\fAP{LPO}$ as
defined above. An IA $ia$ for an origin $\mi{o}$ (e.g.,
$\an{\str{example.com},\str{S}}$) is a message of the form
$ia=\sig{o}{k_u}$. Now, a CAP is of the form
$\an{\mi{uc},\mi{ia}}$. Note that the time stamps are
omitted both from the UC and the IA\@. This models that
both certificates are valid indefinitely. In reality, as
explained in Section~\ref{sec:browserid}, they are valid
for a certain period of time, as indicated by the time
stamps. So our modeling is a safe overapproximation.

We are now ready to define the processes in $\websystem$ as
well as the scripts in $\scriptset$ in more detail. We note
that in Appendix~\ref{app:browseridmodel}, we provide a detailed formal
specification of the processes and scripts in the style of
Algorithm~\ref{alg:dns}. 

All processes in $\websystem$ contain in their initial
states all public keys and the private keys of their
respective domains (if any). We define
$I^p=\{\mapAddresstoAP(p)\}$ for all $p\in
\mathsf{Hon}$.  

\myparagraph{Attacker.}  The $\fAP{attacker}$ process is a
network attacker (see Section~\ref{sec:websystem}), who
uses all addresses for sending and listening. All parties
use the attacker as a DNS server. See
Appendix~\ref{app:attacker} for details.

\myparagraph{Browsers.}  Each $b \in \fAP{B}$ is a web
browser as defined in Section~\ref{sec:web-browsers}. The
initial state contains all secrets owned by $b$, stored
under the origin $\an{\mapDomain(\fAP{LPO}),\str{S}}$ of
LPO; $\mi{sts}$ is $\an{\mapDomain(\fAP{LPO})}$. See
Appendix~\ref{sec:browsers} for details.

\myparagraph{LPO.} The initial state of $\fAP{LPO}$
contains its signing key $k^\fAP{LPO}$, all secrets in
$\PLISecrets$ and the corresponding IDs. The definition of
$R^\fAP{LPO}$ closely follows the description of LPO in
Section~\ref{sec:ident-prov-fallb}. Sessions of $\fAP{LPO}$
expire non-deterministically. UCs are signed using
$k^\fAP{LPO}$. See Appendix~\ref{sec:lpo} for details.

\myparagraph{Relying Parties.}  A relying party $r \in
\fAP{RP}$ is a web server. The definition of $R^r$ follows
the description in Section~\ref{sec:browserid} and the
security considerations in
\cite{mozilla/persona/mdn}.\footnote{Mozilla recommends to
  (1) protect against Cross-site Request Forgery ($R^r$
  checks the Origin header, which is always set in our
  model), (2) verify CAPs on the server (rather than in the
  browser), (3) check if the CAP is issued for the correct
  RP, and (4) verify SSL certificates.}  RP answers any
$\mGet$ request with the script $\str{script\_RP\_index}$
(see below). When receiving an HTTPS $\mPost$ message, RP
checks (among others) if the message contains a valid
CAP\@. If successful, RP responds with an \emph{RP service
  token for ID $i$} of the form $\an{n, i}$, where $i \in
\IDs$ is the ID for which the CAP was issued and $n$ is a
freshly chosen nonce. The RP $r$ keeps a list of such
tokens in its state. Intuitively, a client having such a
token can use the service of $r$ for ID $i$. See
Appendix~\ref{app:relying-parties} for details.

\myparagraph{BrowserID Scripts.}  The set $\scriptset$
consists of the scripts $\Rasp$,\! $\mi{script\_RP\_index}$,\!
$\mi{script\_LPO\_cif}$, and $\mi{script\_LPO\_ld}$ with
their string representations $\str{att\_script}$,
$\str{script\_RP\_index}$,\linebreak $\str{script\_LPO\_cif}$, and
$\str{script\_LPO\_ld}$. The latter two scripts (issued by
$\fAP{LPO}$) are defined in a straightforward way following
the implementation outlined in
Section~\ref{sec:browserid}. The script
$\mi{script\_RP\_index}$ (issued by RP) also includes the
script that is (in reality) loaded from LPO\@. In
particular, this script creates the CIF and the LD
\iframes/subwindows, whose contents (scripts) are loaded
from LPO. See Appendix~\ref{app:browserid-scripts} for
details.

\subsection{Security of the Fixed System}

We call a BrowserID web system $\bidwebsystem$ with the
fixes proposed in Section~\ref{sec:attacks-browserid} a
\emph{fixed BrowserID web system}. We now obtain the
following theorem, which says that such a system satisfies
the security properties \textbf{(A)} and \textbf{(B)}.

\begin{theorem}\label{the:securityFixedBrowserID}
\label{THE:SECURITYFIXEDBROWSERID} %
  Let $\bidwebsystem$ be a fixed BrowserID web
  system. Then, $\bidwebsystem$ is secure. 
\end{theorem}

The proof of this theorem is presented in
Appendix~\ref{app:proofbrowserid}.

\section{Related Work}\label{sec:relatedwork}

Early work in the direction of formal web security analysis
includes work by
Kerschbaum~\cite{kerschbaum-SP-2007-XSRF-prevention}, in
which a Cross-Site Request Forgery protection proposal is
formally analyzed using a simple model expressed using
Alloy, a finite-state model checker
\cite{Jackson-TACAS-2002-Alloy}.

In seminal work, Akhawe et
al.\cite{AkhawBarthLamMitchellSong-CSF-2010} initiated a
more general formal treatment of web security.  Again the
model was provided in the Alloy modeling language. Inspired
by this work, Bansal et
al.\cite{BansalBhargavanetal-POST-2013-WebSpi,BansalBhargavanMaffeis-CSF-2012}
built the WebSpi model for the web infrastructure, which is
encoded in the modeling language (a variant of the applied
pi-calculus \cite{AbadiFournet-POPL-2001}) of ProVerif, a
specialized tool for cryptographic protocol analysis
\cite{Blanchet-CSFW-2001}. Both models have successfully
been applied to find attacks in standards and web
applications.

We see our work as a complement to these models: On the one
hand, the above models support (fully) automated
analysis. On the other hand, our model is much more
comprehensive and accurate, but not directly suitable for
automation.\footnote{The tool-based models are necessarily
  tailored to and limited by constraints of the tools.  For
  example, models for Alloy are necessarily finite
  state. Terms (messages) need to be encoded in some way as
  they are not directly supported. Due to the analysis
  method employed in ProVerif, the WebSpi model is of a
  monotonic nature. For instance, cookies and localStorage
  entries can only be added, but not deleted or modified.
  Also, the number of cookies per request is
  limited. Several features (that have been crucial for the
  analysis of BrowserID) are not supported by the
  tool-based models, including the precise handling of
  windows, documents, and iframes as well as cross-document
  messaging (postMessages), and the ability for an attacker
  to take over a browser after it has been closed. Dealing
  with such features in an automated tool is indeed
  challenging.} We think that, similarly to the area of
cryptography, both approaches, automated analysis and
manual analysis, are very valuable. Clearly, it is highly
desirable to push automated analysis as much as possible,
given that manual proofs are laborious and
error-prone. Conversely, automated approaches may miss
important problems due to the less accurate models they
consider. Moreover, a ``service'' more comprehensive and
accurate models provide, even if they are manually driven,
is that they summarize and condense relevant aspects in the
various standards and specifications for the web. As such,
they are an important basis for the formal foundation and
discourse on web security and can serve as reference models
(for tool-supported analysis, web security researchers, for
developers of web technologies and standards, and maybe for
teaching basic web security concepts).

The BrowserID system has been analyzed before using the
AuthScan tool developed by Bai et
al.~\cite{baiLeietal-NDSS-2013-authscan}. Their work
focusses on the automated extraction of a model from a
protocol implementation. Their analysis of BrowserID is not
very detailed; only two rather trivial attacks are
identified, for example, CAPs that are sent unencrypted can
be replayed by the attacker to an RP\@. There is also work
on the analysis of other web-based single sign-on systems,
such as SAML-based single sign-on, OpenID, and OAuth (see,
e.g.,
\cite{Kumar-ACSAC-2012,ArmandoCarboneetal-FMSE-2008-SAML,HandenSkriverNielson-WITS-2005,SantsaiHaekyBeznosov-CS-2012,McIntyreLuterrothChristofWeber-EDOC-2011-OpenID,SantsaiBeznosov-CCS-2012-OAuth,Gross-ACSAC-2003-SAML,ChariJutlaRoy-IACR-2011,SovisKohlarSchwenk-Sicherheit-2010-OpenID-signatures,WangChenWang-SP-2012-SSO}). However,
none of these works are based on a model of the web
infrastructure.

In
\cite{SonShmatikov-NDSS-2013-postMessage,SovisKohlarSchwenk-Sicherheit-2010-OpenID-signatures,WangChenWang-SP-2012-SSO,HannaShinAkhaweetal-W2SP-2010-postmessage},
potentially problematic usage of postMessages and the
OpenID interface are discussed. While very useful, these
papers do not consider BrowserID or formal models, and they
do not formalize security properties for web applications
or establish formal security guarantees.

Bohannon and Pierce propose a formal model of a web browser
core \cite{BohannonPierce-USENIX-2010}. The scope and goal
of the model is different to ours, but some mechanisms can
be found in both models. B\"orger et al.~present an
approach for the analysis of web application frameworks,
focussing on the server \cite{Borgeretal-ABZ-2012}.

\section{Conclusion}\label{sec:conclusion}

We presented an expressive model of the web infrastructure
and web applications, the most comprehensive model for the
web infrastructure to date.  It contains many
security-relevant features and is designed to closely mimic
standards and specifications for the web. As such, it
constitutes a solid basis for the analysis of a broad range
of web standards and applications.

In our case study, we analyzed the BrowserID system, found
several very critical attacks, proposed fixes, and proved
the fixed system for the case of secondary IdP case secure
w.r.t.~the security properties we specified. The analysis
of this system is out of the scope of other models for the
web infrastructure.

As for future work, it is straightforward to incorporate
further features, such as subdomains, cross-origin resource
sharing, and finer-grained settings for cookie paths and
domains, which we have left out mainly for brevity of
presentation for now. Our model could serve as a basis and
a reference for automated approaches, where one could try
to extend the existing automated approaches or develop new
ones (e.g., based on theorem provers, where higher accuracy
is typically paid by more interaction). Finally, BrowserID
is being used by more and more web sites and it will
continue to be an interesting object of study. An obvious
next step is to analyze BrowserID for the case of primary
IdPs. The model is already expressive enough to carry out
such an analysis.  We also plan to apply our model to other
web applications and web standards.

\section*{Acknowledgement} 

The first author is supported by the \emph{Studienstiftung
  des Deutschen Volkes} (German National Academic
Foundation).

\appendix

\section{Message and Data
  Formats}\label{app:message-data-formats}

We now provide some more details about data and message
formats that are needed for the formal treatment of our web
model and our case study presented in the rest of the
appendix.

\subsection{Notations}\label{app:notation}
For a set $s$ and a sequence $t = \an{t_1,\dots,t_n}$ we
use $t \subsetPairing s$ to say that $t_1,\dots,t_n \in
s$.  We define $\left. x \inPairing t\right. \iff \exists
i: \left. t_i = x\right.$.
We write $t \plusPairing y$ to denote the sequence
$\an{t_1,\dots,t_n,y}$. For a sequence $t=
\an{t_1,\ldots,t_n}$ we define $|t| = n$. If $t$ is not a
sequence, we set $|t| = \notdef$. For a finite set $M$ with
$M = \{m_1, \dots,m_n\}$ we use $\an{M}$ to denote the term
of the form $\an{m_1,\dots,m_n}$. (The order of the
elements does not matter; one arbitrary is chosen.)

We also use specific terms which we call dictionaries.
\begin{definition}\label{def:dictionaries}
  A \emph{dictionary over $X$ and $Y$} is a term of the
  form $\an{\an{k_1, v_1}, \dots, \an{k_n,v_n}}$, where
  $k_1, \dots,k_n \in X$, $v_1,\dots,v_n \in Y$, and the
  keys $k_1, \dots,k_n$ are unique, i.e., $\forall i\neq j:
  k_i \neq k_j$. We call every term $\an{k_i,v_i}$, $i\in
  \{1,\ldots,n\}$, an \emph{element} of the dictionary with
  key $k_i$ and value $v_i$.  We often write $\left[k_1:
    v_1, \dots, k_i:v_i,\dots,k_n:v_n\right]$ instead of
  $\an{\an{k_1, v_1}, \dots, \an{k_n,v_n}}$. We denote the
  set of all dictionaries over $X$ and $Y$ by $\left[X
    \times Y\right]$.
\end{definition}
We note that the empty dictionary is equivalent to the
empty sequence: $[] = \an{}$.  Figure
\ref{fig:dictionaries} shows the short notation for
dictionary operations that will be used when describing the
browser atomic process. For a dictionary $z = \left[k_1:
  v_1, k_2: v_2,\dots k_n:v_n\right]$ we write $k \in z$ to
say that there exists $i$ such that $k=k_i$. We write
$z[k_j] := v_j$ to extract elements. If $k \not\in z$, we
set $z[k] := \an{}$.

\begin{figure}[htb!]\centering
  \begin{align}
    \left[k_1: v_1, \dots, k_i:v_i,\dots,k_n:v_n\right][k_i] = v_i%
  \end{align}\vspace{-2.5em}
  \begin{align}
    \nonumber \left[k_1: v_1, \dots, k_{i-1}:v_{i-1},k_i: v_i, k_{i+1}:v_{i+1}\dots,k_n: v_n\right]-k_i =\\
         \left[k_1: v_1, \dots, k_{i-1}:v_{i-1},k_{i+1}:v_{i+1}\dots,k_n: v_n\right]
  \end{align}
  \caption{Dictionary operators with $1\le i\le n$.}\label{fig:dictionaries}
\end{figure}

Given a term $t = \an{t_1,\dots,t_n}$, we can refer to any
subterm using a sequence of integers. The subterm is
determined by repeated application of the projection
$\pi_i$ for the integers $i$ in the sequence. We call such
a sequence a \emph{pointer}:

\begin{definition}\label{def:pointer}
  A \emph{pointer} is a sequence of non-negative integers. 
\end{definition}

We write $\tau.\ptr{p}$ for the application of the pointer
$\ptr{p}$ to the term $\tau$. This operator is applied from
left to right. For pointers consisting of a single integer,
we may omit the sequence braces for brevity.

\begin{example}
  For the term $\tau = \an{a,b,\an{c,d,\an{e,f}}}$ and the
  pointer $\ptr{p} = \an{3,1}$, the subterm of $\tau$ at
  the position $\ptr{p}$ is $c =
  \proj{1}{\proj{3}{\tau}}$. Also, $\tau.3.\an{3,1} =
  \tau.3.\ptr{p} = \tau.3.3.1 = e$.
\end{example}

To improve readability, we try to avoid writing, e.g.,
$\compn{o}{2}$ or $\proj{2}{o}$ in this document. Instead,
we will use the names of the components of a sequence that
is of a defined form as pointers that point to the
corresponding subterms. E.g., if an \emph{Origin} term is
defined as $\an{\mi{host}, \mi{protocol}}$ and $o$ is an
Origin term, then we can write $\comp{o}{protocol}$ instead
of $\proj{2}{o}$ or $\compn{o}{2}$. See also
Example~\ref{ex:url-pointers}.

In our pseudocode, we will write, for example, 

\medskip

\begin{algorithmic}
  \LetST{$x,y$}{$\an{\str{Constant},x,y} \equiv
    t$}{doSomethingElse}
\end{algorithmic}

\medskip

\noindent for some variables $x,y$, a string
$\str{Constant}$, and some term $t$ to express that $x :=
\proj{2}{t}$, and $y := \proj{3}{t}$ if $\str{Constant}
\equiv \proj{1}{t}$ and if $|\an{\str{Constant},x,y}| =
|t|$,  and that otherwise
$x$ and $y$ are not set and doSomethingElse is executed.

\subsection{URLs}\label{sec:urls}
\begin{definition}
  A \emph{URL} is a term $\an{\tUrl, \mi{protocol},
    \mi{host}, \mi{path}, \mi{params}}$ with $\mi{protocol}
  \in \{\http, \https\}$ (for \textbf{p}lain (HTTP) and
  \textbf{s}ecure (HTTPS)), $\mi{host} \in \dns$,
  $\mi{path} \in \mathbb{S}$ and $\mi{params} \in
  \dict{\mathbb{S}}{\terms}$. The set of all valid URLs
  is $\urls$.
\end{definition}

\begin{example} \label{ex:url-pointers}
  For the URL $u = \an{\tUrl, a, b, c, d}$, $\comp{u}{protocol} =
  a$. If, in the algorithm described later, we say $\comp{u}{path} :=
  e$ then $u = \an{\tUrl, a, b, c, e}$ afterwards. 
\end{example}

\subsection{Origins}\label{sec:origins}
\begin{definition} An \emph{origin} is a term of the form
  $\an{\mi{host}, \mi{protocol}}$ with $\mi{host} \in
  \dns$, $\mi{protocol} \in \{\http, \https\}$. We write
  $\origins$ for the set of all origins.  See Example~\ref{ex:window} for an example for an origin.
\end{definition}

\subsection{Cookies}\label{sec:cookies}
\begin{definition} A \emph{cookie} is a term of the form
  $\an{\mi{name}, \mi{content}}$ where $\mi{name} \in
  \terms$, and $\mi{content}$ is a term of the form
  $\an{\mi{value}, \mi{secure}, \mi{session},
    \mi{httpOnly}}$ where $\mi{value} \in
  \terms$,  $\mi{secure}$, $\mi{session}$,
  $\mi{httpOnly} \in \{\True, \bot\}$.  We write $\cookies$
  for the set of all cookies.
\end{definition}

Note that cookies of the form described here are only
contained in HTTP(S) requests. In responses, only the
components $\mi{name}$ and $\mi{value}$ are transferred as
a pairing of the form $\an{\mi{name}, \mi{value}}$.

\subsection{HTTP Messages}\label{sec:http-messages-full}
\begin{definition}
  An \emph{HTTP request} is a term of the form shown in
  (\ref{eq:default-http-request}). An \emph{HTTP response}
  is a term of the form shown in
  (\ref{eq:default-http-response}).
  \begin{align}
    \label{eq:default-http-request}
    & \hreq{ nonce=\mi{nonce}, method=\mi{method},
      xhost=\mi{host}, xpath=\mi{path},
      parameters=\mi{parameters}, headers=\mi{headers},
      xbody=\mi{body}
    } \\
    \label{eq:default-http-response}
    & \hresp{ nonce=\mi{nonce}, status=\mi{status},
      headers=\mi{headers}, xbody=\mi{body} }
  \end{align}
  The components are defined as follows:
  \begin{itemize}
  \item $\mi{nonce} \in \nonces$ serves to map each
    response to the corresponding request 
  \item $\mi{method} \in \methods$ is one of the HTTP
    methods.
  \item $\mi{host} \in \dns$ is the host name in the HOST
    header of HTTP/1.1.
  \item $\mi{path} \in \mathbb{S}$ is a string indicating
    the requested resource at the server side
  \item $\mi{status} \in \mathbb{S}$ is the HTTP status
    code (i.e., a number between 100 and 505, as defined by
    the HTTP standard)
  \item $\mi{parameters} \in
    \dict{\mathbb{S}}{\terms}$ contains URL parameters
  \item $\mi{headers} \in \dict{\mathbb{S}}{\terms}$,
    containing request/response headers. The dictionary
    elements are terms of one of the following forms: 
    \begin{itemize}
    \item $\an{\str{Origin}, o}$ where $o$ is an origin
    \item $\an{\str{Set{\mhyphen}Cookie}, c}$ where $c$ is
      a sequence of cookies
    \item $\an{\str{Cookie}, c}$ where $c \in
      \dict{\mathbb{S}}{\terms}$ (note that in this header,
      only names and values of cookies are transferred)
    \item $\an{\str{Location}, l}$ where $l \in \urls$
    \item $\an{\str{Strict{\mhyphen}Transport{\mhyphen}Security},\True}$
    \end{itemize}
  \item $\mi{body} \in \terms$ in requests and responses. 
  \end{itemize}
  We write $\httprequests$/$\httpresponses$ for the set of
  all HTTP requests or responses, respectively.
\end{definition}

\begin{example}[HTTP Request and Response]
  \begin{align}
    \label{eq:ex-request}
    \nonumber \mi{r} := & \langle
                   \cHttpReq,
                   n_1,
                   \mPost,
                   \str{example.com},
                   \str{/show},
                   \an{\an{\str{index,1}}},\\ & \quad
                   [\str{Origin}: \an{\str{example.com, \https}}],
                   \an{\str{foo}, \str{bar}}
                \rangle \\
    \label{eq:ex-response} \mi{s} := & \hresp{ nonce=n_1,
      status=200,
      headers=\an{\an{\str{Set{\mhyphen}Cookie},\an{\an{\str{SID},\an{n_2,\bot,\bot,\True}}}}},
      xbody=\an{\str{somescript},x}}
  \end{align}
  \noindent
  An HTTP $\mGet$ request for the URL
  \url{http://example.com/show?index=1} is shown in
  (\ref{eq:ex-request}), with an Origin header and a body
  that contains $\an{\str{foo},\str{bar}}$. A possible
  response is shown in (\ref{eq:ex-response}), which
  contains an httpOnly cookie with name $\str{SID}$ and
  value $n_2$ as well as the string representation
  $\str{somescript}$ of the scripting process
  $\mathsf{script}^{-1}(\str{somescript})$ (which should be
  an element of $\scriptset$) and its initial state
  $x$.
\end{example}

\subsubsection*{Encrypted HTTP
  Messages} \label{sec:http-messages-encrypted-full}

For HTTPS, requests are encrypted using the public key of
the server.  Such a request contains an (ephemeral)
symmetric key chosen by the client that issued the
request. The server is supported to encrypt the response
using the symmetric key.

\begin{definition} An \emph{encrypted HTTP request} is of
  the form $\enc{\an{m, k'}}{k}$, where $k$, $k' \in
  \nonces$ and $m \in \httprequests$. The corresponding
  \emph{encrypted HTTP response} would be of the form
  $\encs{m'}{k'}$, where $m' \in \httpresponses$. We call
  the sets of all encrypted HTTP requests and responses
  $\httpsrequests$ or $\httpsresponses$, respectively.
\end{definition}

\begin{example}
  \begin{align}
    \label{eq:ex-enc-request} \ehreqWithVariable{r}{k'}{\pub(k_\text{example.com})} \\
    \label{eq:ex-enc-response} \ehrespWithVariable{s}{k'}
  \end{align} The term (\ref{eq:ex-enc-request}) shows an
  encrypted request (with $r$ as in
  (\ref{eq:ex-request})). It is encrypted using the public
  key $\pub(k_\text{example.com})$.  The term
  (\ref{eq:ex-enc-response}) is a response (with $s$ as in
  (\ref{eq:ex-response})). It is encrypted symmetrically
  using the (symmetric) key $k'$ that was sent in the
  request (\ref{eq:ex-enc-request}).
\end{example}

\subsection{DNS Messages}\label{sec:dns-messages}
\begin{definition} A \emph{DNS request} is a term of the form
$\an{\cDNSresolve, \mi{domain}, \mi{nonce}}$ where $\mi{domain} \in
\dns$, $\mi{nonce} \in \nonces$. We call the set of all DNS requests
$\dnsrequests$.
\end{definition}

\begin{definition} A \emph{DNS response} is a term of the form
$\an{\cDNSresolved, \mi{result}, \mi{nonce}}$ with $\mi{result} \in
\addresses$, $\mi{nonce} \in \nonces$. We call the set of all DNS
responses $\dnsresponses$.
\end{definition}

As already mentioned in Section~\ref{sec:DNSservers}, DNS
servers are supposed to include the nonce they received in
a DNS request in the DNS response that they send back.

\section{Detailed Description of the Browser Model}
\label{sec:deta-descr-brows}

Following the informal description of the browser model in
Section~\ref{sec:web-browsers}, we now present a formal
model. We start by introducing some notation and
terminology. 

\subsection{Notation and Terminology (Web Browser State)}

\begin{definition} A \emph{window} is a term of the form $w
  = \an{\mi{nonce}, \mi{documents}, \mi{opener}}$ with
  $\mi{nonce} \in \nonces$, $\mi{documents} \subsetPairing
  \documents$ (defined below), $\mi{opener} \in \nonces
  \cup \{\bot\}$ where $\comp{d}{active} = \True$ for
  exactly one $d \inPairing \mi{documents}$ if
  $\mi{documents}$ is not empty (we then call $d$ the
  \emph{active document of $w$}).  We write $\windows$ for
  the set of all windows.  We write
  $\comp{w}{activedocument}$ to denote the active document
  inside window $w$ if it exists and $\an{}$ else.
\end{definition}
We will refer to the window nonce as \emph{(window)
  reference}. See Example~\ref{ex:window} for an example of
a window.

The documents contained in a window term
to the left of the active document are the previously
viewed documents (available to the user via the ``back''
button) and the documents in the window term to the right
of the currently active document are documents available
via the ``forward'' button, as will be clear from the
description of web browser model (see
Section~\ref{sec:descr-web-brows}).

A window $a$ may have opened a top-level window
$b$ (i.e., a window term which is not a subterm of a
document term). In this case, the \emph{opener} part of the
term $b$ is the nonce of $a$, i.e., $\comp{b}{opener} =
\comp{a}{nonce}$.

\begin{definition} A \emph{document} $d$ is a term of the
  form 
  \begin{align*}
  \an{\mi{nonce}, \mi{origin}, \mi{script},
  \mi{scriptstate},\mi{scriptinput}, \mi{subwindows},
  \mi{active}}  
  \end{align*}
 where $\mi{nonce} \in \nonces$,
  $\mi{origin} \in \origins$, $\mi{script}$,
  $\mi{scriptstate}$, $\mi{scriptinput} \in \terms$,
  $\mi{subwindows} \subsetPairing \windows$, $\mi{active}
  \in \{\True, \bot\}$.  A \emph{limited document} is a
  term of the form $\an{\mi{nonce}, \mi{subwindows}}$ with
  $\mi{nonce}$, $\mi{subwindows}$ as above.  A window $w
  \inPairing \mi{subwindows}$ is called a \emph{subwindow}
  (of $d$).  We write $\documents$ for the set of all
  documents.
\end{definition}
We will refer to the document nonce as \emph{(document)
  reference}.  An example for window and document terms was given in
Example~\ref{ex:window}.

We can now define the set of states of web browsers. Note
that we use the dictionary notation that we introduced in
Definition~\ref{def:dictionaries}.

\begin{definition} Let $\mi{OR} := \left\{\an{o,r} \middle|\,
    o \in \origins,\, r \in \nonces\right\}$. The \emph{set
    of states $Z^p$ of a web browser atomic process} $p$
  consists of the terms of the form
  \begin{align*}
    \langle\mi{windows}, \mi{secrets}, \mi{cookies},
    \mi{localStorage}, \mi{sessionStorage},
    \mi{keyMapping}, \\\mi{sts}, \mi{DNSaddress},
    \mi{nonces}, \mi{pendingDNS}, \mi{pendingRequests},
    \mi{isCorrupted}\rangle
  \end{align*}
  where $\mi{windows} \subsetPairing \windows$,
  $\mi{secrets} \in \dict{\origins}{\nonces}$,
  $\mi{cookies}$ is a dictionary over $\dns$ and
  dictionaries of $\cookies$, $\mi{localStorage} {\in}\!
  \dict{\origins}{\terms}$, $\mi{sessionStorage} {\in}\!
  \dict{\mi{OR}}{\terms}$, $\mi{keyMapping} {\in}$ $
  \dict{\dns}{\nonces}$, $\mi{sts} {\subsetPairing} \dns$,
  $\mi{DNSaddress} {\in} \addresses$, $\mi{nonces}
  {\subsetPairing} \nonces$, $\mi{pendingDNS} {\in}\!
  \dict{\nonces}{\terms}$, $\mi{pendingRequests} {\in}$
  $\terms$, $\mi{isCorrupted} \in \{\bot, \fullcorrupt,
  \closecorrupt\}$.  We call the set of all states of
  standard HTTP browsers $\mathsf{SHBStates}$.
\end{definition}

\begin{definition} For two window terms $w$ and $w'$ we
  write $w \windowChildOf w'$ if $w \inPairing$ \linebreak
  $\comp{\comp{w'}{activedocument}}{subwindows}$. We write
  $\windowChildOfX$ for the transitive closure.
\end{definition}

In the following description of the web browser relation
$R^p$ we will use the helper functions
$\mathsf{Subwindows}$, $\mathsf{Docs}$, $\mathsf{Clean}$,
$\mathsf{CookieMerge}$ and $\mathsf{AddCookie}$. 

Given a browser state $s$, $\mathsf{Subwindows}(s)$ denotes
the set of all pointers\footnote{Recall the definition of a
  pointer from Definition~\ref{def:pointer}.} to windows in
the window list $\comp{s}{windows}$, their (active)
documents, and the subwindows of these documents
(recursively). We exclude subwindows of inactive documents
and their subwindows. With $\mathsf{Docs}(s)$ we denote the
set of pointers to all active documents in the set of
windows referenced by $\mathsf{Subwindows}(s)$.
\begin{definition} 
  For a browser state $s$ we denote by
  $\mathsf{Subwindows}(s)$ the minimal set of
  pointers that satisfies the
  following conditions: (1) For all windows $w \inPairing
  \comp{s}{windows}$ there is a $\ptr{p} \in
  \mathsf{Subwindows}(s)$ such that $\compn{s}{\ptr{p}} =
  w$. (2) For all $\ptr{p} \in \mathsf{Subwindows}(s)$, the
  active document $d$ of the window $\compn{s}{\ptr{p}}$
  and every subwindow $w$ of $d$ there is a pointer
  $\ptr{p'} \in \mathsf{Subwindows}(s)$ such that
  $\compn{s}{\ptr{p'}} = w$.

  Given a browser state $s$, the set $\mathsf{Docs}(s)$ of
  pointers to active documents is the minimal set such that
  for every $\ptr{p} \in \mathsf{Subwindows}(s)$, there is
  a pointer $\ptr{p'} \in \mathsf{Docs}(s)$ with
  $\compn{s}{\ptr{p'}} =
  \comp{\compn{s}{\ptr{p}}}{activedocument}$.
\end{definition}

The function $\mathsf{Clean}$ will be used to determine
which information about windows and documents the script
running in the document $d$ has access to.
\begin{definition} Let $s$ be a browser state and $d$ a
  document.  By $\mathsf{Clean}(s, d)$ we denote the term
  that equals $\comp{s}{windows}$ but with all inactive
  documents removed (including their subwindows etc.) and
  all subterms that represent non-same-origin documents
  w.r.t.~$d$ replaced by a limited document $d'$ with the
  same nonce and the same subwindow list. Note that
  non-same-origin documents on all levels are replaced by
  their corresponding limited document.
\end{definition}

The function $\mathsf{CookieMerge}$ merges two sequences of
cookies together: When used in the browser,
$\mi{oldcookies}$ is the sequence of existing cookies for
some origin, $\mi{newcookies}$ is a sequence of new cookies
that was outputted by some script. The sequences are merged
into a set of cookies using an algorithm that is based on
the \emph{Storage Mechanism} algorithm described in
RFC6265.
\begin{definition} \label{def:cookiemerge} The set
  $\mathsf{CookieMerge}(\mi{oldcookies}, \mi{newcookies})$
  for two sequences  $\mi{oldcookies}$ and $\mi{newcookies}$
  of cookies (where the cookies in $\mi{oldcookies}$ have
  pairwise different names) is defined by the following
  algorithm: From $\mi{newcookies}$ remove all cookies $c$
  that have $c.\str{content}.\str{httpOnly} \equiv
  \True$. For any $c$, $c' \inPairing \mi{newcookies}$,
  $\comp{c}{name} \equiv \comp{c'}{name}$, remove the
  cookie that appears left of the other in
  $\mi{newcookies}$.  Let $m$ be the set of cookies that
  have a name that either appears in $\mi{oldcookies}$ or
  in $\mi{newcookies}$, but not in both. For all pairs of
  cookies $(c_\text{old}, c_\text{new})$ with $c_\text{old}
  \inPairing \mi{oldcookies}$, $c_\text{new} \inPairing
  \mi{newcookies}$, $\comp{c_\text{old}}{name} \equiv
  \comp{c_\text{new}}{name}$, add $c_\text{new}$ to $m$ if
  $\comp{\comp{c_\text{old}}{content}}{httpOnly} \equiv \bot$
  and add $c_\text{old}$ to $m$ otherwise. The result of
  $\mathsf{CookieMerge}(\mi{oldcookies}, \mi{newcookies})$
  is $m$. 
\end{definition}

The function $\mathsf{AddCookie}$ adds a cookie $c$ to the
sequence of cookies contained in the sequence
$\mi{oldcookies}$. It is again based on the algorithm
described in RFC6265 but simplified for the use in the
browser model.
\begin{definition} \label{def:addcookie} The sequence
  $\mathsf{AddCookie}(\mi{oldcookies}, c)$, where
  $\mi{oldcookies}$ is a sequence of cookies with different
  names and $c$ is a cookie $c$, is defined by the
  following algorithm: Let $m := \mi{oldcookies}$. Remove
  any $c'$ from $m$ that has $\comp{c}{name} \equiv
  \comp{c'}{name}$. Append $c$ to $m$ and return $m$.
\end{definition}

The function $\mathsf{NavigableWindows}$ returns a set of
windows that a document is allowed to navigate. We closely
follow \cite{html5}, Section~5.1.4 for this definition.
\begin{definition} $\mathsf{NavigableWindows}(\ptr{w}, s')$
  is the set $\ptr{W} \subseteq
  \mathsf{Subwindows}(s')$ of pointers to windows that the
  active document in $\ptr{w}$ is allowed to navigate. The
  set $\ptr{W}$ is defined to be the minimal set such that
  for every $\ptr{w'}
  \in \mathsf{Subwindows}(s')$ the following is true: 
\begin{itemize}
\item If $\comp{\comp{\compn{s'}{\ptr{w}'}}{activedocument}}{origin}
      \equiv
      \comp{\comp{\compn{s'}{\ptr{w}}}{activedocument}}{origin}$ (the active documents in $\ptr{w}$ and $\ptr{w'}$ are
  same-origin), then $\ptr{w'} \in \ptr{W}$, and
\item If ${\compn{s'}{\ptr{w}} \childof \compn{s'}{\ptr{w'}}}$
  $\wedge$ $\nexists\, \ptr{w}'' \in
  \mathsf{Subwindows}(s')$ with $\compn{s'}{\ptr{w}'}
  \childof \compn{s'}{\ptr{w}''}$ ($\ptr{w'}$ is a
  top-level window and $\ptr{w}$ is an ancestor window of
  $\ptr{w'}$), then $\ptr{w'} \in \ptr{W}$, and
\item If $\exists\, \ptr{p} \in \mathsf{Subwindows}(s')$ such
  that $\compn{s'}{\ptr{w}'} \windowChildOfX
  \compn{s'}{\ptr{p}}$ $\wedge$
  $\comp{\comp{\compn{s'}{\ptr{p}}}{activedocument}}{origin}
  =
  \comp{\comp{\compn{s'}{\ptr{w}}}{activedocument}}{origin}$
  ($\ptr{w'}$ is not a top-level window but there is an
  ancestor window $\ptr{p}$ of $\ptr{w'}$ with an active
  document that has the same origin as the active document
  in $\ptr{w}$), then $\ptr{w'} \in \ptr{W}$, and
\item If $\exists\, \ptr{p} \in \mathsf{Subwindows}(s')$ such
  that $\comp{\compn{s'}{\ptr{w'}}}{opener} =
  \comp{\compn{s'}{\ptr{p}}}{nonce}$ $\wedge$ $\ptr{p} \in
  \ptr{W}$ ($\ptr{w'}$ is a top-level window---it has an
  opener---and $\ptr{w}$ is allowed to navigate the opener
  window of $\ptr{w'}$, $\ptr{p}$), then $\ptr{w'} \in
  \ptr{W}$. 
\end{itemize}
\end{definition}

\subsection{Description of the Web Browser Atomic
  Process}\label{sec:descr-web-brows}
We will now describe the relation $R^p$ of a standard HTTP
browser $p$.  For a tuple $r =
\left(\left(\left(a{:}f{:}m\right), s\right), \left(M,
    s'\right)\right)$ we define $r$ to belong to $R^p$
if\/f the non-deterministic algorithm presented in
Section~\ref{app:mainalgorithmwebbrowserprocess}, when
given $\left(\left(a{:}f{:}m\right), s\right)$ as input,
terminates with \textbf{stop}~$M$,~$s'$, i.e., with output
$M$ and $s'$. Recall that $\left(a{:}f{:}m\right)$ is an
(input) event and $s$ is a (browser) state, $M$ is a set of
(output) events, and $s'$ is a new (browser) state.

The notation $\textbf{let}\ n \leftarrow N$ is used to
describe that $n$ is chosen non-deterministically from the
set $N$.  We write $\textbf{for each}\ s \in M\ \textbf{do}$
to denote that the following commands (until \textbf{end
  for}) are repeated for every element in $M$, where the
variable $s$ is the current element. The order in which the
elements are processed is chosen non-deterministically.

We first define some functions which will be used in the
main algorithm presented in
Section~\ref{app:mainalgorithmwebbrowserprocess}.

\subsubsection{Functions} \label{app:proceduresbrowser} In
the description of the following functions we use $a$,
$f$, $m$, $s$ and $N^p$ as read-only global input
variables. Also, the functions use the set $N^p$ as a
read-only set. All other variables are local variables or
arguments.

$\mathsf{TAKENONCE}$ returns a nonce from the set of unused
nonces and modifies the browser state such that the nonce
is added to the sequence of used nonces. Note that this
function returns two values, the nonce $n$ and the modified
state $s'$.
\captionof{algorithm}{\label{alg:takenonce}
  Non-deterministically choose a fresh nonce.}
\begin{algorithmic}[1]
  \Function{$\mathsf{TAKENONCE}$}{$s'$}
    \LetND{$n$}{$\left\{x \middle| x \in N^p \wedge x \not\inPairing \comp{s'}{nonces} \right\}$} %
    \Append{$n$}{$\comp{s'}{nonces}$}
    \State \Return $n, s'$
  \EndFunction
\end{algorithmic} \setlength{\parindent}{1em}

The following function, $\mathsf{GETNAVIGABLEWINDOW}$, is
called by the browser to determine the window that is
\emph{actually} navigated when a script in the window
$s'.\ptr{w}$ provides a window reference for navigation
(e.g., for opening a link). When it is given a window
reference (nonce) $\mi{window}$,
$\mathsf{GETNAVIGABLEWINDOW}$ returns a pointer to a
selected window term in $s'$:
\begin{itemize}
\item If $\mi{window}$ is the string $\wBlank$, a new
  window is created and a pointer to that window is
  returned.
\item If $\mi{window}$ is a nonce (reference) and there is
  a window term with a reference of that value in the
  windows in $s'$, a pointer $\ptr{w'}$ to that window term
  is returned, as long as the window is navigable by the
  current window's document (as defined by
  $\mathsf{NavigableWindows}$ above).
\end{itemize}
In all other cases, $\ptr{w}$ is returned instead (the
script navigates its own window).
\captionof{algorithm}{\label{alg:getnavigablewindow}
  Determine window for navigation.}
\begin{algorithmic}[1]
  \Function{$\mathsf{GETNAVIGABLEWINDOW}$}{$\ptr{w}$, $\mi{window}$, $s'$}
    \If{$\mi{window} \equiv \wBlank$} \Comment{Open a new window when $\wBlank$ is used}
      \Let {$n$, $s'$}{\textsf{TAKENONCE}$(s')$}
      \Let{$w'$}{$\an{n, \an{}, \comp{\compn{s'}{\ptr{w}}}{nonce} }$}
      \Append{$w'$}{$\comp{s'}{windows}$}  \textbf{and} let
      $\ptr{w}'$ be a pointer to this new element in $s'$
      \State \Return{$(\ptr{w}', s')$}
    \EndIf
    \LetNDST{$\ptr{w}'$}{$\mathsf{NavigableWindows}(\ptr{w},
      s')$}{$\comp{\compn{s'}{\ptr{w}'}}{nonce} \equiv
      \mi{window}$}{\textbf{return} $(\ptr{w}, s')$}
    \State \Return{$(\ptr{w'}, s')$}
  \EndFunction
\end{algorithmic} \setlength{\parindent}{1em}

The following function takes a window reference as input
and returns a pointer to a window as above, but it checks
only that the active documents in both windows are
same-origin. It creates no new windows.
\captionof{algorithm}{\label{alg:getwindow} Determine same-origin window.}
\begin{algorithmic}[1]
  \Function{$\mathsf{GETWINDOW}$}{$\ptr{w}$, $\mi{window}$, $s'$}
    \LetNDST{$\ptr{w}'$}{$\mathsf{Subwindows}(s')$}{$\comp{\compn{s'}{\ptr{w}'}}{nonce} \equiv \mi{window}$}{\textbf{return} $(\ptr{w}, s')$}
    \If{
      $\comp{\comp{\compn{s'}{\ptr{w}'}}{activedocument}}{origin}
      \equiv
      \comp{\comp{\compn{s'}{\ptr{w}}}{activedocument}}{origin}$
    }
      \State \Return{$(\ptr{w}', s')$}
    \EndIf
    \State \Return{$(\ptr{w}, s')$}
  \EndFunction
\end{algorithmic} \setlength{\parindent}{1em}

The next function is used to stop any pending
requests for a specific window. From the pending requests
and pending DNS requests it removes any requests with the
given window reference $n$.
\captionof{algorithm}{\label{alg:cancelnav} Cancel pending requests for given window.}
\begin{algorithmic}[1]
  \Function{$\mathsf{CANCELNAV}$}{$n$, $s'$}
    \State \textbf{remove all} $\an{n, \mi{req}, \mi{key}, \mi{f}}$ \textbf{ from } $\comp{s'}{pendingRequests}$ \textbf{for any} $\mi{req}$, $\mi{key}$, $\mi{f}$
    \State \textbf{remove all} $\an{x, \an{n, \mi{message}, \mi{protocol}}}$ \textbf{ from } $\comp{s'}{pendingDNS}$ \textbf{for any} $\mi{x}$, $\mi{message}$, $\mi{protocol}$
    \State \Return{$s'$}
  \EndFunction
\end{algorithmic} \setlength{\parindent}{1em}

The following function takes an HTTP request
$\mi{message}$ as input, adds cookie and origin headers to
the message, creates a DNS request for the hostname given
in the request and stores the request in
$\comp{s'}{pendingDNS}$ until the DNS resolution
finishes. For normal HTTP requests, $\mi{reference}$ is a
window reference. For \xhrs, $\mi{reference}$ is a value of
the form $\an{\mi{document}, \mi{nonce}}$ where
$\mi{document}$ is a document reference and $\mi{nonce}$ is
some nonce that was chosen by the script that initiated the
request. $\mi{protocol}$ is either $\http$ or
$\https$. $\mi{origin}$ is the origin header value that is
to be added to the HTTP request.

\captionof{algorithm}{\label{alg:send} Prepare headers, do DNS resolution, save message. }
\begin{algorithmic}[1]
  \Function{$\mathsf{SEND}$}{$\mi{reference}$, $\mi{message}$, $\mi{protocol}$, $\mi{origin}$, $s'$}
    \If{$\comp{\mi{message}}{host} \inPairing \comp{s'}{sts}$}
      \Let{$\mi{protocol}$}{$\https$}
    \EndIf
    \Let{ $\mi{cookies}$}{$\langle\{\an{\comp{c}{name}, \comp{\comp{c}{content}}{value}} | c\inPairing \comp{s'}{cookies}\left[\comp{\mi{message}}{host}\right]$} \label{line:assemble-cookies-for-request}
        \Statex[6] $\wedge \left(\comp{\comp{c}{content}}{secure} \implies \left(\mi{protocol} = \https\right)\right) \}\rangle$ \label{line:cookie-rules-http}
    \Let{$\comp{\mi{message}}{headers}[\str{Cookie}]$}{$\mi{cookies}$}
    \If{$\mi{origin} \not\equiv \bot$}
      \Let{$\comp{\mi{message}}{headers}[\str{Origin}]$}{$\mi{origin}$}
    \EndIf
    \Let{$n, s'$}{\textsf{TAKENONCE}$(s')$} 
    \Let{$\comp{s'}{pendingDNS}[n]$}{$\an{\mi{reference},
        \mi{message}, \mi{protocol}}$} \label{line:add-to-pendingdns}
    \State \textbf{stop} $\{(\comp{s'}{DNSaddress} : a :
    \an{\cDNSresolve, \mi{host}, n})\}$, $s'$
  \EndFunction
\end{algorithmic} \setlength{\parindent}{1em}
\noindent

The following two functions have informally been described
in Section~\ref{sec:browserrelation}.

The function $\mathsf{RUNSCRIPT}$ performs a script
execution step of the script in the document
$\compn{s'}{\ptr{d}}$ (which is part of the window
$\compn{s'}{\ptr{w}}$). A new script and document state is
chosen according to the relation defined by the script and
the new script and document state is saved. Afterwards, the
$\mi{command}$ that the script issued is interpreted.  Note
that \textbf{for each} (Line~\ref{line:for-each}) works in
a non-deterministic order.

\captionof{algorithm}{\label{alg:runscript} Execute a script.}
\begin{algorithmic}[1]
  \Function{$\mathsf{RUNSCRIPT}$}{$\ptr{w}$, $\ptr{d}$, $s'$}
    \Let {$n$, $s'$}{\textsf{TAKENONCE}$(s')$}
    \Let{$\mi{tree}$}{$\mathsf{Clean}(s', \compn{s'}{\ptr{d}})$} \label{line:clean-tree}

    \Let{ $\mi{cookies}$}{$\langle\{\an{\comp{c}{name}, \comp{\comp{c}{content}}{value}} | c \inPairing \comp{s'}{cookies}\left[  \comp{\comp{\compn{s'}{\ptr{d}}}{origin}}{host}  \right]$}
        \Statex[6] $\wedge\,\comp{\comp{c}{content}}{httpOnly} = \bot$
        \Statex[6] $\wedge\,\left(\comp{\comp{c}{content}}{secure} \implies \left(\comp{\comp{\compn{s'}{\ptr{d}}}{origin}}{protocol} \equiv \https\right)\right) \}\rangle$ \label{line:assemble-cookies-for-script}
    \LetND{$\mi{tlw}$}{$\comp{s'}{windows}$ \textbf{such that} $\mi{tlw}$ is the top-level window containing $\mi{d}$} 
    \Let{$\mi{sessionStorage}$}{$\comp{s'}{sessionStorage}\left[\an{\comp{\compn{s'}{\ptr{d}}}{origin}, \comp{\mi{tlw}}{nonce}}\right]$}
    \Let{$\mi{localStorage}$}{$\comp{s'}{localStorage}\left[\comp{\compn{s'}{\ptr{d}}}{origin}\right]$}
    \Let{$\mi{secret}$}{$\comp{s'}{secrets}\left[\comp{\compn{s'}{\ptr{d}}}{origin}\right]$} \label{line:browser-secrets}
    \State \textbf{let} $\mi{nonces}$ be an infinite subset of $\left\{x \middle| x \in N^p \wedge x \not\inPairing \comp{s'}{nonces} \right\}$
    \LetND{$R$}{$\mathsf{script}^{-1}(\comp{\compn{s'}{\ptr{d}}}{script})$}
    \LetBreak{$\mi{in}\!$}{$\!\langle\mi{tree}, \comp{\compn{s'\!}{\ptr{d}}}{nonce}, \comp{\compn{s'\!}{\ptr{d}}}{scriptstate}, \comp{\compn{s'\!}{\ptr{d}}}{scriptinput}, \mi{cookies}, \mi{localStorage}, \mi{sessionStorage}$ $\mi{secret}\rangle$}
    \LetND{$\mi{state}'$}{$\terms$,
      $\mi{cookies}' \gets \mathsf{Cookies}$,
      $\mi{localStorage}' \gets \terms$, $\mi{command}
      \gets \terms$, $\mi{out} := \an{\mi{state}', \mi{cookies}', \mi{localStorage}',$ $\mi{sessionStorage}', \mi{command}}$
      \textbf{such that} $((\mi{in}, \mi{nonces}), \mi{out}) \in R$}  \label{line:trigger-script}
    \For{\textbf{each} $n \in d_{\nonces}(\an{\mi{in}, \mi{out}}) \cap N^p$} \label{line:for-each}
      \Append{$n$}{$\comp{s'}{nonces}$}
    \EndFor
    \Let{$\comp{s'}{cookies}\left[\comp{\comp{\compn{s'}{\ptr{d}}}{origin}}{host}\right]$}{$\langle\mathsf{CookieMerge}(\comp{s'}{cookies}\left[\comp{\comp{\compn{s'}{\ptr{d}}}{origin}}{host}\right],$ $\mi{cookies}')\rangle$} \label{line:cookiemerge}
    \Let{$\comp{s'}{localStorage}\left[\comp{\compn{s'}{\ptr{d}}}{origin}\right]$}{$\mi{localStorage}'$}
    \Let{$\comp{s'}{sessionStorage}\left[\an{\comp{\compn{s'}{\ptr{d}}}{origin}, \comp{\mi{tlw}}{nonce}}\right]$}{$\mi{sessionStorage}'$}
    \Let{$\comp{\compn{s'}{\ptr{d}}}{scriptstate}$}{$state'$}
    \Switch{$\mi{command}$}
      \Case{$\an{}$}
        \State \textbf{stop} $\{\}$, $s'$
      \EndCase
      \Case{$\an{\tHref, \mi{url},
          \mi{hrefwindow}}$}\footnote{See the definition of
        URLs in Appendix~\ref{sec:urls}.}
      \Let{$\ptr{w}'$,
        $s'$}{$\mathsf{GETNAVIGABLEWINDOW}$($\ptr{w}$,
        $\mi{hrefwindow}$, $s'$)} \Let{$\mi{req}$}{$\hreq{
          method=\mGet, host=\comp{\mi{url}}{host},
          nonce=n, path=\comp{\mi{url}}{path},
          headers=\an{},
          parameters=\comp{\mi{url}}{params}, body=\an{}
        }$}
      \Let{$s'$}{$\mathsf{CANCELNAV}(\comp{\compn{s'}{\ptr{w}'}}{nonce}, s')$}
      \State \textsf{SEND}($\comp{\compn{s'}{\ptr{w}'}}{nonce}$, $\mi{req}$, $\comp{\mi{url}}{protocol}$, $\bot$, $s'$) \label{line:send-href}
      \EndCase
      \Case{$\an{\tIframe, \mi{url}, \mi{window}}$}
        \Let{$\ptr{w}'$, $s'$}{$\mathsf{GETWINDOW}(\ptr{w}, \mi{window}, s')$}
        \Let{$\mi{req}$}{$\hreq{
            method=\mGet,
            host=\comp{\mi{url}}{host},
            nonce=n,
            path=\comp{\mi{url}}{path},
            headers=\an{},
            parameters=\comp{\mi{url}}{params},
            body=\an{}
          }$}
        \Let {$n$, $s'$}{\textsf{TAKENONCE}$(s')$}
        \Let{$w'$}{$\an{n, \an{}, \bot}$}
        \Append{$w'$}{$\comp{\comp{\compn{s'}{\ptr{w}'}}{activedocument}}{subwindows}$}
        \State \textsf{SEND}(n, $\mi{req}$, $\comp{\mi{url}}{protocol}$, $\bot$, $s'$) \label{line:send-iframe}
      \EndCase
      \Case{$\an{\tForm, \mi{url}, \mi{method}, \mi{data}, \mi{hrefwindow}}$}
        \If{$\mi{method} \not\in \{\mGet, \mPost\}$} \footnote{The working draft for HTML5 allowed for DELETE and PUT methods in HTML5 forms. However, these have since been removed. See \url{http://www.w3.org/TR/2010/WD-html5-diff-20101019/\#changes-2010-06-24}.}
          \State \textbf{stop} $\{\}$, $s'$
        \EndIf
        \Let{$\ptr{w}'$, $s'$}{$\mathsf{GETNAVIGABLEWINDOW}$($\ptr{w}$, $\mi{hrefwindow}$, $s'$)}
        \If{$\mi{method} = \mGet$}
          \Let{$\mi{body}$}{$\an{}$}
          \Let{$\mi{params}$}{$\mi{data}$}
          \Let{$\mi{origin}$}{$\bot$}
        \Else
          \Let{$\mi{body}$}{$\mi{data}$}
          \Let{$\mi{params}$}{$\comp{\mi{url}}{params}$}
          \Let{$\mi{origin}$}{$\comp{\compn{s'}{\ptr{d}}}{origin}$}
        \EndIf
        \Let{$\mi{req}$}{$\hreq{
            method=\mi{method},
            host=\comp{\mi{url}}{host},
            nonce=n,
            path=\comp{\mi{url}}{path},
            headers=\an{},
            parameters=\mi{params},
            xbody=\mi{body}
          }$}
        \Let{$s'$}{$\mathsf{CANCELNAV}(\comp{\compn{s'}{\ptr{w}'}}{nonce}, s')$}
        \State \textsf{SEND}($\comp{\compn{s'}{\ptr{w}'}}{nonce}$, $\mi{req}$, $\comp{\mi{url}}{protocol}$, $\mi{origin}$, $s'$) \label{line:send-form}
      \EndCase
      \Case{$\an{\tSetScript, \mi{window}, \mi{script}}$}
        \Let{$\ptr{w}'$, $s'$}{$\mathsf{GETWINDOW}(\ptr{w}, \mi{window}, s')$}
        \Let{$\comp{\comp{\compn{s'}{\ptr{w}'}}{activedocument}}{script}$}{$\mi{script}$}
        \State \textbf{stop} $\{\}$, $s'$
      \EndCase
      \Case{$\an{\tSetScriptState, \mi{window}, \mi{scriptstate}}$}
        \Let{$\ptr{w}'$, $s'$}{$\mathsf{GETWINDOW}(\ptr{w}, \mi{window}, s')$}
        \Let{$\comp{\comp{\compn{s'}{\ptr{w}'}}{activedocument}}{scriptstate}$}{$\mi{scriptstate}$}
        \State \textbf{stop} $\{\}$, $s'$
      \EndCase
      \Case{$\an{\tXMLHTTPRequest, \mi{url}, \mi{method}, \mi{data}, \mi{xhrreference}}$}
        \If{$\mi{method} \in \{\mConnect, \mTrace, \mTrack\}$} \footnote{According to W3C XMLHTTPRequest definition.}
          \State \textbf{stop} $\{\}$, $s'$
        \EndIf
        \If{$\comp{\mi{url}}{host} \not\equiv \comp{\comp{\compn{s'}{\ptr{d}}}{origin}}{host}$ $\vee$ $\comp{\mi{url}}{protocol} \not\equiv \comp{\comp{\compn{s'}{\ptr{d}}}{origin}}{protocol}$} \footnote{We only allow same origin requests for now.}
          \State \textbf{stop} $\{\}$, $s'$
        \EndIf
        \If{$\mi{method} \in \{\mGet, \mHead\}$}
          \Let{$\mi{data}$}{$\an{}$}
          \Let{$\mi{origin}$}{$\bot$}
        \Else
          \Let{$\mi{origin}$}{$\comp{\compn{s'}{\ptr{d}}}{origin}$}
        \EndIf
        \Let{$\mi{req}$}{$\hreq{
            method=\mi{method},
            host=\comp{\mi{url}}{host},
            nonce=n,
            path=\comp{\mi{url}}{path},
            headers={},
            parameters=\comp{\mi{url}}{params},
            xbody=\mi{data}
          }$}
        \State \textsf{SEND}($\an{\comp{\compn{s'}{\ptr{d}}}{nonce}, \mi{xhrreference}}$, $\mi{req}$, $\comp{\mi{url}}{protocol}$, $\mi{origin}$, $s'$)\label{line:send-xhr}
      \EndCase
      \Case{$\an{\tBack, \mi{window}}$} \footnote{Note that navigating a window using the back/forward buttons does not trigger a reload of the affected documents. While real world browser may chose to refresh a document in this case, we assume that the complete state of a previously viewed document is restored.}
        \Let{$\ptr{w}'$, $s'$}{$\mathsf{GETNAVIGABLEWINDOW}$($\ptr{w}$, $\mi{window}$, $s'$)}
        \If{$\exists\, \ptr{j} \in \mathbb{N}, \ptr{j} >
          1$ \textbf{such that} $\comp{\compn{\comp{\compn{s'}{\ptr{w'}}}{documents}}{\ptr{j}}}{active} \equiv \True$}
          \Let{$\comp{\compn{\comp{\compn{s'}{\ptr{w'}}}{documents}}{\ptr{j}}}{active}$}{$\bot$}
          \Let{$\comp{\compn{\comp{\compn{s'}{\ptr{w'}}}{documents}}{(\ptr{j}-1)}}{active}$}{$\True$}
          \Let{$s'$}{$\mathsf{CANCELNAV}(\comp{\compn{s'}{\ptr{w}'}}{nonce}, s')$}
        \EndIf
        \State \textbf{stop} $\{\}$, $s'$
      \EndCase
      \Case{$\an{\tForward, \mi{window}}$}
        \Let{$\ptr{w}'$, $s'$}{$\mathsf{GETNAVIGABLEWINDOW}$($\ptr{w}$, $\mi{window}$, $s'$)}
        \If{$\exists\, \ptr{j} \in \mathbb{N} $ \textbf{such that} $\comp{\compn{\comp{\compn{s'}{\ptr{w'}}}{documents}}{\ptr{j}}}{active} \equiv \True$ $\wedge$ \breakalgo{5} $\compn{\comp{\compn{s'}{\ptr{w'}}}{documents}}{(\ptr{j}+1)} \in \mathsf{Documents}$}
          \Let{$\comp{\compn{\comp{\compn{s'}{\ptr{w'}}}{documents}}{\ptr{j}}}{active}$}{$\bot$}
          \Let{$\comp{\compn{\comp{\compn{s'}{\ptr{w'}}}{documents}}{(\ptr{j}+1)}}{active}$}{$\True$}
          \Let{$s'$}{$\mathsf{CANCELNAV}(\comp{\compn{s'}{\ptr{w}'}}{nonce}, s')$}
        \EndIf
        \State \textbf{stop} $\{\}$, $s'$
      \EndCase
      \Case{$\an{\tClose, \mi{window}}$}
        \Let{$\ptr{w}'$, $s'$}{$\mathsf{GETNAVIGABLEWINDOW}$($\ptr{w}$, $\mi{window}$, $s'$)}
        \State \textbf{remove} $\compn{s'}{\ptr{w'}}$ from the sequence containing it 
        \State \textbf{stop} $\{\}$, $s'$
      \EndCase

      \Case{$\an{\tPostMessage, \mi{window}, \mi{message}, \mi{origin}}$}
        \LetND{$\ptr{w}'$}{$\mathsf{Subwindows}(s')$ \textbf{such that} $\comp{\compn{s'}{\ptr{w}'}}{nonce} \equiv \mi{window}$}
        \If{$\exists \ptr{j} \in \mathbb{N}$ \textbf{such that} $\comp{\compn{\comp{\compn{s'}{\ptr{w'}}}{documents}}{\ptr{j}}}{active} \equiv \True$ $\wedge \breakalgo{5} (\mi{origin} \not\equiv \bot \implies \comp{\compn{\comp{\compn{s'}{\ptr{w'}}}{documents}}{\ptr{j}}}{origin} \equiv \mi{origin})$}    
            \Append{\breakalgo{5}$\an{\tPostMessage, \comp{\compn{s'}{\ptr{w}}}{nonce}, \comp{\compn{s'}{\ptr{d}}}{origin}, \mi{message}}$}{$\comp{\compn{\comp{\compn{s'}{\ptr{w'}}}{documents}}{\ptr{j}}}{scriptinput}$}
        \EndIf
      \EndCase
    \EndSwitch
  \EndFunction
\end{algorithmic} \setlength{\parindent}{1em}

The function $\mathsf{PROCESSRESPONSE}$ is responsible for
processing an HTTP response ($\mi{response}$) that was
received as the response to a request ($\mi{request}$) that
was sent earlier. In $\mi{reference}$, either a window or a
document reference is given (see explanation for
Algorithm~\ref{alg:send} above). Again, $\mi{protocol}$ is
either $\http$ or $\https$.

The function first saves any cookies that were contained in
the response to the browser state, then checks whether a
redirection is requested (Location header). If that is not
the case, the function creates a new document (for normal
requests) or delivers the contents of the response to the
respective receiver (for \xhr responses).
\captionof{algorithm}{\label{alg:processresponse} Process an HTTP response.}
\begin{algorithmic}[1]
\Function{$\mathsf{PROCESSRESPONSE}$}{$\mi{response}$, $\mi{reference}$, $\mi{request}$, $\mi{protocol}$, $s'$}
  \Let{$n, s'$}{\textsf{TAKENONCE}$(s')$} 
  \If{$\mathtt{Set{\mhyphen}Cookie} \in
    \comp{\mi{response}}{headers}$}
    \For{\textbf{each} $c \inPairing \comp{\mi{response}}{headers}\left[\mathtt{Set{\mhyphen}Cookie}\right]$, $c \in \mathsf{Cookies}$}
      \Let{$\comp{s'}{cookies}\left[\comp{\comp{\mi{request}}{url}}{host}\right]$}{$\mathsf{AddCookie}(\comp{s'}{cookies}\left[\comp{\comp{\mi{request}}{url}}{host}\right], c)$} \label{line:set-cookie}
    \EndFor
  \EndIf  
  \If{$\mathtt{Strict{\mhyphen}Transport{\mhyphen}Security} \in \comp{\mi{response}}{headers}$ $\wedge$ $\mi{protocol} \equiv \https$}
    \Append{$\comp{\mi{request}}{host}$}{$\comp{s'}{sts}$}
  \EndIf
  \If{$\mathtt{Location} \in \comp{\mi{response}}{headers} \wedge \comp{\mi{response}}{status} \in \{303, 307\}$} \label{line:location-header} \footnote{The RFC for HTTPbis (currently in draft status), which obsoletes RFC 2616, does not specify whether a POST/DELETE/etc. request that was answered with a status code of 301 or 302 should be rewritten to a GET request or not (``for historic reasons'' that are detailed in Section~7.4.). 
As the specification is clear for the status codes 303 and 307 (and most browsers actually follow the specification in this regard), we focus on modeling these.}
    \Let{$\mi{url}$}{$\comp{\mi{response}}{headers}\left[\mathtt{Location}\right]$}
    \Let{$\mi{method}'$}{$\comp{\mi{request}}{method}$} \footnote{While the standard demands that users confirm redirections of non-safe-methods (e.g., POST), we assume that users generally confirm these redirections.}
    \Let{$\mi{body}'$}{$\comp{\mi{request}}{body}$} \footnote{If, for example, a GET request is redirected and the original request contained a body, this body is preserved, as HTTP allows for payloads in messages with all HTTP methods, except for the TRACE method (a detail which we omit). 
Browsers will usually not send body payloads for methods that do not specify semantics for such data in the first place.}
    \If{$\str{Origin} \in \comp{request}{headers}$}
      \Let{$\mi{origin}$}{$\an{\comp{request}{headers}[\str{Origin}], \an{\comp{request}{host}, \mi{protocol}}}$}
    \Else
      \Let{$\mi{origin}$}{$\bot$}
    \EndIf
    \If{$\comp{\mi{response}}{status} \equiv 303 \wedge \comp{\mi{request}}{method} \not\in \{\mGet, \mHead\}$}
      \Let {$\mi{method}'$}{$\mGet$}
      \Let{$\mi{body}'$}{$\an{}$}
    \EndIf
    \If{$\nexists\, \ptr{w} \in \mathsf{Subwindows}(s')$ \textbf{such that} $\comp{\compn{s'}{\ptr{w}}}{nonce} \equiv \mi{reference}$} \Comment{Do not redirect XHRs.}
      \State \textbf{stop} $\{\}$, $s$
    \EndIf
    \Let{$\mi{req}$}{$\hreq{
            method=\mi{method'},
            host=\comp{\mi{url}}{host},
            nonce=n,
            path=\comp{\mi{url}}{path},
            headers=\an{},
            parameters=\comp{\mi{url}}{params},
            xbody=\mi{body}'
          }$}
    \State \textsf{SEND}($\mi{reference}$, $\mi{req}$, $\comp{\mi{url}}{protocol}$, $\mi{origin}$, $s'$)\label{line:send-redirect}
  \EndIf

  \If{$\exists\, \ptr{w} \in \mathsf{Subwindows}(s')$ \textbf{such that} $\comp{\compn{s'}{\ptr{w}}}{nonce} \equiv \mi{reference}$} \Comment{normal response}
    \Let{$\mi{script}$}{$\proj{1}{\comp{\mi{response}}{body}}$}
    \Let{$\mi{scriptstate}$}{$\proj{2}{\comp{\mi{response}}{body}}$}
    \Let{$d$}{$\an{n, \an{\comp{\mi{request}}{host}, \comp{\mi{request}}{protocol}}, \mi{script}, \mi{scriptstate}, \an{}, \an{}, \True}$} \label{line:take-script} \label{line:set-origin-of-document}
    \If{$\comp{\compn{s'}{\ptr{w}}}{documents} \equiv \an{}$}
      \Let{$\comp{\compn{s'}{\ptr{w}}}{documents}$}{$\an{d}$}
    \Else
      \LetND{$\ptr{i}$}{$\mathbb{N}$ \textbf{such that} $\comp{\compn{\comp{\compn{s'}{\ptr{w}}}{documents}}{\ptr{i}}}{active} \equiv \True$}
      \Let{$\comp{\compn{\comp{\compn{s'}{\ptr{w}}}{documents}}{\ptr{i}}}{active}$}{$\bot$}
      \State \textbf{remove} $\compn{\comp{\compn{s'}{\ptr{w}}}{documents}}{(\ptr{i}+1)}$ and all following documents from $\comp{\compn{s'}{\ptr{w}}}{documents}$
      \Append{$d$}{$\comp{\compn{s'}{\ptr{w}}}{documents}$}
    \EndIf
    \State \textbf{stop} $\{\}$, $s'$
  \ElsIf{$\exists\, \ptr{w} \in \mathsf{Subwindows}(s')$, $\ptr{d}$ \textbf{such that} $\comp{\compn{s'}{\ptr{d}}}{nonce} \equiv \proj{1}{\mi{reference}} \wedge \breakalgo{3} \compn{s'}{\ptr{d}} = \comp{\compn{s'}{\ptr{w}}}{activedocument}$} \label{line:process-xhr-response} \Comment{process XHR response}
    \Let{$\comp{\compn{s'}{\ptr{d}}}{scriptinput}$}{$\comp{\compn{s'}{\ptr{d}}}{scriptinput} {\plusPairing} \an{\tXMLHTTPRequest, \comp{\mi{response}}{body}, \proj{2}{\mi{reference}}}$}
  \EndIf
\EndFunction
\end{algorithmic} \setlength{\parindent}{1em}

\subsubsection{Main Algorithm}\label{app:mainalgorithmwebbrowserprocess}
This is the main algorithm of the browser relation. It was
already presented informally in
Section~\ref{sec:web-browsers} and follows the structure
presented there. It receives the message $m$ as input, as
well as $a$, $f$ and $s$ as above.

\captionof{algorithm}{\label{alg:browsermain} Main Algorithm}
\begin{algorithmic}[1]
\Statex[-1] \textbf{Input:} $(a{:}f{:}m),s$
  \Let{$s'$}{$s$}

  \If{$\comp{s}{isCorrupted} \equiv \fullcorrupt$} 
    \Let{$\comp{s'}{pendingRequests}$}{$\an{m, \comp{s}{pendingRequests}}$} \Comment{Collect incoming messages}
    \LetND{$m'$}{$d_{N^p}(s')$}
    \LetND{$a'$}{$\addresses$}
    \State \textbf{stop} $\{(a'{:}a{:}m')\}$, $s'$
  \ElsIf{$\comp{s}{isCorrupted} \equiv \closecorrupt$}
    \Let{$\comp{s'}{pendingRequests}$}{$\an{m, \comp{s}{pendingRequests}}$} \Comment{Collect incoming messages}
    \Let{$N^\text{clean}$}{$N^p \setminus \{n | n \inPairing \comp{s}{nonces}\}$} \label{line:key-not-used-anymore}
    \LetND{$m'$}{$d_{N^\text{clean}}(s')$}
    \LetND{$a'$}{$\addresses$}
    \Let{$\comp{s'}{nonces}$}{$\comp{s}{nonces}$}
    \State \textbf{stop} $\{(a'{:}a{:}m')\}$, $s'$
  \EndIf
  \Let {$n$, $s'$}{\textsf{TAKENONCE}$(s')$}
  \If{$m \equiv \trigger$} \Comment{A special trigger message. }
    \LetND{$\mi{switch}$}{$\{1,2\}$}
    \If{$\mi{switch} \equiv 1$} \Comment{Run some script.}
      \LetNDST{$\ptr{w}$}{$\mathsf{Subwindows}(s')$}{$\comp{\compn{s'}{\ptr{w}}}{documents} \neq \an{}$}{\textbf{stop} $\{\}$, $s'$}
      \Let{$\ptr{d}$}{$\ptr{w} \plusPairing \str{activedocument}$}
      \State \textsf{RUNSCRIPT}($\ptr{w}$, $\ptr{d}$, $s'$)
    \ElsIf{$\mi{switch} \equiv 2$} \Comment{Create some new request.}
      \Let{$w'$}{$\an{n, \an{}, \bot}$}
      \Append{$w'$}{$\comp{s'}{windows}$}
      \LetND{$\mi{protocol}$}{$\{\http, \https\}$}
      \LetND{$\mi{host}$}{$\dns$}
      \LetND{$\mi{path}$}{$\mathbb{S}$}
      \LetND{$\mi{parameters}$}{$\dict{\mathbb{S}}{\mathbb{S}}$}
      \Let {$n'$, $s'$}{\textsf{TAKENONCE}$(s')$}
      \Let{$\mi{req}$}{$\hreq{
          method=\mGet,
          host=\mi{host},
          nonce=n',
          path=\mi{path},
          headers=\an{},
          parameters=\mi{parameters},
          body=\an{}
        }$}
      \State \textsf{SEND}($n$, $\mi{req}$, $\mi{protocol}$, $\bot$, $s'$)\label{line:send-random}
    \EndIf
  \ElsIf{$m \equiv \fullcorrupt$} \Comment{Request to corrupt browser}
    \Let{$\comp{s'}{isCorrupted}$}{$\fullcorrupt$}
    \State \textbf{stop} $\{\}$, $s'$
  \ElsIf{$m \equiv \closecorrupt$} \Comment{Close the browser}
    \Let{$\comp{s'}{secrets}$}{$\an{}$}  
    \Let{$\comp{s'}{windows}$}{$\an{}$}
    \Let{$\comp{s'}{pendingDNS}$}{$\an{}$}
    \Let{$\comp{s'}{pendingRequests}$}{$\an{}$}
    \Let{$\comp{s'}{sessionStorage}$}{$\an{}$}
    \State \textbf{let} $\comp{s'}{cookies} \subsetPairing \cookies$ \textbf{such that} \breakalgo{2} $(c \inPairing \comp{s'}{cookies}) {\iff} (c \inPairing \comp{s}{cookies} \wedge \comp{\comp{c}{content}}{session} \equiv \bot$)
    \Let{$\comp{s'}{isCorrupted}$}{$\closecorrupt$}
    \State \textbf{stop} $\{\}$, $s'$
  \ElsIf{$\exists\, \an{\mi{reference}, \mi{request}, \mi{key}, f}$
      $\inPairing \comp{s'}{pendingRequests}$
      \textbf{such that} \hfill{} \breakalgo{2} $\proj{1}{\decs{m}{\mi{key}}} \equiv \cHttpResp$ }
    \Comment{Encrypted HTTP response}
    \Let{$m'$}{$\decs{m}{\mi{key}}$}
    \If{$\comp{m'}{nonce} \not\equiv \comp{\mi{request}}{nonce}$}
      \State \textbf{stop} $\{\}$, $s$
    \EndIf
    \State \textbf{remove} $\an{\mi{reference}, \mi{request}, \mi{key}, f}$ \textbf{from} $\comp{s'}{pendingRequests}$
    \State \textsf{PROCESSRESPONSE}($m'$, $\mi{reference}$, $\mi{request}$, $\https$, $s'$)
  \ElsIf{$\proj{1}{m} \equiv \cHttpResp$ $\wedge$ $\exists\, \an{\mi{reference}, \mi{request}, \bot, f}$ $\inPairing \comp{s'}{pendingRequests}$ \textbf{such that} \breakalgo{2} $\comp{m'}{nonce} \equiv \comp{\mi{request}}{key}$ }
    \State \textbf{remove} $\an{\mi{reference}, \mi{request}, \bot, f}$ \textbf{from} $\comp{s'}{pendingRequests}$
    \State \textsf{PROCESSRESPONSE}($m$, $\mi{reference}$, $\mi{request}$, $\http$, $s'$)
  \ElsIf{$m \in \dnsresponses$} \Comment{Successful DNS response}
      \If{$\comp{m}{nonce} \not\in \comp{s}{pendingDNS}$}
        \State \textbf{stop} $\{\}$, $s$
      \EndIf
      \Let{$\an{\mi{reference}, \mi{message}, \mi{protocol}}$}{$\comp{s}{pendingDNS}[\comp{m}{nonce}]$}
      \If{$\mi{protocol} \equiv \https$}
        \Let{$k, s'$}{\textsf{TAKENONCE}($s'$)} \label{line:takenonce-k}
        \Append{$\langle\mi{reference}$, $\mi{message}$, $\mi{k}$, $\comp{m}{result}\rangle$}{$\comp{s'}{pendingRequests}$} \label{line:add-to-pendingrequests-https}
        \Let{$\mi{message}$}{$\enc{\an{\mi{message},\mi{k}}}{\comp{s'}{keyMapping}\left[\comp{\mi{message}}{host}\right]}$} \label{line:select-enc-key}
      \Else
        \Append{$\langle\mi{reference}$, $\mi{message}$, $\bot$, $\comp{m}{result}\rangle$}{$\comp{s'}{pendingRequests}$} \label{line:add-to-pendingrequests}
      \EndIf
      \Let{$\comp{s'}{pendingDNS}$}{$\comp{s'}{pendingDNS} - \comp{m}{nonce}$}
      \State \textbf{stop} $\{(\comp{m}{result}{:}a{:}\mi{message})\}$, $s'$
  \Else
    \State \textbf{stop} $\{\}$, $s$
  \EndIf
\end{algorithmic} \setlength{\parindent}{1em}

\section{DNS Servers}\label{app:dns-servers}

An informal description of a DNS server is given in
Section~\ref{sec:DNSservers}. We now provide a formal definition of a
DNS server $d$. A DNS server $d$ is an atomic DY process $(I^d,
\{s^d_0\}, R^d, s^d_0, N^d)$ with a finite set of addresses $I^d$. The
(initial) state $s^d_0 \in [\mathbb{S} \times \addresses]$ is a
mapping of a finite set of domain names to addresses.  We note that
our definition of DNS servers does not change the mapping from domain
names to addresses at all. The set of states of the atomic DY process
of a DNS server therefore contains only $s^d_0$.

We now specify the relation $R^d \subseteq (\events \times \{s^d_0\} )
\times (2^\events \times \{s^d_0\})$ of $d$.  Just like in
Appendix~\ref{sec:descr-web-brows}, we describe this relation by a
non-deterministic algorithm.

\captionof{algorithm}{\label{alg:dns} Relation of a DNS server $R^d$}
\begin{algorithmic}[1]
 \Statex[-1] \textbf{Input:} $(a{:}f{:}m),s$

  \LetST{$\mi{domain},n$}{$\an{\cDNSresolve,\mi{domain},n} \equiv m$}{\textbf{stop} $\{\},s$}
  \If{$\mi{domain} \in s$}
    \Let{$\mi{addr}$}{$s[\mi{domain}]$}
    \Let{$m'$}{$\an{\cDNSresolved,\mi{addr},n}$}
    \State \textbf{stop} $\{(f : a : m')\},s$
  \EndIf
  \State \textbf{stop} $\{\},s$

\end{algorithmic}

\section{Step-By-Step Description of BrowserID (Primary IdP)}\label{app:browserid-lowlevel}

We now present additional details of the JavaScript
implementation of BrowserID\@. While the basic steps have
been shown in Section~\ref{sec:javascript-descr}, we will
now again refer to Figure~\ref{fig:browserid-lowlevel-ld}
and provide a step-by-step description. As above, we focus
on the main login flow without the CIF, and we leave out
steps for fetching additional resources (like JavaScript
files) and some less relevant postMessages and
\xhrs.

We (again) assume that the user uses a ``fresh'' browser,
i.e., the user has not been logged in before. The user has
already opened a document of some RP (RP-Doc) in her
browser. RP-Doc includes a JavaScript file, which provides
the BrowserID API. The user is now about to click on a
login button in order to start a BrowserID login.

\myparagraph{Phase~\refprotophase{ld-start-1}.} After the
user has clicked on the login button, RP-Doc opens a new
browser window, the \emph{login dialog}
(LD)~\refprotostep{ld-open}. The document of LD is loaded
from LPO~\refprotostep{ld-init-1}. Now, LD sends a
\emph{ready} postMessage~\refprotostep{ld-rpdoc-ready-1} to
its opener, which is RP-Doc. RP-Doc then responds by
sending a \emph{request}
postMessage~\refprotostep{rpdoc-ld-request-1}. This
postMessage may contain additional information like a name
or a logo of RP-Doc.  LD then fetches the so-called
\emph{session context} from LPO
using~\refprotostep{ld-ctx-1}. The session context contains
information about whether the user is already logged in at
LPO, which, by our assumption, is not the case at this
point. The session context also contains an \xsrf
protection token which will be sent in all subsequent POST
requests to LPO\@. Also, an $\str{httpOnly}$ cookie called
\texttt{browserid\_state} is set, which contains an LPO
session identifier.  Now, the user is prompted to enter her
email address (\emph{login email address}), which she wants
to use to log in at RP~\refprotostep{ld-user-email}. LD
sends the login email address to LPO via an
\xhr~\refprotostep{ld-addrinfo-1}, in order to get
information about the IdP the email address belongs to. The
information from this so-called \emph{support document} may
be cached at LPO for further use. LPO extracts the domain
part of the login email address and fetches an information
document~\refprotostep{lpo-idp-wk-1} from a fixed path
(\url{/.well-known/browserid}) at the IdP. This document
contains the public key of IdP, and two paths, the
provisioning path and the authentication path at IdP. These
paths will be used later in the login process by LD. LPO
converts these paths into URLs and sends them in its
response~\refprotostep{ld-addrinfo-1-resp} to the
requesting \xhr~\refprotostep{ld-addrinfo-1}.

\myparagraph{Phase~\refprotophase{ld-prov-1}.} As there is
no record about the login email address in the localStorage under
the origin of LPO, the LD now tries to get a UC for this
identity. For that to happen, the LD creates a new iframe,
the \emph{provisioning iframe}
(PIF)~\refprotostep{ld-pif-open-1}. The PIF's document is
loaded~\refprotostep{pif-init-1} from the provisioning URL
LD has just received before
in~\refprotostep{ld-addrinfo-1-resp}. The PIF now interacts
with the LD via
postMessages~\refprotostep{pif-ld-pms-1}. As the user is
currently not logged in, the PIF tells the LD that the user
is not authenticated yet. This also indicates to the LD
that the PIF has finished operation. The LD then closes the
PIF~\refprotostep{ld-pif-close-1}.

\myparagraph{Phase~\refprotophase{ld-auth}.} Now, the LD
saves the login email address in the localStorage indexed
by a fresh nonce. This nonce is stored in the
sessionStorage to retrieve the email address later from the
localStorage again. Next, the LD navigates itself to the
authentication URL it has received
in~\refprotostep{ld-addrinfo-1-resp}. The loaded document
now interacts with the user and the
IdP~\refprotostep{idp-ld-auth} in order to establish some
authenticated session depending on the actual IdP
implementation, which is out of scope of the BrowserID
standard. For example, during this authentication
procedure, the IdP may issue some session cookie.

\myparagraph{Phase~\refprotophase{ld-start-2}.} After the
authentication to the IdP has been completed, the
authentication document navigates the LD to the LD URL
again. The LD's document is fetched again from LPO and the
login process starts over. The following steps are similar
to Phase~\refprotophase{ld-start-1}: The ready and request
postMessages are exchanged and the session context is
fetched. As the user has not been authenticated to LPO yet,
the session context still contains the same information as
above in~\refprotostep{ld-ctx-1}. Now, the user is not
prompted to enter her email address again. The email
address is fetched from the localStorage under the index of
the nonce stored in the sessionStorage. Now, the address
information is requested again from LPO.

\myparagraph{Phase~\refprotophase{ld-prov-2}.} As there
still is no UC belonging to the login email address in the
localStorage, the PIF is created again. As the user now has
established an authenticated session with the IdP, the PIF
asks the LD to generate a fresh key pair. After the LD has
generated the key pair~\refprotostep{gen-key-pair}, it
stores the key pair in the localStorage (under the origin
of LPO) and sends the public key to the PIF as a
postMessage~\refprotostep{pubkey-ld-pif}. The following
steps \refprotostep{req-uc}--\refprotostep{send-uc} are not
specified in the BrowserID protocol. Typically, the PIF
would send the public key to IdP (via an
\xhr)~\refprotostep{req-uc}. The IdP would create the
UC~\refprotostep{certify-uc} and send it back to the
PIF~\refprotostep{send-uc}.  The PIF then sends the UC to
the LD~\refprotostep{recv-uc}, which stores it in the
localStorage. Now, the LD closes the PIF.

\myparagraph{Phase~\refprotophase{ld-lpo-auth}.} The LD is
now able to create a CAP, as it has access to a UC and the
corresponding private key in its localStorage. First, LD
creates an IA for LPO~\refprotostep{ld-gen-cap-lpo}. The IA
and the UC is then combined to a CAP, which is then sent to
LPO in an \xhr POST message~\refprotostep{ld-lpo-auth}. LPO
is now able to verify this CAP with the public key of IdP,
which LPO has already fetched and cached before
in~\refprotostep{lpo-idp-wk-1}. If the CAP is valid, LPO
considers its session with the user's browser to be
authenticated for the email address the UC in the CAP is
issued for.

\myparagraph{Phase~\refprotophase{ld-cap}.} Now,
in~\refprotostep{ld-lpo-list-emails}, the LD fetches a list
of email addresses, which LPO considers to be owned by the
user. If the login email address would not appear in this
list, LD would abort the login process. After this, the LD
fetches the address information about the login email
address again in~\refprotostep{ld-addrinfo-3}. Using this
information, LD validates if the UC is signed by the
correct party (primary/secondary IdP).  Now, LD generates
an IA for the sender's origin of the request
postMessage~\refprotostep{rpdoc-ld-request-1} (which was
repeated in Phase~\refprotophase{ld-start-2}) using the
private key from the localStorage~\refprotostep{ld-gen-cap} (the IA is generated for
the login email address). Also, it is recorded in the
localStorage that the user is now logged in at RP with this
email address. The LD then combines the IA with the UC
stored in the localStorage to the CAP, which is then sent
to RP-Doc in the \emph{response}
postMessage~\refprotostep{ld-rpdoc-cap}.

This concludes the login process that runs in
LD. Afterwards, RP-Doc closes LD~\refprotostep{ld-close}.

\section{Sideshow/BigTent OpenID Flow}\label{app:sideshow-openid-flow}

We will now give concrete examples of an OpenID flow that
is started when Sideshow or BigTent want to authenticate a
user via OpenID (as presented in
Section~\ref{sec:sideshow}). We will show typical requests
and responses in such a flow and discuss the parameters
that are important for the attacks presented in
Section~\ref{sec:attacks-browserid}. We focus on Sideshow
(and thus, Google), the URLs for BigTent (with Yahoo) are
similar.

\subsection{OpenID Authentication Request}
As we already discussed in Section~\ref{sec:sideshow},
Sideshow maintains a session with the user. Sideshow issues
UCs to a user only if the session is authenticated. This
authentication is done via OpenID. When detecting that a
session is not authenticated, Sideshow redirects the user's
browser to the so-called \emph{OpenID endpoint URL} of
Google/Gmail. This URL may look as follows (line breaks
added for readability):

\begin{verbatim}
https://www.google.com/accounts/o8/ud
 ?openid.mode=checkid_setup
 &openid.ns=http%
 &openid.ns.ax=http%
 &openid.ax.mode=fetch_request
 &openid.ax.type.email=http%
 &openid.ax.required=email
 &openid.ns.ui=http%
 &openid.ui.mode=popup
 &openid.identity=
  http%
 &openid.claimed_id=
  http%
 &openid.return_to=
  https%
 &openid.realm=https%
\end{verbatim}

In this URL, the parameter \texttt{openid.ax.type.email}
encodes that Sideshow requests under the name ``email'' in
the namespace ``ax'' an attribute of the type
\nolinkurl{http://axschema.org/contact/email}. Per
definition of the OpenID attribute exchange
schema,\footnote{See
  \url{http://openid.net/specs/openid-attribute-exchange-1_0.html}}
this denotes the request for an email address. Note that
Google's OpenID endpoint is (per the OpenID protocol) not
obliged to follow this request and may issue an OpenID
assertion without a (signed) email address.

The parameter \texttt{openid.return\_to} contains the URL
to which Google redirects the user after issuing the
assertion. The assertion and possibly other information are
appended to this URL.

\subsection{OpenID Authentication Response}
After accessing the above OpenID endpoint URL, the user
authenticates with Google and confirms that Google releases
the requested information to Sideshow. Google then creates
an OpenID assertion and appends it to the
\texttt{openid.return\_to} URL that was contained in the
OpenID authentication request. Finally, OpenID redirects
the user's browser to the resulting URL, which may look as
follows:

\begin{verbatim}
https://gmail.login.persona.org/authenticate/verify
 ?openid.ns=http%
 &openid.mode=id_res
 &openid.op_endpoint=https%
 &openid.response_nonce=2013-09-24T11%
 &openid.return_to=https%
 &openid.assoc_handle=1.AMlYA(...)ubxCOqB
 &openid.signed=op_endpoint
  %
  %
  %
  %
  %
  %
  %
  %
  %
 &openid.sig=BIPe1PIwitMp365MUEtd34IJLUs%
 &openid.identity=
  https%
 &openid.claimed_id=
  https%
 &openid.ns.ext1=http%
 &openid.ext1.mode=fetch_response
 &openid.ext1.type.email=http%
 &openid.ext1.value.email=user%
 &openid.ns.ext2=http%
 &openid.ext2.mode=popup
\end{verbatim}

First note that the namespace ``ax'' was renamed by Gmail:
What was prefixed with \texttt{openid.ax} in the request is
now prefixed with \texttt{openid.ext1}. The parameter
\texttt{openid.signed} contains the names of the parameters
that have actually been MACed into the signature given in
\texttt{openid.sig}. 

Note that the receiver of the assertion can not know the
exact renaming performed by Gmail and must, although the
renaming is obvious in this case, rely on the type
parameters to determine the actual contents of the
parameters. In this case, \texttt{openid.ext1.type.email}
contains the AX schema type for an email address, saying
that \texttt{openid.ext1.value.email} actually contains the
requested email address.

The parameter \texttt{openid.assoc\_handle} contains the ID
of a temporary symmetric key created and stored at Google
that is used for the MAC. 
\subsection{Verification}
After receiving the above request, Sideshow can forward all
attributes to a verification service at Google. The URL for
this verification is
\nolinkurl{https://www.google.com/accounts/o8/ud?openid.mode=checkid_authentication}
(same as in the authentication request, but with a
different \texttt{openid.mode} value). Sideshow appends to
this URL all attributes from the authentication response,
i.e., the assertion, except for the \texttt{openid.mode}
parameter. Google (only) checks that the MAC in
\texttt{openid.sig} is correct (using the symmetric key
stored earlier) and answers with ``true'' or ``false''
accordingly.

\section{Step-By-Step Description of BrowserID (Secondary IdP)}\label{app:browserid-sidp-lowlevel}

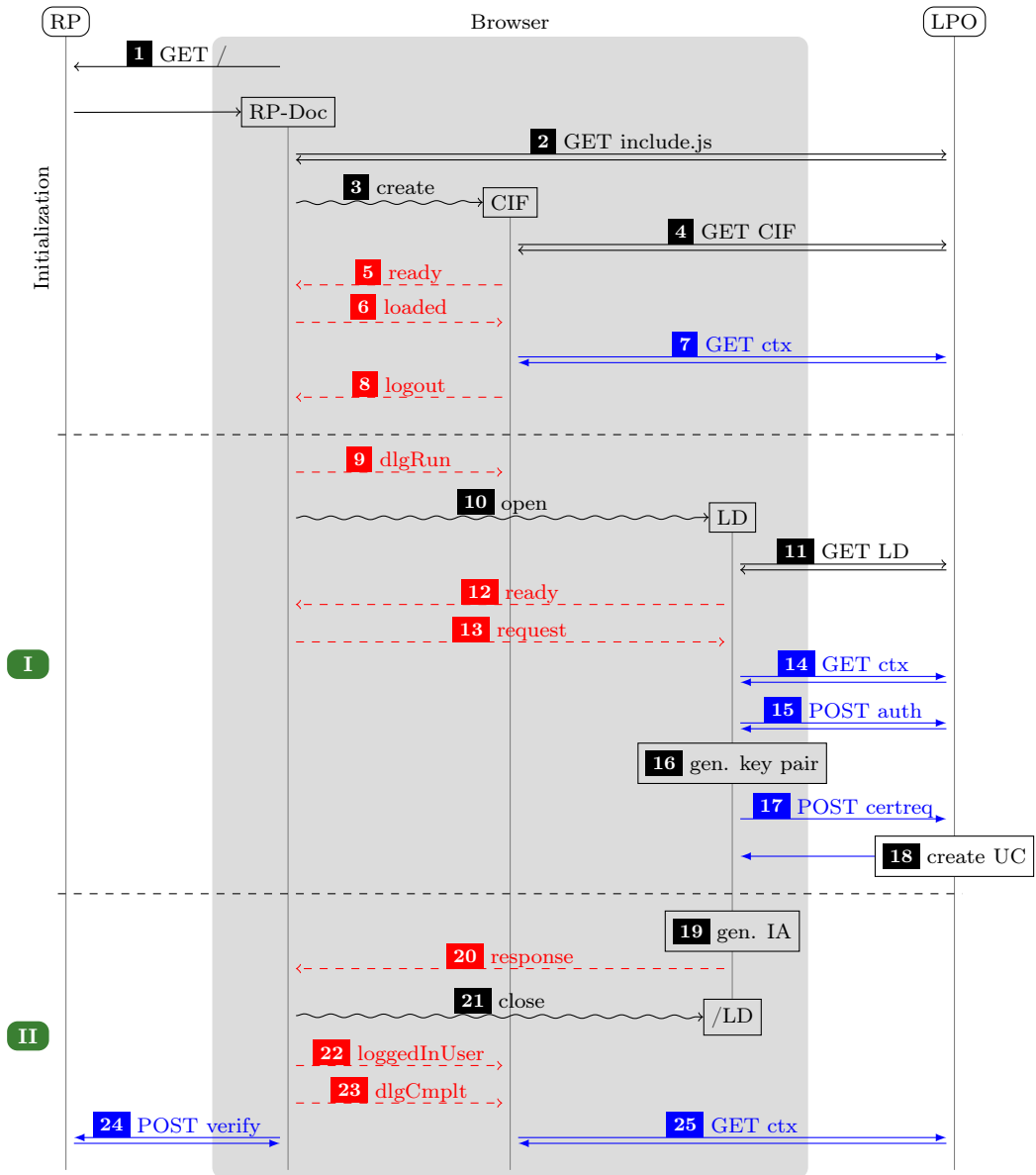
\begin{figure}[t!]\centering
 \scriptsize{
 \begin{tikzpicture}

  \matrix [column sep={7pc,between origins}, row sep=2.5ex]
  {
   \node[draw,anchor=base,rounded corners](rp){RP}; & & \node[anchor=base](browser){Browser}; & & \node[draw,anchor=base, rounded corners](lpo){LPO};\\

   \node(rp-init-1){}; & \node(rp-doc-init-1){}; & & & \\

   \node(rp-init-2){}; & \node[draw,anchor=base](rp-doc){RP-Doc}; & & & \\

   & \node(rp-doc-incl){}; & & & \node(lpo-incl){}; \\

   & \node(rp-doc-iframe-cif){}; & \node[draw,anchor=base](cif-iframe){CIF}; & & \\

   & & \node(cif-init){}; & & \node(lpo-cif-init){}; \\

   & \node(rp-doc-cif-rdy){}; & \node(cif-rdy){}; & & \\ 

   & \node(rp-doc-cif-ld){}; & \node(cif-ld){}; & & \\ 

   & & \node(cif-ctx){}; & & \node(lpo-cif-ctx){}; \\

   & \node(rp-doc-logout){}; & \node(cif-logout){}; & & \\

   \node(phase-0-1-left){}; & & & & \node(phase-0-1-right){}; \\

   & \node(rp-doc-dlg-run){}; & \node(cif-dlg-run){}; & & \\

   & \node(rp-doc-iframe-dlg){}; & & \node[draw,anchor=base](dlg-iframe){LD}; & \\ [0.5ex]

   & & & \node(dlg-init){}; & \node(lpo-dlg-init){}; \\

   & \node(rp-doc-dlg-rdy){}; & & \node(dlg-rdy){}; & \\

   & \node(rp-doc-dlg-req){}; & & \node(dlg-req){}; & \\

   & & & \node(dlg-ctx){}; & \node(lpo-dlg-ctx){}; \\ [1ex]

   & & & \node(dlg-auth){}; & \node(lpo-dlg-auth){}; \\

   & & & \node(dlg-gen-key){}; & \\ [2ex]

   & & & \node(dlg-cert-1){}; & \node(lpo-dlg-cert-1){}; \\

   & & & \node(dlg-cert-2){}; & \node(lpo-dlg-cert-sign){}; \\

   \node(phase-1-2-left){}; & & & & \node(phase-1-2-right){}; \\

   & & & \node(dlg-gen-ia){}; & \\

   & \node(rp-doc-dlg-res){}; & & \node(dlg-res){}; & \\
 
   & \node(rp-doc-dlg-close){}; & & \node[draw,fill=Gainsboro,anchor=base](dlg-close){/LD}; \\

   & \node(rp-doc-cif-liu){}; & \node(cif-liu){}; & & \\

   & \node(rp-doc-cif-dlg-cpt){}; & \node(cif-dlg-cpt){}; & & \\

   \node(rp-vrfy){}; & \node(rp-doc-vrfy){}; & \node(cif-ctx2){}; & & \node(lpo-cif-ctx2){}; \\

   \node(rp-end){}; & \node(rp-doc-end){}; & \node(cif-end){}; & \node(dlg-end){}; & \node(lpo-end){};\\
  };

  \tikzstyle{xhrArrow} = [color=blue,decoration={markings, mark=at
    position 1 with {\arrow[color=blue]{triangle 45}}}, preaction
  = {decorate}]

  \draw [->] (rp-doc-init-1) to node [above=-3pt]{\protostep{sidp-rp-doc-init} GET / } (rp-init-1);
  \draw [->] (rp-init-2.40) -- (rp-doc);

  \draw [->] (rp-doc-incl.20) to node [above=-3pt]{\protostep{sidp-rp-doc-incl} GET include.js } (lpo-incl.160);
  \draw [->] (lpo-incl.200) -- (rp-doc-incl.340);

  \draw [->,snake=snake,segment amplitude=0.2ex] (rp-doc-iframe-cif.40) to node [above=-2pt] {\protostep{sidp-create-cif} create} (cif-iframe);

  \draw [->] (cif-init.20) to node [above=-3pt]{\protostep{sidp-cif-init} GET CIF } (lpo-cif-init.160);
  \draw [->] (lpo-cif-init.200) -- (cif-init.340);

  \draw [->,color=red,dashed] (cif-rdy) to node [above=-3pt]{\protostep{sidp-cif-ready} ready } (rp-doc-cif-rdy);

  \draw [->,color=red,dashed] (rp-doc-cif-ld) to node [above=-2pt]{\protostep{sidp-cif-loaded} loaded } (cif-ld);

  \draw [->,color=blue,>=latex] (cif-ctx.20) to node [above=-3pt]{\protostep{sidp-cif-ctx-1} GET ctx } (lpo-cif-ctx.160);
  \draw [->,color=blue,>=latex] (lpo-cif-ctx.200) -- (cif-ctx.340);

  \draw [->,color=red,dashed] (cif-logout) to node [above=-3pt]{\protostep{sidp-cif-logout} logout } (rp-doc-logout);

  \draw [->,color=red,dashed] (rp-doc-dlg-run) to node [above=-3pt]{\protostep{sidp-rp-doc-dlg-run} dlgRun } (cif-dlg-run);

  \draw [->,snake=snake,segment amplitude=0.2ex] (rp-doc-iframe-dlg.40) to node [above=-2pt] {\protostep{sidp-ld-open} open } (dlg-iframe);

  \draw [->] (dlg-init.20) to node [above=-3pt]{\protostep{sidp-ld-init} GET LD } (lpo-dlg-init.160);
  \draw [->] (lpo-dlg-init.200) -- (dlg-init.340);

  \draw [->,color=red,dashed] (dlg-rdy) to node [above=-3pt]{\protostep{sidp-ld-ready} ready } (rp-doc-dlg-rdy);

  \draw [->,color=red,dashed] (rp-doc-dlg-req) to node [above=-3pt]{\protostep{sidp-ld-request} request} (dlg-req);

  \draw [->,color=blue,>=latex] (dlg-ctx.20) to node [above=-3pt]{\protostep{sidp-ld-ctx} GET ctx } (lpo-dlg-ctx.160);
  \draw [->,color=blue,>=latex] (lpo-dlg-ctx.200) -- (dlg-ctx.340);

  \draw [->,color=blue,>=latex] (dlg-auth.20) to node [above=-3pt]{\protostep{sidp-ld-auth} POST auth } (lpo-dlg-auth.160);
  \draw [->,color=blue,>=latex] (lpo-dlg-auth.200) -- (dlg-auth.340);

  \node [draw,fill=Gainsboro] at (dlg-gen-key) {\protostep{sidp-gen-key-pair} gen. key pair };

  \draw [->,color=blue,>=latex] (dlg-cert-1) to node [above=-3pt]{\protostep{sidp-ld-certreq} POST certreq} (lpo-dlg-cert-1);
  \node [draw,fill=White] (lpo-dlg-cert-sign-draw) at (lpo-dlg-cert-sign) {\protostep{sidp-create-uc} create UC };
  \draw [->,color=blue,>=latex] (lpo-dlg-cert-sign-draw) -- (dlg-cert-2);

  \node [draw,fill=Gainsboro] at (dlg-gen-ia) {\protostep{sidp-gen-ia} gen. IA};

  \draw [->,color=red,dashed] (dlg-res) to node [above=-3pt]{\protostep{sidp-ld-response} response} (rp-doc-dlg-res);

  \draw [->,snake=snake,segment amplitude=0.2ex] (rp-doc-dlg-close.40) to node [above=-2pt] {\protostep{sidp-ld-close} close} (dlg-close);

  \draw [->,color=red,dashed] (rp-doc-cif-liu) to node [above=-3pt]{\protostep{sidp-rp-doc-liu} loggedInUser} (cif-liu);

  \draw [->,color=red,dashed] (rp-doc-cif-dlg-cpt) to node [above=-3pt]{\protostep{sidp-rp-doc-dlgCmplt} dlgCmplt} (cif-dlg-cpt);

  \draw [->,color=blue,>=latex] (rp-doc-vrfy.160) to node [above=-3pt]{\protostep{sidp-rp-doc-verify} POST verify} (rp-vrfy.20);
  \draw [->,color=blue,>=latex] (rp-vrfy.340) -- (rp-doc-vrfy.200);

  \draw [->,color=blue,>=latex] (cif-ctx2.20) to node [above=-3pt]{\protostep{sidp-cif-ctx-2} GET ctx } (lpo-cif-ctx2.160);
  \draw [->,color=blue,>=latex] (lpo-cif-ctx2.200) -- (cif-ctx2.340);

  \begin{pgfonlayer}{background}
   \node (rp-doc-a) [above of=rp-doc]{};
   \node (rp-doc-al) [left of=rp-doc-a]{};
   \node (dlg-end-br) [right of=dlg-end]{};
   \filldraw [color=Gainsboro,rounded corners] (rp-doc-al) rectangle (dlg-end-br);

  \draw [color=gray] (rp.270) -- (rp-end);
  \draw [color=gray] (rp-doc.270) -- (rp-doc-end);
  \draw [color=gray] (cif-iframe.270) -- (cif-end);
  \draw [color=gray] (dlg-iframe.270) -- (dlg-close);
  \draw [color=gray] (lpo.270) -- (lpo-end);

  \end{pgfonlayer}

    \draw [dashed] (phase-0-1-left.180) -- (phase-0-1-right.0);
    \draw [dashed] (phase-1-2-left.180) -- (phase-1-2-right.0);

    \node(phase-0-label) at ($(phase-0-1-left)!0.5!(rp)$) [rotate=90,above=1ex] {Initialization};
    \node(phase-1-label) at ($(phase-1-2-left)!0.5!(phase-0-1-left)$) [left=1ex] {\refcopybigprotophase{provisioning}};
    \node(phase-2-label) at ($(rp-end)!0.5!(phase-1-2-left)$) [left=1ex] {\refcopybigprotophase{authentication}};

 \end{tikzpicture}}
\caption{BrowserID sIdP flow overview. Black arrows (open tips) denote
  HTTP messages, blue arrows (filled tips) denote \xhrs, red (dashed)
  arrows are \pms, snake lines are commands to the browser.}
\label{fig:browserid-lowlevel}

\end{figure}

We here provide a detailed description of the BrowserID
flow including the LD and the CIF, but considering only the
sIdP mode. The formal model and analysis of this setting is
presented in Section~\ref{sec:analysisbrowserid} and
Appendix~\ref{app:browseridmodel}. As above, we leave out
steps for fetching additional resources (like JavaScript
files) and some less relevant \pms.

As before, we assume that the user uses a ``fresh''
browser, i.e., the user has not been logged in before.
Figure~\ref{fig:browserid-lowlevel} shows a sequence
diagram of the run. In comparison to the high-level
description with Phases~\refbigprotophase{provisioning}
and~\refbigprotophase{authentication} in
Section~\ref{sec:browseridoverview}, a new phase, the
\emph{Initialization}, is shown, which contains some
initialization steps and some steps of the CIF for
automatic creation of CAPs.

\myparagraph{Initialization.} First, the user opens a web
page of an RP (RP-Doc) in her
browser~\refprotostep{sidp-rp-doc-init}. RP-Doc includes
the BrowserID JavaScript from
LPO~\refprotostep{sidp-rp-doc-incl}. The JavaScript in
RP-Doc then initializes the BrowserID JavaScript, which
first creates a communication \iframe (CIF) within
RP-Doc~\refprotostep{sidp-create-cif}. The content of the
CIF is loaded from LPO~\refprotostep{sidp-cif-init}. When
the CIF has been initialized successfully, it sends a
\texttt{ready} \pm to the BrowserID JavaScript in
RP-Doc~\refprotostep{sidp-cif-ready}, which in turn
responds with the \texttt{loaded}
\pm~\refprotostep{sidp-cif-loaded}. This message may
contain an email address, which we ignore for now (see
below). The CIF saves the sender's origin of this \pm, as
it identifies the RP it is working with.\footnote{Note that
  for postMessages the sender origin cannot be spoofed and
  is always correct (see \cite{html5webmessaging20120501}
  for details).}  It then fetches the session context from
LPO using \xhr~\refprotostep{sidp-cif-ctx-1}. The session
context contains information about whether the user is
already logged in at LPO, which, by our assumption, is not
the case at this point. The session context also contains
an \xsrf protection token which will be sent in all
subsequent POST requests to LPO. Also, an $\str{httpOnly}$
cookie called \texttt{browserid\_state} is set, which
contains an LPO session identifier. The CIF finishes the
initialization by sending a \texttt{logout}
\pm~\refprotostep{sidp-cif-logout} to RP-Doc, indicating
that the browser is currently not logged in at this RP. (In
RP-Doc, this calls a message handler \texttt{onlogout}
which RP may have registered.)

\myparagraph{Phase~\refbigprotophase{provisioning}.} Now
the user starts to log in at RP, typically by clicking on a
login button in RP-Doc, which calls the BrowserID login
function. This JavaScript first tells the CIF that it will
now open a login dialog window (LD) by sending the
\texttt{dialog\_running}
\pm~\refprotostep{sidp-rp-doc-dlg-run}; this pauses the
CIF, in particular, automatic CAP generation (see
below). The BrowserID login function then opens the
LD~\refprotostep{sidp-ld-open}.  Its content is fetched
from LPO~\refprotostep{sidp-ld-init}. When it is fully
loaded, it sends a \texttt{ready} \pm to the BrowserID
JavaScript in RP-Doc~\refprotostep{sidp-ld-ready}, which is
answered by sending a \texttt{request} \pm
back~\refprotostep{sidp-ld-request}, indicating that the
sender's origin of this \pm requests a CAP.

After this request, the dialog fetches the session context
at LPO~\refprotostep{sidp-ld-ctx} (similar to what the CIF has done
before). As the user is still not logged in, she is now
asked to provide an email address and a password for
LPO. These are then sent to LPO by an
\xhr~\refprotostep{sidp-ld-auth}. 
If the credentials are accepted,
a new key pair is generated by LD's JavaScript~\refprotostep{sidp-gen-key-pair}
and the public key along with the entered email address is
sent to LPO~\refprotostep{sidp-ld-certreq} in order to get a UC~\refprotostep{sidp-create-uc}
(see Section~\ref{sec:browseridoverview}).  Moreover, the
key pair and the UC are stored in the \ls under the origin
of LPO.  

\myparagraph{Phase~\refbigprotophase{authentication}.}
Afterwards, an IA containing the so-called \emph{audience} (the
sender's origin of the \texttt{request} \pm~\refprotostep{sidp-ld-request})
and some expiration date is created, signed (with the
generated private key), and combined with the certificate
to a CAP~\refprotostep{sidp-gen-ia}.  In the \ls it is recorded that the
user is logged in at RP with the current email address.
The CAP and the email address are now sent back to RP-Doc
in a \texttt{response} \pm~\refprotostep{sidp-ld-response}. After this, the LD
is closed~\refprotostep{sidp-ld-close}.

The BrowserID JavaScript in RP-Doc informs the CIF that it
now thinks that the email address received in the
\texttt{response} \pm is logged in~\refprotostep{sidp-rp-doc-liu}. Next, it
tells the CIF that the LD is now closed~\refprotostep{sidp-rp-doc-dlgCmplt}, by
which the CIF is awoken from pausing. The CIF then fetches
the session context again~\refprotostep{sidp-cif-ctx-2} (as in~\refprotostep{sidp-cif-ctx-1})
in order to perform some additional checks (see below).

It is not specified in the BrowserID system how the RP-Doc
has to process the CAP received in
step~\refprotostep{sidp-ld-response}. Typically, as already mentioned in
Section~\ref{sec:browseridoverview}, the RP-Doc would send
the CAP to the RP's server~\refprotostep{sidp-rp-doc-verify}, which then can
verify the CAP. If successful, RP can consider the user
(with the email address mentioned in the CAP) to be logged
in and send her some token, the RP service token (as
introduced in Section~\ref{sec:browseridoverview}).

\subsection{Additional Checks} We note that when \pms are
sent, the BrowserID system makes certain checks. These
checks are carried out by two different (Mozilla)
JavaScript libraries. The communication between RP-Doc and
CIF is realized with the library \emph{JSChannel} and the
communication between RP-Doc and LD is realized with the
library \emph{WinChan}.

First recall that \pms can be sent by providing the
receiver's origin. The browser ensures that such a \pm can
only be received by a document having the origin the sender
has specified. If the receiver's origin does not match the
one specified by the sender, the sender receives a
JavaScript exception. However, the sender is not required
to provide a receiver's origin, so any receiver can receive
the \pm. Also, a receiver can check from which window a \pm
was sent and which origin the sender belongs to. 

Now, the CIF only sends or accepts \pms to or from its
parent window (which typically should be RP-Doc). However,
the CIF does not check any origin while receiving \pms and
does not provide an origin when sending \pms. When RP-Doc
receives a message it expects to come from CIF, RP-Doc
checks if the origin of this message matches LPO and if the
sender's window is the window of CIF. RP-Doc always
provides LPO as the receiver's origin when sending messages
to the window it believes to contain the CIF.

The LD sends~\refprotostep{sidp-ld-ready} only to its opener (RP-Doc)
without providing any receiver's origin to check. After
this, the LD accepts only one \texttt{request} \pm~\refprotostep{sidp-ld-request} and blocks any further incoming \pms. The
sender's origin of the \texttt{request} \pm~\refprotostep{sidp-ld-request} is
used by LD to determine the receiver's origin of the
\texttt{response} \pm~\refprotostep{sidp-ld-response}. LD also fixes the
receiver of~\refprotostep{sidp-ld-response} to be its opener. When RP-Doc sends
the \texttt{request} \pm~\refprotostep{sidp-ld-request} to the LD, it sets the
receiver's origin to be LPO in the \pm. However, any \pm
RP-Doc expects to be sent by LD is not checked (see also
Section~\ref{sec:attacks-browserid}).

During the interaction between RP-Doc and LD, an additional
check is set in place at both parties: If one of both
documents is navigated away, the window of LD is closed
immediately (and therefore any process in the LD is
aborted).

We also note that step~\refprotostep{sidp-rp-doc-dlgCmplt} triggers two checks in
the CIF: First, the CIF checks the current login status at
LPO, by fetching the session context~\refprotostep{sidp-cif-ctx-2}. Second,
the CIF compares the email address received in~\refprotostep{sidp-rp-doc-liu}
to the one that is marked as being logged in at RP in the
\ls (under the origin of LPO). If in one of the checks the
user is considered to be not logged in, a \texttt{logout}
\pm is sent to RP-Doc (similarly to~\refprotostep{sidp-cif-logout}). Otherwise, if in the second check a mismatch
is detected, the CIF creates a new CAP according to the
information in the \ls and sends it as a so-called
\texttt{login} \pm to RP-Doc. Whether this CAP is used by
RP-Doc or the one received in step~\refprotostep{sidp-ld-response} depends on
the way the RP-Doc uses the API provided by BrowserID. One
possibility (which is considered in the BrowserID test
implementation) is that RP-Doc relays all received CAPs to
the RP server with an \xhr. The RP server, as already
mentioned above, would then verify each CAP it receives and
issue an RP service token every time. This is also what is
done in our model of the BrowserID system.

\subsection{Automatic CAP Creation} Once a run as described
above is completed, an RP-Doc can get CAPs directly during
the Initialization from the CIF: The CIF will automatically
issue a fresh CAP and send it to RP-Doc (in a
\texttt{login} \pm instead of~\refprotostep{sidp-cif-logout}) iff (1) some
email address is marked as logged in at RP in the \ls
(under the origin LPO), (2) if an email address was
provided in the \texttt{loaded} \pm~\refprotostep{sidp-cif-loaded}, this email
address differs from the one recorded in the \ls, and (3)
the user is logged in at LPO (indicated in~\refprotostep{sidp-cif-ctx-1}). If
necessary, a new key pair is created and a corresponding
new UC is requested at LPO.

\subsection{LPO Session} As mentioned before, in the
initial run LPO establishes a session with the browser by
setting a cookie \texttt{browserid\_state} (in step~\refprotostep{sidp-cif-ctx-1}) on the client-side.

If such a run is started again (possibly with some other
RP) with the same browser in an LPO session in which the
user is already logged in at LPO, the user is not asked
again by the LD to provide her credentials. Instead she is
presented a list of her email addresses (which is fetched
from LPO and cached in the \ls) in order to choose one
address. Then, she is asked if she trusts the computer she
is using and is given the option to be logged in for one
month or ``for this session only'' (\emph{ephemeral}
session). However, in any case cookies will be stored for
some time in the browser and will be valid for some time on
the LPO server (one hour to 30 days).

\subsection{Logout} We have to differentiate between two
ways of logging out: an RP logout and an LPO logout. An RP
logout is handled by the CIF after it has received a
\texttt{logout} \pm from RP-Doc. The CIF changes the \ls
(under the origin of LPO) such that no email address is
recorded to be logged in at RP and replies to RP-Doc with a
\texttt{logout} \pm. RP-Doc can run some callback it may
have registered before.

An LPO logout essentially requires to logout at the web
site of LPO. The LPO logout removes all key pairs and
certificates from the \ls and invalidates the session on
the LPO server.

\section{BrowserID Model}\label{app:browseridmodel}

In this section, we provide the full BrowserID model, i.e.,
the web system $\bidwebsystem=(\bidsystem, \scriptset,
\mathsf{script}, E_0)$, and its security properties. We
note that our model considers the BrowserID system with the
fixes proposed in Section~\ref{sec:attacks-browserid}.

We first note that a UC $\mi{uc}$ for the identity (email
address) $i$ and public key (verification key) $\pub(k_u)$,
where $k_u$ is the private key (signing key) of the user,
is modeled to be a message of the form $\mi{uc}=\sig{\an{i,
    \pub(k_u)}}{k^\LPO}$, where $k^\LPO$ is the signing key
of $\LPO$. An IA $ia$ for an origin
$\mi{o}=\an{\str{example.com},\str{S}}$ is a message of the
form $ia=\sig{o}{k_u'}$. Now, a CAP is of the form
$\an{\mi{uc},\mi{ia}}$. We call that CAP \emph{valid} iff
$k_u = k_u'$. Note that the time stamps are omitted both
from the UC and the IA. This models that both certificates
are valid indefinitely. In reality, as explained in
Section~\ref{sec:browserid}, they are valid for a certain
period of time, as indicated by the time stamps. So our
modeling is a safe overapproximation.

The set $\nonces$ of nonces is partitioned into three
disjoint sets, an infinite set $N^\bidsystem$, an infinite
set $K_\text{private}$, and a finite set $\PLISecrets$. The
set $N^\bidsystem$ is further partitioned into infinite
sets of nonces, one set $N^p\subseteq N^\bidsystem$ for
every $p\in \bidsystem$.

The sets $\addresses$ and $\dns$ have already been
introduced in Section~\ref{sec:browerIDmodel}. Let
$\mapAddresstoAP$ and $\mapDomain$ denote the assignments
from atomic processes to sets of $\addresses$ and $\dns$,
respectively, where browsers (in $\fAP{B})$ do not have a
domain. If $\mapDomain$ or $\mapAddresstoAP$ returns a set
with only one element, we often write $\mapDomain(x)$ or
$\mapAddresstoAP(x)$ to refer to the element. Let
$\mapKey\colon \dns \to K_\text{private}$ be an injective
mapping that assigns a private key to every domain.
The set $\PLISecrets\subseteq \nonces$ is the set of
passwords (secrets) the browsers share with $\LPO$. Let $\IDs$
be the finite set of a email address (IDs) of the form
$\an{\mi{name},d}$, with $\mi{name}\in \mathbb{S}$ and
$d\in \dns$, registered at $\LPO$. As mentioned in
Section~\ref{sec:browerIDmodel}, different browsers own
disjoint sets of secrets and different secrets are assigned
disjoint sets of IDs. Let $\mapPLItoOwner: \PLISecrets \to
\fAP{B}$ denote the mapping that assigns to each secret the
browser that owns this secret. Let $\mapIDtoPLI: \IDs \to
\PLISecrets$ denote the mapping that assigns to each
identity the secret it belongs to. Now, we define the
mapping $\mapIDtoOwner: \IDs \to \fAP{B}, i \mapsto
\mapPLItoOwner(\mapIDtoPLI(i))$, which assigns to each
identity the browser that owns this identity.

We are now ready to define the attacker, the browsers, $\LPO$,
the relying parties, and the scripts used in BrowserID in
more detail. As for the attacker and browsers, nothing much
needs to be specified. The attacker will be modeled as a
network attacker as specified in
Section~\ref{sec:websystem} and browsers will simply be web
browsers as specified in Section~\ref{sec:web-browsers}. We
only need to fix the addresses the attacker and the
browsers can listen to and their initial states.

\subsection{Attacker}\label{app:attacker} As mentioned, the attacker
$\fAP{attacker}$ is modeled to be a network attacker as
specified in Section~\ref{sec:websystem}. We allow it to
listen to/spoof all available IP addresses, and hence,
define $I^\fAP{attacker} = \addresses$. His initial state
is $s_0^\fAP{attacker} = \an{attdoms, \mi{pubkeys}}$, where
$attdoms$ is a sequence of all domains along with the
corresponding private keys owned by the attacker and
$\mi{pubkeys}$ is a sequence of all domains and the
corresponding public keys. All other parties use the
attacker as a DNS server.

\subsection{Browsers}\label{sec:browsers} 
Each $b \in \fAP{B}$ is a web browser as defined in
Section~\ref{sec:web-browsers}, with $I^b :=
\{\mapAddresstoAP(b)\}$ being its address and the initial
state $s_0^b$ defined as follows: the key mapping maps
every domain to its public key, according to the mapping
$\mapKey$; the DNS address is
$\mapAddresstoAP(\fAP{attacker})$; the secrets are those
owned by $b$ (as defined above) and they are indexed by the
origin $\an{\domLPO,\str{S}}$; $\mi{sts}$ is
$\an{\domLPO}$. (Without the latter, the
attacker could trivially inject, by an HTTP response, its
own \texttt{browserid\_state} cookie and by this violate
property \textbf{(B)}.)

\subsection{LPO} \label{sec:lpo}$\LPO$ is a an atomic DY process
$(I^\fAP{LPO}, Z^\fAP{LPO}, R^\fAP{LPO}, s^\fAP{LPO}_0,
N^\fAP{LPO})$ with the IP address $I^\fAP{LPO} =
\{\mapAddresstoAP(\fAP{LPO})\}$.
The initial state $s^\LPO_0$ of $\LPO$ contains
the private key of its domain, its signing key
$k^\LPO$, all secrets in $\PLISecrets$ and the
corresponding sequences of IDs. ($\LPO$ does not need the
public keys of other parties, which is why we omit them
from $\LPO$'s initial state.). The definition of $R^\LPO$
follows the description of $\LPO$ in
Section~\ref{sec:browserid} in a straightforward way.
HTTP responses by $\LPO$ can contain, besides complex
terms (e.g., for \xhr responses), strings
representing scripts, namely the script
$\str{script\_LPO\_cif}$ run in the CIF and the script
$\str{script\_LPO\_ld}$ run in the LD. These scripts are
defined in Appendices~\ref{app:scriptLPOcif} and
\ref{app:scriptLPOld}, respectively.

Before we provide a detailed formal specification of
$\LPO$, we first provide an informal description.

\myparagraph{Client sessions at $\LPO$.}  Any party can
establish a \emph{session} at $\LPO$. Such a session can
either be authenticated or unauthenticated. Roughly
speaking, a session becomes authenticated if a client has
provided a secret $sec\in \PLISecrets$ to $\LPO$ in the
session. Such a session is then authenticated for all IDs
associated with $sec$, i.e., for all such IDs an UC can be
requested from $\LPO$ within the authenticated session.
(Recall that for every secret, $\LPO$ contains a list of all
IDs associated with that secret.)  An authenticated session
can (non-deterministically) \emph{expire}, i.e. the
authenticated session can get unauthenticated or it is
removed completely. Such an expiration is used to model a
user logout or a session expiration caused by a timeout.

More specifically, a session is identified by a nonce,
which is issued by $\LPO$. Moreover, a session is associated
with some xsrfToken, which is also a nonce issued by
$\LPO$. $\LPO$ stores all information about established sessions
in its state as a dictionary indexed by the session
identifier. In this dictionary, for every session $\LPO$
stores a pair containing the xsrfToken and, in
authenticated sessions, the sequence of all IDs associated
with the secret provided in the session, or, in
unauthenticated sessions, the empty sequence $\an{}$ of
IDs. On the receiver side (typically a browser) $\LPO$ placed,
by appropriate headers in its HTTPS responses, a cookie
named $\str{browserid\_state}$ whose value is the session
identifier (a nonce). This cookie is flagged to be a
session, httpOnly, and secure cookie.

\myparagraph{HTTPSRequests to $\LPO$.}  $\LPO$ answers only to
certain requests (listed below). All such requests have to
be over HTTPS. Also, all responses send by $\LPO$ contain the
$\str{Strict \mhyphen Transport \mhyphen Security}$
header. 

\begin{description}

\item[\texttt{GET /cif}.] $\LPO$ replies to this request by
  providing the script $\str{script\_LPO\_cif}$.

\item[\texttt{GET /ld}.] $\LPO$ replies to this request by
  providing the script $\str{script\_LPO\_ld}$.

\item[\texttt{GET /ctx}.]  $\LPO$ replies with a session
  context. More precisely, $\LPO$ first checks if a cookie
  $\str{browserid\_state}$ was sent within this request and
  if its value identifies a session within $\LPO$'s state. If
  such a session exists, $\LPO$ responds to such a request
  with the list of (authenticated) IDs for this
  session,\footnote{In the real implementation, the session
    context only contains a flag indicating the
    authentication state of the session. However, another
    GET request interface is available to retrieve the list
    of authenticated IDs for the current session. Here, for
    simplicity, we right away provide all authenticated IDs
    in the session context.} the xsrfToken, and a $\str{Set
    \mhyphen Cookie}$ header, which sets the
  $\str{browserid\_state}$ cookie. If no cookie
  $\str{browserid\_state}$ was sent in the request, or if
  the value of the cookie $\str{browserid\_state}$ does not
  identify a session within $\LPO$'s state, $\LPO$ first creates
  a new session. Such a new session contains a fresh nonce
  as a session identifier, the empty sequence $\an{}$ of
  IDs, and a fresh nonce as a xsrfToken. Once such a
  session is created, $\LPO$ responds as above.

\item[\texttt{POST /auth}.] This request is sent to
  authenticate a session at $\LPO$. A request to this
  interface has to contain some secret $sec\in \PLISecrets$
  in its body. The request also has to contain the cookie
  $\str{browserid\_state}$ which has to refer to some
  session in the state of $\LPO$. Moreover, the request has to
  contain an xsrfToken in its body which has to match the
  one recorded in the considered session in $\LPO$'s
  state. The session recorded in the state of $\LPO$ is then
  modified to include the sequence of all IDs associated
  with $sec$. The response to such a request contains some
  static acknowledgment. 

\item[\texttt{POST /certreq}.] Such a request is sent to
  $\LPO$ in order to request a UC. The request has to contain
  an ID and a public key in its body, for which a UC is
  requested. The request also has to contain a cookie named
  $\str{browserid\_state}$ which has to refer to some
  session recorded in the state of $\LPO$. Moreover, the
  request has to contain an xsrfToken in its body which
  matches the xsrfToken in the considered session record in
  $\LPO$'s state. Also, the sequence of IDs in the considered
  session recorded in $\LPO$'s state has to contain the ID
  provided in the (body of the) request. This ID is then
  paired with the public key sent in the request and the
  resulting pair is then signed with $k^\LPO$. In
  other words, a UC is created for the ID and the public
  key provided in the request. Finally, $\LPO$ responds with
  this UC.

\end{description}

We now define $\LPO$ formally as an atomic DY process
$(I^\LPO, Z^\LPO, R^\LPO, s^\LPO_0,
N^\LPO)$. As already mentioned, we define $I^\LPO
= \{\mapAddresstoAP(\LPO)\}$.

In order to define the set $Z^\LPO$ of states of $\LPO$,
we first define the terms describing the session context of
a session. 

\begin{definition}
  A term of the form
  $\an{\mi{ids},\mi{xsrfToken}}$ with $\mi{ids} \subsetPairing
  \IDs$ and $\mi{xsrfToken} \in \nonces$ is called an
  \emph{LPO session context}. We denote the set of all LPO
  session contexts by $\LPOSessionCTXs$.
\end{definition}

Now, we define the set $Z^\LPO$ of states of LPO as
well as the initial state $s^\LPO_0$ of LPO.

\begin{definition}
  A state $s\in Z^\LPO$ of LPO is a term of the form
  $\langle\mi{nonces}$, $\mi{sslkey}$, $\mi{signkey}$,
  $\mi{sessions}$, $\mi{secrets}\rangle$ where $\mi{nonces}
  \subsetPairing \nonces$ (used nonces),
  $\mi{sslkey}=\mapKey(\mapDomain(\LPO))$, $\mi{signkey}=
  k^{\LPO}$, $\mi{sessions} \in
  \dict{\nonces}{\LPOSessionCTXs}$, and $\mi{secrets} \in
  \dict{\PLISecrets}{\terms}$ is a dictionary which assigns
  to every secret $sec\in \PLISecrets$ the sequence of all
  IDs associated with $\mi{sec}$.\footnote{As mentioned
    before, the state of LPO does not need to contain
    public keys.}

    The \emph{initial state $s^\LPO_0$ of LPO} is a state of LPO
    with $s^\LPO_0.\str{nonces} = \an{}$ and $s^\LPO_0.\str{sessions} =
    \an{}$.
\end{definition}

\begin{example} A possible state $s$ of LPO may look like
  this:
\begin{align*}
  s &=
  \an{\an{n_1,\ldots,n_m},k,k^\LPO,\mi{sessions},\mi{secrets}} 
\end{align*}
with
\begin{align*}
  \mi{sessions} &= \an{ \an{\str{sessionid_1},\an{\an{\mi{id}'_1, \ldots, \mi{id}'_l},\str{xsrfToken}}} , \ldots } \\
  \mi{secrets} &= \an{ \an{\str{secret_1},\an{\mi{id}_1, \ldots, \mi{id}_p}} , \ldots  }
\end{align*}
\end{example}

We now specify the relation $R^\LPO \subseteq (\events \times Z^\LPO) \times
  (2^\events \times Z^\LPO)$ of LPO. Just like
in Appendix~\ref{sec:descr-web-brows}, we describe this
relation by a non-deterministic algorithm.  We note that we
use the function \textsf{TAKENONCE} introduced in
Section~\ref{app:proceduresbrowser} for this purpose.

\newpage
\vspace*{-1.78em}
\captionof{algorithm}{\label{alg:lpo} Relation of LPO
  $R^\LPO$ }
\begin{algorithmic}[1]
\Statex[-1] \textbf{Input:} $(a{:}f{:}m),s$
 \Let{$s'$}{$s$}

 \If{$m \equiv \trigger$} \Comment{Triggers a
   (non-deterministic) logout or session expiration}
  \If{$s'.\str{sessions} \not\equiv \an{}$}
  \LetND{$\mi{sessionid}$}{$\{\mi{id}_j | \mi{id}_j \in s'.\str{sessions}\}$}
  \If{$\mi{sessionid} \inPairing s'.\str{sessions}$}
   \LetND{$\mi{choice}$}{$\{\str{logout},\str{expire}\}$}
   \If{$\mi{choice \equiv \str{logout}}$}
    \Let{$\mi{session}$}{$s'.\str{sessions}[\mi{sessionid}]$}
    \Let{$\mi{session}[ids]$}{$\an{}$}
    \Let{$s'.\str{sessions}[\mi{sessionid}]$}{$\mi{session}$}
    \State \textbf{stop} $\{\},s'$
   \Else
    \State \textbf{remove} the element with key $\mi{sessionid}$ from the dictionary $s'.\str{sessions}$
    \State \textbf{stop} $\{\},s'$
   \EndIf
  \EndIf
  \EndIf

 \Else
 
  \LetST{$m_\text{dec}$, $k'$}{$\an{m_\text{dec},k'} \equiv
    \dec{m}{s'.\str{sslkey}}$}{\textbf{stop}~$\{\},s$}

  \LetST{$n$, $\mi{method}$, $\mi{path}$, $\mi{params}$, $\mi{headers}$, $\mi{body}$}
        {$ 
           \hreq{ nonce=n,
            method=\mi{method},
            host=\domLPO,
            path=\mi{path},
            parameters=\mi{params},
            headers=\mi{headers},
            body=\mi{body}
          } \equiv m_\text{dec}$}{\textbf{stop}~$\{\},s$} \label{line:lpo-does-https-only}

        \If{$\mi{method} \equiv \mGet \wedge \mi{path} \equiv
          \str{/cif}$} \Comment{Deliver CIF script}
   \LetBreak{$m'$}{$\ehresp{ nonce=n,
                         status=200,
                         headers=\an{\an{\str{Strict \mhyphen Transport \mhyphen Security},\True}},
                         body={\langle\str{script\_LPO\_cif},\linebreak[4]\mi{initState_{cif}}\rangle}
                       }{k'}$}
                     \StatexNoIndent where $\mi{initState_{cif}}$ is the initial scriptstate of $\mi{script\_LPO\_cif}$ according to Definition~\ref{def:scriptstatelpocif}.
   \Let{$E$}{$\{  (f{:}a{:}m') \}$}
   \State \textbf{stop} $E$, $s'$

  \ElsIf{$\mi{method} \equiv \mGet \wedge \mi{path} \equiv \str{/ld}$} \Comment{Deliver LD script.}

   \LetBreak{$m'$}{$\ehresp{ nonce=n,
                         status=200,
                         headers=\an{\an{\str{Strict \mhyphen Transport \mhyphen Security},\True}},
                         body=\an{\str{script\_LPO\_ld},\mi{initState_{ld}}}
                       }{k'}$} 
                     \StatexNoIndent where $\mi{initState_{ld}}$ is the initial scriptstate of $\mi{script\_LPO\_ld}$ according to Definition~\ref{def:scriptstatelpold}.
   \Let{$E$}{$\{ (f{:}a{:}m') \}$}
   \State \textbf{stop} $E$, $s'$

  \ElsIf{$\mi{method} \equiv \mGet \wedge \mi{path} \equiv \str{/ctx}$} \Comment{Deliver context information.}

   \Let{$\mi{sessionid}$}{$\mi{headers}[\str{Cookie}][\str{browserid\_state}]$}

   \If{$\mi{sessionid} \not\inPairing s'.\str{sessions}$} \Comment{Create new session if needed.}
    \Let{$\mi{sessionid},s'$}{\textsf{TAKENONCE}($s'$)}
    \Let{$\mi{ids}$}{$\an{}$}
    \Let{$\mi{xsrfToken},s'$}{\textsf{TAKENONCE}($s'$)} \label{line:chose-xsrftoken}
    \Let{$s'.\str{sessions}$}{$s'.\str{sessions} \plusPairing \an{\mi{sessionid},\an{\mi{ids},\mi{xsrfToken}}}$}
   \EndIf

   \Let{$\mi{result}$}{$s'.\str{sessions}[\mi{sessionid}]$} \label{line:release-xsrftoken}

   \Let{$\mi{headers'}$}{$\langle
                                \an{\str{Strict \mhyphen Transport \mhyphen Security},\True},$
              \Statex            $ \an{\str{Set \mhyphen Cookie},
                                    \an{\an{\str{browserid\_state},
                                        \an{ \mi{sessionid},
                                             \True,
                                             \True,
                                             \True
                                           }
                         }}}\rangle$} \label{line:set-session-cookie}

   \Let{$m'$}{$\ehresp{ nonce=n,
                         status=200,
                         headers=\mi{headers'},
                         body=\mi{result}
                       }{k'}$}
   \Let{$E$}{$\{ (f{:}a{:}m') \}$}
   \State \textbf{stop} $E$, $s'$

  \ElsIf{$\mi{method} \equiv \mPost \wedge \mi{path} \equiv \str{/auth}$} \Comment{Authenticate session.}

   \Let{$\mi{sessionid}$}{$\mi{headers}[\str{Cookie}][\str{browserid\_state}]$}
   \LetST{$\mi{secret}$, $\mi{xsrfToken}$}{$ \an{\mi{secret},\mi{xsrfToken}}\equiv \mi{body}$}{\textbf{stop}~$\{\},s$}

   \If{$\mi{sessionid} \inPairing s'.\str{sessions}$}
    \If{$\mi{secret} \inPairing s'.\str{secrets} \wedge s'.\str{sessions}[\mi{sessionid}].\str{xsrfToken} \equiv \mi{xsrfToken}$}
     \Let{$\mi{ids}$}{$s'.\str{secrets}[\mi{secret}]$}
     \Let{$s'.\str{sessions}[\mi{sessionid}].\str{ids}$}{$\mi{ids}$}
     \Let{$m'$}{$\ehresp{ nonce=n,
                         status=200,
                         headers=\an{\an{\str{Strict \mhyphen Transport \mhyphen Security},\True}},
                         body=\True
                       }{k'}$}
     \Let{$E$}{$\{ (f{:}a{:}m') \}$}
     \State \textbf{stop} $E$, $s'$
    \EndIf
   \EndIf

  \ElsIf{$\mi{method} \equiv \mPost \wedge \mi{path} \equiv \str{/certreq}$} \Comment{Sign pubkey, deliver UC} \label{line:sign-pubkey}

   \Let{$\mi{sessionid}$}{$\mi{headers}[\str{Cookie}][\str{browserid\_state}]$} 
   \Let{$\mi{ids}$}{$s'.\str{sessions}[\mi{sessionid}].\str{ids}$}
   \Let{$\mi{xsrfToken}$}{$s'.\str{sessions}[\mi{sessionid}].\str{xsrfToken}$}
   \LetST{$\mi{id}$, $\mi{pubkey}$, $\mi{xsrfToken}'$}{$\an{\mi{id},\mi{pubkey},\mi{xsrfToken}'} \equiv \mi{body}$}{\textbf{stop} $\{\},s$}

   \If{$\mi{id} \inPairing \mi{ids} \wedge \mi{xsrfToken} \equiv \mi{xsrfToken}'$}

    \Let{$\mi{uc}$}{$\sig{\an{\mi{id},\mi{pubkey}}}{s'.\str{signkey}}$} \label{line:lpo-issue-uc} \label{line:start-certreq}

    \Let{$m'$}{$\ehresp{ nonce=n,
                         status=200,
                         headers=\an{\an{\str{Strict \mhyphen Transport \mhyphen Security},\True}},
                         body=\mi{uc}
                       }{k'}$}
    \Let{$E$}{$\{ (f{:}a{:}m') \}$}
    \State \textbf{stop} $E$, $s'$ \label{line:end-certreq}

   \EndIf

  \EndIf

 \EndIf
 \State \textbf{stop} $\{\},s$
\end{algorithmic} \setlength{\parindent}{1em}

\subsection{Relying Parties} \label{app:relying-parties} 

A relying party $r \in \fAP{RP}$ is a web server modeled as
an atomic DY process $(I^r, Z^r, R^r, s^r_0, N^r)$ with the
address $I^r := \{\mapAddresstoAP(r)\}$. Its initial state
$s^r_0$ contains its domain, the private key associated
with its domain, the public key of LPO, and the set of
service token it has provided. The definition of $R^r$
again follows the description in
Appendix~\ref{app:browserid-sidp-lowlevel}. RP only accepts
messages sent over HTTPS. Whenever $r$ receives a $\mGet$
message, it returns the script $\str{script\_RP\_index}$
(see below) and sets the
\texttt{Strict-Trans\-port-Security} header. If $r$
receives an HTTPS $\mPost$ message, it checks if (1) the
message contains a valid CAP for $r$, and (2) the header of
the message contains an $\str{Origin}$ header which only
contains a single origin and this origin is $r$'s domain
with HTTPS. If this check is successful, $r$ responds with
a token of the form $\an{n, i}$ (sent in the body of the
response), where $i \in \IDs$ is the ID for which the CAP
was issued and $n$ is a freshly chosen nonce. We call, as
mentioned in Section~\ref{sec:analysisbrowseridsidp},
$\an{n, i}$ an \emph{RP service token (for ID $i$).} As
mentioned, $r$ keeps a list of such tokens in its
state. Intuitively, a client in possession of such a token
can use the service of $r$ for ID $i$ (e.g., access data of
$i$ at $r$).

We now provide the formal definition of $r$ as an atomic DY
process $(I^r, Z^r, R^r, s^r_0, N^r)$. As mentioned, we
define $I^r = \{\mapAddresstoAP(r)\}$. Next, we define the
set $Z^r$ of states of $r$ and the initial state $s^r_0$ of
$r$.

\begin{definition}
  A state $s\in Z^r$ of an RP $r$ is a term of the form
  $\langle\mi{nonces}$, $\mi{sslkey}$, $\mi{domain}$,
  $\mi{pubkLPO}$, $\mi{serviceTokens}\rangle$ where
  $\mi{nonces} \subsetPairing \nonces$ (used nonces),
  $\mi{sslkey}=\mapKey(\mapDomain(\LPO))$,
  $\mi{domain}=\mapDomain(r)$,
  $\mi{pubkLPO}=\pub(\mapKey(\mapDomain(\LPO)))$,
  $\mi{serviceTokens}\in\dict{\nonces}{\mathbb{S}}$.

  The \emph{initial state $s^r_0$ of $r$} is a state of $r$
  with $s^r_0.\str{nonces} = \an{}$ and
  $s^r_0.\str{serviceTokens} = \an{}$.
\end{definition}

We now specify the relation $R^r \subseteq (\events \times Z^r ) \times
  (2^\events \times Z^r)$ of $r$. Just like
in Appendix~\ref{sec:descr-web-brows}, we describe this
relation by a non-deterministic algorithm.  We note that we
use the function \textsf{TAKENONCE} introduced in
Section~\ref{app:proceduresbrowser} for this purpose.

\captionof{algorithm}{\label{alg:rp} Relation of a Relying
  Party $R^r$}
\begin{algorithmic}[1]
 \Statex[-1] \textbf{Input:} $(a{:}f{:}m),s$
  \Let{$s'$}{$s$}

  \LetST{$m_\text{dec},k'$}{$\an{m_\text{dec},k'} \equiv \dec{m}{s'.\str{sslkey}}$}{\textbf{stop}~$\{\},s$}

  \LetST{$n$, $\mi{method}$, $\mi{path}$, $\mi{params}$, $\mi{headers}$, $\mi{body}$}
        {$
           \hreq{ nonce=n,
            method=\mi{method},
            host=s'.\str{domain},
            path=\mi{path},
            parameters=\mi{params},
            headers=\mi{headers},
            body=\mi{body}
          } \equiv  m_\text{dec}$}{\textbf{stop}~$\{\},s$}

  \If{$\mi{method} \equiv \mGet$} \Comment{Deliver RP's index script}

   \LetBreak{$m'$}{$\ehresp{ nonce=n,
                         status=200,
                         headers=\!\an{\an{\str{Strict \mhyphen Transport \mhyphen Security},\!\True}},
                         body=\an{\str{script\!\_\!RP\!\_\!index},\mi{initState_{rp\!\_\!index}}}
                       }{k'}$}
                     \StatexNoIndent where $\mi{initState_{rp\_index}}$ is the initial scriptstate of $\mi{script\_RP\_index}$ according to Definition~\ref{def:scriptstaterpindex}.
   \Let{$E$}{$\{  (f{:}a{:}m')
        \}$}

    \State \textbf{stop} $E$, $s'$
  
  \ElsIf{$(\mi{method}\equiv\mPost) \wedge (\mi{headers}\equiv\an{\an{\str{Origin},\an{s'.\str{domain},\str{S}}}})$} \Comment{Check received CAP}

     \LetST{$\mi{uc}$, $\mi{ia}$}{$\an{\mi{uc},\mi{ia}} \equiv \mi{body}$}{\textbf{stop} $\{\},s$}  \Comment{Extract UC and IA}
     \Let{$i$}{$\proj{1}{\unsig{\mi{uc}}}$}  \Comment{Extract ID from UC}
     \Let{$\mi{pku}$}{$\proj{2}{\unsig{\mi{uc}}}$} \Comment{Extract pubkey from UC}
     \Let{$o$}{$\unsig{\mi{ia}}$} \Comment{Extract audience from IA}

     \If{$(\checksig{\mi{uc}}{s'.\str{pubkLPO}} \equiv \True) \wedge (\checksig{\mi{ia}}{\mi{pku} } \equiv \True \wedge o \equiv \an{s'.\str{domain},\https})$} \label{line:rp-checksig} \label{line:rp-accepts-https-ia-only} 

        \Let{$n', s'$}{\textsf{TAKENONCE}($s'$)} \Comment{Issue service token}
        \Let{$\comp{s'}{\str{serviceTokens}}$}{$\comp{s'}{\str{serviceTokens}} + \an{n',i}$}
        \Let{$m'$}{$\ehresp{ nonce=n,
                         status=200,
                         headers=\an{},
                         body=\an{n',i}
                       }{k'}$}
        \Let{$E$}{$\{ (f{:}a{:}m')
        \}$}
        \State \textbf{stop} $E$, $s'$
     \EndIf

  \EndIf
 \State \textbf{stop} $\{\}$, $s$

\end{algorithmic} \setlength{\parindent}{1em}

\subsection{BrowserID Scripts}  \label{app:browserid-scripts}
As already mentioned in Section~\ref{sec:browerIDmodel},
the set $\scriptset$ of the web system
$\bidwebsystem=(\bidsystem, \scriptset, \mathsf{script},
E_0)$ consists of the scripts $\Rasp$,
$\mi{script\_RP\_index}$, $\mi{script\_LPO\_cif}$, and
$\mi{script\_LPO\_ld}$ with their string representations
$\str{att\_script}$, $\str{script\_RP\_index}$,
$\str{script\_LPO\_cif}$, and $\str{script\_LPO\_ld}$
(defined by $\mathsf{script}$). The script $\Rasp$ is the
attacker script (see Section~\ref{sec:websystem}). The
formal model of the latter two scripts follows the
description in Appendix~\ref{app:browserid-sidp-lowlevel} in a
straightforward way. The script $\mi{script\_RP\_index}$
defines the script of the RP index page. In reality, this
page has its own script(s) and includes a script from
LPO. In our model, we combine both scripts to
$\mi{script\_RP\_index}$.  In particular, this script is
responsible for creating the CIF and the LD
\iframes/subwindows, whose content (scripts) are loaded
from LPO.

In what follows, the scripts $\mi{script\_RP\_index}$,
$\mi{script\_LPO\_cif}$, and $\mi{script\_LPO\_ld}$ are
defined formally. First, we introduce some notation and
helper functions.

\subsubsection{Some Notation and Helper Functions}\label{sec:some-notation-helper}

In the formal description of the scripts we use an
abbreviation for URLs for LPO. We write
$\mathsf{URL}^\LPO_\mi{path}$ to describe the
following URL term: $\an{\tUrl, \https,
  \domLPO,\mi{path},\an{}}$
Also, we call
$\mathsf{origin}_\LPO$ the origin of LPO which
describes the following origin term:
$\an{\domLPO,\https}$.

The \ls under the origin of LPO used by the scripts
$\mi{script\_LPO\_cif}$ and $\mi{script\_LPO\_ld}$ is
organized as follows: it is a dictionary containing only
one entry. This entry consists of the key $\str{siteInfo}$
and (as its value) a dictionary where this dictionary has
origins as keys with IDs as values indicating that a
certain ID (of the user) is logged in at the referenced
origin. Here is an example a possible \ls.

\begin{example} \label{ex:lpo-localstorage}
  \begin{align}
    \an{
        \an{\str{siteInfo},\an{
                                \an{ \an{\str{domain_{RP1}},\str{S}},\str{id_1} },
                                \an{ \an{\str{domain_{RP2}},\str{S}},\str{id_1} },
                                \an{ \an{\str{domain_{RP3}},\str{S}},\str{id_2} }
                              }
           }
       }
  \end{align} This example shows a \ls under the origin of $\LPO$, indicating that the user is logged in at $\str{domain_{RP1}}$ and $\str{domain_{RP2}}$ with $\str{id_1}$ and at $\str{domain_{RP3}}$ with $\str{id_2}$ (using HTTPS).
\end{example}

In order to simplify the description of the scripts, several helper functions are used.

\myparagraph{CHOOSEINPUT.}  As explained in
Section~\ref{sec:web-browsers}, the state of a document
contains a term, say $\mi{scriptinputs}$, which records the
input this document has obtained so far (via \xhrs and
\pms). If the script of the document is activated, it will
typically need to pick one input message from
$\mi{scriptinputs}$ and record which input it has already
processed. For this purpose, the function
$\mathsf{CHOOSEINPUT}(s',\mi{scriptinputs})$ is used, where
$s'$ denotes the scripts current state. It saves the
indexes of already handled messages in the scriptstate $s'$
and chooses a yet unhandled input message from
$\mi{scriptinputs}$. The index of this message is then
saved in the scriptstate (which is returned to the script).

\captionof{algorithm}{\label{alg:chooseinput} Choose an unhandled input message for a script}
\begin{algorithmic}[1]
  \Function{$\mathsf{CHOOSEINPUT}$}{$s',\mi{scriptinputs}$}
  \LetST{$\mi{iid}$}{$\mi{iid} \in
    \{1,\cdots,|\mi{scriptinputs}|\} \wedge \mi{iid} \not\inPairing
    s'.\str{handledInputs}$}{\Return$(\bot,s')$}
  \Let{$\mi{input}$}{$\proj{\mi{iid}}{\mi{scriptinputs}}$}
  \Let{$s'.\str{handledInputs}$}{$s'.\str{handledInputs} \plusPairing
    \mi{iid}$}

  \State \Return$(\mi{input}, s')$
  \EndFunction
\end{algorithmic} \setlength{\parindent}{1em}

\myparagraph{PARENTWINDOW.} To determine the nonce
referencing the parent window in the browser, the function
$\mathsf{PARENTWINDOW}(\mi{tree}, \mi{docnonce})$ is
used. It takes the term $\mi{tree}$, which is the (partly
cleaned) tree of browser windows the script is able to see
and the document nonce $\mi{docnonce}$, which is the nonce
referencing the current document the script is running in,
as input. It outputs the nonce referencing the window which
directly contains in its subwindows the window of the
document referenced by $\mi{docnonce}$. If there is no such
window (which is the case if the script runs in a document
of a top-level window), $\mathsf{PARENTWINDOW}$ returns
$\bot$.

\myparagraph{SUBWINDOWS.} This function takes a term
$\mi{tree}$ and a document nonce $\mi{docnonce}$ as input
just as the function above. If $\mi{docnonce}$ is not a
reference to a document contained in $\mi{tree}$, then
$\mathsf{SUBWINDOWS}(\mi{tree},\mi{docnonce})$ returns
$\an{}$. Otherwise, let $\an{\mi{docnonce}$, $\mi{origin}$,
  $\mi{script}$, $\mi{scriptstate}$, $\mi{scriptinput}$,
  $\mi{subwindows}$, $\mi{active}}$ denote the subterm of
$\mi{tree}$ corresponding to the document referred to by
$\mi{docnonce}$. Then,
$\mathsf{SUBWINDOWS}(\mi{tree},\mi{docnonce})$ returns
$\mi{subwindows}$.

\myparagraph{AUXWINDOW.} This function takes a term
$\mi{tree}$ and a document nonce $\mi{docnonce}$ as input
as above. From all window terms in $\mi{tree}$ that have
the window containing the document identified by
$\mi{docnonce}$ as their opener, it selects one
non-deterministically and returns its nonce. If there is no such
window, it returns the nonce of the window containing
$\mi{docnonce}$.

\myparagraph{OPENERWINDOW.} This function takes a
term $\mi{tree}$ and a document nonce $\mi{docnonce}$ as
input as above. It returns the window nonce of the opener
window of the window that contains the document identified
by $\mi{docnonce}$. Recall that the nonce identifying the
opener of each window is stored inside the window term. If
no document with nonce $\mi{docnonce}$ is found in the tree
$\mi{tree}$, $\notdef$ is returned.

\myparagraph{GETWINDOW.} This function takes a term
$\mi{tree}$ and a document nonce $\mi{docnonce}$ as input
as above. It returns the nonce of the window containing $\mi{docnonce}$.

\myparagraph{GETORIGIN.} To extract the origin of a
document, the function
$\mathsf{GETORIGIN}(\mi{tree},\mi{docnonce})$ is used. This
function searches for the document with the identifier
$\mi{docnonce}$ in the (cleaned) tree $\mi{tree}$ of the
browser's windows and documents. It returns the origin $o$
of the document. If no document with nonce $\mi{docnonce}$
is found in the tree $\mi{tree}$, $\notdef$ is returned.

\subsubsection{$\mi{script\_LPO\_cif}$}\label{app:scriptLPOcif}

As defined in Section~\ref{sec:websystem}, a script is a
relation that takes as input a term and a set of nonces it
may use. It outputs a new term. As specified in
Section~\ref{sec:web-browsers} (Triggering the Script of a
Document (\textbf{\hlExp{$m = \trigger$},
  \hlExp{$\mi{action} = 1$}})) and formally specified in
Algorithm~\ref{alg:runscript}, the input term is provided
by the browser. It contains the current internal state of
the script (which we call \emph{scriptstate} in what
follows) and additional information containing all browser
state information the script has access to, such as the
input the script has obtained so far via \xhrs and \pms,
information about windows, etc. The browser expects the
output term to have a specific form, as also specified in
Section~\ref{sec:web-browsers} and
Algorithm~\ref{alg:runscript}. The output term contains,
among other information, the new internal scriptstate.

As for $\mi{script\_LPO\_cif}$, this script models the
script run in the CIF, as sketched in
Appendix~\ref{app:browserid-sidp-lowlevel}.

We first describe the structure of the internal scriptstate
of the script $\mi{script\_LPO\_cif}$.

\begin{definition} \label{def:scriptstatelpocif}
A scriptstate $s$ of $\mi{script\_LPO\_cif}$ is a term of the form $\langle q$, $\mi{parentOrigin}$, $\mi{loggedInUser}$, $\mi{pause}$, $\mi{context}$, $\mi{key}$, $\mi{handledInputs}$, $\mi{refXHRctx}$, $\mi{refXHRcert} \rangle$ where $q \in \mathbb{S}$, $\mi{parentOrigin} \in \origins \cup \{\bot\}$, $\mi{loggedInUser} \in \IDs \cup \{\an{},\bot\}$, $\mi{pause} \in \{\True,\bot\}$, $\mi{context} \in \terms$, $\mi{key} \in \nonces \cup \{\bot\}$, $\mi{handledInputs} \subsetPairing \mathbb{N}$, $\mi{refXHRctx},\mi{refXHRcert} \in \nonces \cup \{\bot\}$.

The initial state $\mi{initState_{cif}}$ of
$\mi{script\_LPO\_cif}$ is the state
$\an{\str{init},\bot,\bot,\bot,\bot,\bot,\an{},\bot,\bot}$.
\end{definition}

Before we provide the formal specification of the relation
that defines the behavior of $\mi{script\_LPO\_cif}$, we
present an informal description. The behavior mainly
depends on the state $q$ the script is in.

\begin{description}
\item[$q=\str{init}$] is the initial state. It's only
  transition handles no input and outputs a \pm
  \texttt{cifready} to its parent window and transitions to
  $\str{default}$.

\item[$q=\str{default}$] is the state to which
  $\str{script\_LPO\_cif}$ always returns to. This state
  handles all \pms the CIF expects to receive. If the \pm
  received was sent from the parent window of the CIF, it
  behaves as follows:

  \begin{description}

  \item[\pm \texttt{loaded}] records the sender's origin of
    the received \pm as the remote origin in the scriptstate. Also, an ID, which represents the assumption of
    the sender on who it believes to be logged in, is saved
    in the scriptstate. If the $\mi{pause}$ flag in the
    scriptstate is $\True$ it transitions to the state
    $\str{default}$. Otherwise, it is checked, if the
    current context in the scriptstate is $\bot$. If the
    check is true, the script transitions to the state
    $\str{fetchContext}$, or to the state
    $\str{checkAndEmit}$ otherwise.

  \item[\pm \texttt{dlgRun}] sets $\mi{pause}$ flag in the
    scriptstate to $\True$ and transitions to
    $\str{default}$.

  \item[\pm \texttt{dlgCmplt}] sets the $\mi{pause}$ flag
    in the scriptstate to $\bot$.
    It then transitions to the state $\str{fetchContext}$.

  \item[\pm \texttt{loggedInUser}] has to contain an
    ID. This ID is saved in the scriptstate and then the
    script transitions to $\str{default}$.

  \item[\pm \texttt{logout}] removes the entry for the RP
    (recorded in the scriptstate) from the \ls and then
    transitions to the state $\str{sendLogout}$.
  \end{description}

\item[$q=\str{fetchContext}$] sends an \xhr to LPO with a
  $\mGet$ request to the path \texttt{/ctx} and then
  transitions to the state $\str{receiveContext}$.

\item[$q=\str{receiveContext}$] expects an \xhr response as
  input containing the session context. This context is
  saved as the current context in the scriptstate. The
  script transitions to $\str{checkAndEmit}$.

\item[$q=\str{checkAndEmit}$] lets the script transition to
  $\str{requestUC}$ iff (1) some email address is marked as
  logged in at RP in the \ls, (2) if an email address is
  recorded in the current scriptstate, this email
  address differs from the one recorded in the \ls, and (3)
  the user is marked as logged in in the current
  context. Otherwise, if the email address recorded in the
  current scriptstate is $\an{}$, the script transitions to
  $\str{default}$, else it transitions to $\str{sendLogout}$.

\item[$q=\str{requestUC}$] creates a new private key (by
  taking a fresh nonce), stores the key in the scriptstate,
  and sends out an \xhr $\mPost$ request with the ID
  recorded in the localStorage for the parent window's
  origin and the public key (which can be derived from the
  private key) to $\LPO$ to get a UC. The script transitions
  to $\str{receiveUC}$.

\item[$q=\str{receiveUC}$] receives an \xhr response (from
  $\LPO$) containing a UC. It creates an IA for the parent
  window's origin, combines the UC and the IA to a CAP, and
  sends the CAP as \texttt{login} \pm to the parent
  window. The script then transitions back to the
  $\str{default}$ state.

\item[$q=\str{sendLogout}$] sends a \texttt{logout} \pm to
  the parent document and then the script transitions to
  the $\str{default}$ state.

\end{description}

We now specify the relation $\mi{script\_LPO\_cif}
\subseteq (\terms \times 2^{\nonces})\times \terms$ of the
CIF's scripting process formally. Just like in
Appendix~\ref{sec:descr-web-brows}, we describe this
relation by a non-deterministic algorithm. 

Just like all scripts, as explained in
Section~\ref{sec:web-browsers} (see also
Algorithm~\ref{alg:runscript} for the formal
specification), the input term this script obtains from the
browser contains the cleaned tree of the browser's windows
and documents $\mi{tree}$, the nonce of the current
document $\mi{docnonce}$, its own scriptstate
$\mi{scriptstate}$ (as defined in
Definition~\ref{def:scriptstatelpocif}), a sequence of all
inputs $\mi{scriptinput}$ (also containing already handled
inputs), a dictionary $\mi{cookies}$ of all accessible
cookies of the document's domain, the \ls
$\mi{localStorage}$ belonging to the document's origin, the
secrets $\mi{secret}$ of the document's origin, and a set
$\mi{nonces}$ of fresh nonces as input. The script returns
a new scriptstate $s'$, a new set of cookies
$\mi{cookies'}$, a new \ls $\mi{localStorage'}$, and a term
$\mi{command}$ denoting a command to the browser.

\captionof{algorithm}{\label{alg:scriptlpocif} Relation of $\mi{script\_LPO\_cif}$ }
\begin{algorithmic}[1]
\Statex[-1] \textbf{Input:} $\an{\mi{tree},\mi{docnonce},\mi{scriptstate},\mi{scriptinputs},\mi{cookies},\mi{localStorage},\mi{sessionStorage},\mi{secret}},\mi{nonces}$ 
\Let{$s'$}{$\mi{scriptstate}$}
\Let{$\mi{cookies}'$}{$\mi{cookies}$}
\Let{$\mi{localStorage}'$}{$\mi{localStorage}$} 

\Switch{$s'.\str{q}$}
 \Case{$\str{init}$}
  \Let{$\mi{command}$}{$\an{\tPostMessage,\textsf{PARENTWINDOW}(\mi{tree}, \mi{docnonce}),\an{\str{cifready},\an{}},\bot}$}
  \Let{$s'.\str{q}$}{$\str{default}$}
  \State \textbf{stop} $\an{s',\mi{cookies}',\mi{localStorage}',\mi{sessionStorage},\mi{command}}$
 \EndCase

 \Case{$\str{default}$}
  \Let{$\mi{input},s'$}{\textsf{CHOOSEINPUT}($s',\mi{scriptinputs}$)}
  \If{$\proj{1}{\mi{input}} \equiv \tPostMessage$}
   \Let{$\mi{senderWindow}$}{$\proj{2}{\mi{input}}$}
   \Let{$\mi{senderOrigin}$}{$\proj{3}{\mi{input}}$}
   \Let{$m$}{$\proj{4}{\mi{input}}$}%

   \If{$\mi{senderWindow} \equiv \textsf{PARENTWINDOW}(\mi{tree}, \mi{docnonce})$}
   \Switch{$m$}
    \Case{$\an{\str{loaded},\mi{id}}$}
      \Let{$s'.\str{parentOrigin}$}{$\mi{senderOrigin}$} \label{line:set-parent-origin}
      \Let{$s'.\str{loggedInUser}$}{$\mi{id}$}
      \If{$s'.\str{pause} \equiv \True$}
       \State \textbf{stop} $\an{s',\mi{cookies}',\mi{localStorage}',\mi{sessionStorage},\an{}}$
      \ElsIf{$s'.\str{context} \equiv \bot$}
       \Let{$s'.\str{q}$}{$\str{fetchContext}$}
       \State \textbf{stop} $\an{s',\mi{cookies}',\mi{localStorage}',\mi{sessionStorage},\an{}}$
      \Else
       \Let{$s'.\str{q}$}{$\str{checkAndEmit}$}
       \State \textbf{stop} $\an{s',\mi{cookies}',\mi{localStorage}',\mi{sessionStorage},\an{}}$
      \EndIf
    \EndCase

    \Case{$\an{\str{dlgRun},\an{}}$}
      \Let{$s'.\str{pause}$}{$\True$}
      \State \textbf{stop} $\an{s',\mi{cookies}',\mi{localStorage}',\mi{sessionStorage},\an{}}$ 
    \EndCase

    \Case{$\an{\str{dlgCmplt},\an{}}$}
      \Let{$s'.\str{pause}$}{$\bot$}
      \Let{$s'.\str{q}$}{$\str{fetchContext}$}
      \State \textbf{stop} $\an{s',\mi{cookies}',\mi{localStorage}',\mi{sessionStorage},\an{}}$ 
    \EndCase

    \Case{$\an{\str{loggedInUser},\mi{id}}$}
      \Let{$s'.\str{loggedInUser}$}{$\mi{id}$}
      \State \textbf{stop} $\an{s',\mi{cookies}',\mi{localStorage}',\mi{sessionStorage},\an{}}$ 
   \EndCase

    \Case{$\an{\str{logout},\an{}}$}
    \EndCase
   \EndSwitch
   \EndIf
  \EndIf
 \EndCase

 \Case{$\str{fetchContext}$}
  \LetND{$s'.\str{refXHRctx}$}{$\mi{nonces}$}
  \Let{$\mi{command}$}{$\an{\tXMLHTTPRequest,\textsf{URL}^\LPO_\str{/ctx},\mGet,\an{},s'.\str{refXHRctx}}$} \label{line:ctx-over-https}
  \Let{$s'.\str{q}$}{$\str{receiveContext}$}
  \State \textbf{stop} $\an{s',\mi{cookies}',\mi{localStorage}',\mi{sessionStorage},\mi{command}}$
 \EndCase

 \Case{$\str{receiveContext}$}
  \Let{$\mi{input},s'$}{\textsf{CHOOSEINPUT}($s',\mi{scriptinputs}$)}
  \If{$(\proj{1}{\mi{input}} \equiv \tXMLHTTPRequest) \wedge (\proj{3}{\mi{input}} \equiv s'.\str{refXHRctx})$}
   \Let{$s'.\str{context}$}{$\proj{2}{\mi{input}}$} \label{line:store-context}
   \Let{$s'.\str{q}$}{$\str{checkAndEmit}$}
   \State \textbf{stop} $\an{s',\mi{cookies}',\mi{localStorage}',\mi{sessionStorage},\an{}}$
  \EndIf

 \EndCase

 \Case{$\str{checkAndEmit}$}
  \Let{$\mi{lid}$}{$\mi{localStorage}'[\str{siteInfo}][s'.\str{parentOrigin}]$}
  \If{$(\mi{lid} \not\equiv \an{}) \wedge (s'.\str{loggedInUser} \notin \{\an{},\bot\} \Rightarrow (s'.\str{loggedInUser} \not\equiv \mi{lid})) \wedge (\proj{1}{s'.\str{context}} \not\equiv \an{})$}
   \Let{$s'.\str{q}$}{$\str{requestUC}$}
   \State \textbf{stop} $\an{s',\mi{cookies}',\mi{localStorage}',\mi{sessionStorage},\an{}}$
  \ElsIf{$s'.\str{loggedInUser} \equiv \an{}$}
   \Let{$s'.\str{q}$}{$\str{default}$}
   \State \textbf{stop} $\an{s',\mi{cookies}',\mi{localStorage}',\mi{sessionStorage},\an{}}$
  \Else
   \Let{$s'.\str{q}$}{$\str{sendLogout}$}
   \State \textbf{stop} $\an{s',\mi{cookies}',\mi{localStorage}',\mi{sessionStorage},\an{}}$
  \EndIf
 \EndCase

 \Case{$\str{requestUC}$}
  \Let{$\mi{id}$}{$\mi{localStorage}'[\str{siteInfo}][s'.\str{parentOrigin}]$}
  \LetND{$s'.\str{key}$}{$\mi{nonces}$} \label{line:choose-key}
  \Let{$\mi{body}$}{$\an{\mi{id},\pub(s'.\str{key}),s'.\str{context}.\str{xsrfToken}}$}
  \LetND{$s'.\str{refXHRcert}$}{$\mi{nonces} \setminus \{s'.\str{key}\}$}
  \Let{$\mi{command}$}{$\an{\tXMLHTTPRequest,\textsf{URL}^\LPO_\str{/certreq},\mPost,\mi{body},s'.\str{refXHRcert}}$}
  \Let{$s'.\str{q}$}{$\str{receiveUC}$}
  \State \textbf{stop} $\an{s',\mi{cookies}',\mi{localStorage}',\mi{sessionStorage},\mi{command}}$
 \EndCase

 \Case{$\str{receiveUC}$}
  \Let{$\mi{input},s'$}{\textsf{CHOOSEINPUT}($s',\mi{scriptinputs}$)}
  \If{$(\proj{1}{\mi{input}} \equiv \tXMLHTTPRequest) \wedge (\proj{3}{\mi{input}} \equiv s'.\str{refXHRcert})$}
   \Let{$\mi{uc}$}{$\proj{2}{\mi{input}}$}
   \Let{$\mi{ia}$}{$\sig{s'.\str{parentOrigin}}{s'.\str{key}}$}
   \Let{$\mi{cap}$}{$\an{\mi{uc},\mi{ia}}$}
   \Let{$\mi{command}$}{}
 \Statex $\an{\tPostMessage,\textsf{PARENTWINDOW}(\mi{tree}, \mi{docnonce}),\an{\str{login},\mi{cap}},s'.\str{parentOrigin}}$ \label{line:create-login-pm}
   \Let{$s'.\str{q}$}{$\str{default}$}
   \State \textbf{stop} $\an{s',\mi{cookies}',\mi{localStorage}',\mi{sessionStorage},\mi{command}}$ \label{line:send-login-pm}
  \EndIf
 \EndCase

 \Case{$\str{sendLogout}$}
  \Let{$\mi{command}$}{$\an{\tPostMessage,\textsf{PARENTWINDOW}(\mi{tree}, \mi{docnonce}),\an{\str{logout},\an{}},\bot}$}
  \Let{$s'.\str{q}$}{$\str{default}$}
  \State \textbf{stop} $\an{s',\mi{cookies}',\mi{localStorage}',\mi{sessionStorage},\mi{command}}$
 \EndCase

\EndSwitch

\State \textbf{stop} $\an{\mi{scriptstate},\mi{cookies},\mi{localStorage},\mi{sessionStorage},\an{}}$

\end{algorithmic} \setlength{\parindent}{1em}

\subsubsection{$\mi{script\_LPO\_ld}$}\label{app:scriptLPOld}

The script $\mi{script\_LPO\_ld}$ models the script that
runs in the LD. Its formal specification, presented next,
follows the one presented above for
$\mi{script\_LPO\_cif}$.

\begin{definition}\label{def:scriptstatelpold}
  A \emph{scriptstate $s$ of $\mi{script\_LPO\_ld}$} is a term of
  the form $\langle q$, $\mi{requestOrigin}$,
  $\mi{context}$, $\mi{key}$, $\mi{handledInputs}$,
  $\mi{refXHRctx}$, $\mi{refXHRauth}$, $\mi{refXHRcert}
  \rangle$ with $q \in \mathbb{S}$, $\mi{requestOrigin} \in
  \origins \cup \{\bot\}$, $\mi{context} \in \terms$,
  $\mi{key} \in \nonces \cup \{\bot\}$, $\mi{handledInputs}
  \subsetPairing \mathbb{N}$,
  $\mi{refXHRctx},\mi{refXHRauth}, \mi{refXHRcert} \in
  \nonces \cup \{\bot\}$.

  The initial state $\mi{initState_{ld}}$ is the state
  $\an{\str{init},\bot,\bot,\bot,\an{},\bot,\bot,\bot}$.
\end{definition}

Before we provide the formal specification of the relation
that defines the behavior of $\mi{script\_LPO\_ld}$, we
present an informal description. The behavior mainly
depends on the state $q$ the script is in.

\begin{description}
\item[$q\equiv\str{init}$] is the initial state. Its only
  transition takes no input and outputs a \pm
  \texttt{ldready} to its parent window and transitions to
  $\str{start}$.

\item[$q\equiv\str{start}$] expects a \texttt{request} \pm. The
  sender's origin of this \pm is recorded as the requesting
  origin in the scriptstate. An \xhr is sent to $\LPO$ with
  a $\mGet$ request to the path \texttt{/ctx} and then the
  script transitions to the state $\str{receiveContext}$.

\item[$q\equiv\str{receiveContext}$] expects an \xhr response as
  input containing the session context. This context is
  saved as the current context in the scriptstate. If the
  received context contains $\an{}$ as the ID list, the
  script transitions to the state
  $\str{requestAuth}$. Else, the script transitions to
  $\str{requestUC}$.

\item[$q\equiv\str{requestAuth}$] sends an \xhr $\mPost$ request
  to $\LPO$ with the path \texttt{/auth} containing a
  browser's secret. The script then transitions to the
  state $\str{receiveAuth}$.

\item[$q\equiv\str{receiveAuth}$] expects an \xhr response as
  input containing $\True$. The script then sends an \xhr
  to $\LPO$ with a $\mGet$ request to the path \texttt{/ctx}
  and then transitions to the state $\str{receiveContext}$.

\item[$q\equiv\str{requestUC}$] chooses (non-deterministically) an id,
  chooses a fresh private key and sends the id and the
  public key (corresponding to the private key) as an \xhr
  $\mPost$ request to $\LPO$ with the path
  \texttt{/certreq}. The script then transitions to
  $\str{receiveUC}$

\item[$q\equiv\str{receiveUC}$] receive UC from $\LPO$, create
  IA, combine with UC to CAP, record ID as logged in at the
  requester's origin. Send CAP in \pm to parent. Go to
  state $\str{null}$

\item[$q\equiv\str{null}$] do nothing.

\end{description}

We now formally specify the relation $\mi{script\_LPO\_ld}
\subseteq (\terms \times 2^{\nonces})\times \terms$ of the
LD's scripting process. Just like in
Appendix~\ref{sec:descr-web-brows}, we describe this
relation by a non-deterministic algorithm. Like all
scripts, the input term given to this script is determined
by the browser and the browser expects a term of a specific
form (see Algorithm~\ref{alg:runscript}).

\captionof{algorithm}{\label{alg:scriptlpold} Relation of $\mi{script\_LPO\_ld}$ }
\begin{algorithmic}[1]
\Statex[-1] \textbf{Input:} $\an{\mi{tree},\mi{docnonce},\mi{scriptstate},\mi{scriptinputs},\mi{cookies},\mi{localStorage},\mi{sessionStorage},\mi{secret}},\mi{nonces}$
\Let{$s'$}{$\mi{scriptstate}$}
\Let{$\mi{cookies}'$}{$\mi{cookies}$}
\Let{$\mi{localStorage}'$}{$\mi{localStorage}$}

\Switch{$s'.\str{q}$}
 \Case{$\str{init}$}
  \Let{$\mi{command}$}{$\an{\tPostMessage,\textsf{OPENERWINDOW}(\mi{tree}, \mi{docnonce}),\an{\str{ldready},\an{}},\bot}$}
  \Let{$s'.\str{q}$}{$\str{start}$}
  \State \textbf{stop} $\an{s',\mi{cookies}',\mi{localStorage}',\mi{sessionStorage},\mi{command}}$
 \EndCase

 \Case{$\str{start}$}
  \Let{$\mi{input},s'$}{\textsf{CHOOSEINPUT}($s',\mi{scriptinputs}$)}
  \If{$\proj{1}{\mi{input}} \equiv \tPostMessage$}
   \Let{$\mi{senderWindow}$}{$\proj{2}{\mi{input}}$}
   \Let{$\mi{senderOrigin}$}{$\proj{3}{\mi{input}}$}
   \Let{$m$}{$\proj{4}{\mi{input}}$}

   \If{$m \equiv \an{\str{request},\an{}}$}
    \Let{$s'.\str{requestOrigin}$}{$\mi{senderOrigin}$}
    \LetND{$s'.\str{refXHRctx}$}{$\mi{nonces}$}
    \Let{$\mi{command}$}{$\an{\tXMLHTTPRequest,\textsf{URL}^\LPO_\str{/ctx},\mGet,\an{},s'.\str{refXHRctx}}$} \label{line:ctx-over-https-2}
    \Let{$s'.\str{q}$}{$\str{receiveContext}$}
    \State \textbf{stop} $\an{s',\mi{cookies}',\mi{localStorage}',\mi{sessionStorage},\mi{command}}$
   \EndIf
  \EndIf
 \EndCase

 \Case{$\str{receiveContext}$}
  \Let{$\mi{input},s'$}{\textsf{CHOOSEINPUT}($s',\mi{scriptinputs}$)}
  \If{$(\proj{1}{\mi{input}} \equiv \tXMLHTTPRequest) \wedge (\proj{3}{\mi{input}} \equiv s'.\str{refXHRctx})$}
   \Let{$s'.\str{context}$}{$\proj{2}{\mi{input}}$} \label{line:store-context-2}
   \If{$\proj{1}{s'.\str{context}} \equiv \an{}$}
    \Let{$s'.\str{q}$}{$\str{requestAuth}$}
    \State \textbf{stop} $\an{s',\mi{cookies}',\mi{localStorage}',\mi{sessionStorage},\an{}}$
   \Else
    \Let{$s'.\str{q}$}{$\str{requestUC}$}
    \State \textbf{stop} $\an{s',\mi{cookies}',\mi{localStorage}',\mi{sessionStorage},\an{}}$
   \EndIf
  \EndIf
 \EndCase

 \Case{$\str{requestAuth}$}
  \Let{$\mi{body}$}{$\an{\mi{secret},s'.\str{context}.\str{xsrfToken}}$}
  \LetND{$s'.\str{refXHRauth}$}{$\mi{nonces}$}
  \Let{$\mi{command}$}{$\an{\tXMLHTTPRequest,\textsf{URL}^\LPO_\str{/auth},\mPost,\mi{body},s'.\str{refXHRauth}}$}
  \Let{$s'.\str{q}$}{$\str{receiveContext}$}
  \State \textbf{stop} $\an{s',\mi{cookies}',\mi{localStorage}',\mi{sessionStorage},\mi{command}}$ \label{line:lpo-secrets-stop}
 \EndCase

 \Case{$\str{receiveAuth}$}
  \Let{$\mi{input},s'$}{\textsf{CHOOSEINPUT}($s',\mi{scriptinputs}$)}
  \If{$(\proj{1}{\mi{input}} \equiv \tXMLHTTPRequest) \wedge (\proj{3}{\mi{input}} \equiv s'.\str{refXHRauth})$}
   \If{$\proj{2}{\mi{input}} \equiv \True$}
    \Let{$\mi{command}$}{$\an{\tXMLHTTPRequest,\textsf{URL}^\LPO_\str{/ctx},\mGet,\an{}}$}
    \Let{$s'.\str{q}$}{$\str{receiveContext}$}
    \State \textbf{stop} $\an{s',\mi{cookies}',\mi{localStorage}',\mi{sessionStorage},\mi{command}}$    
   \EndIf
  \EndIf
 \EndCase

 \Case{$\str{requestUC}$}
  \LetND{$\mi{id}$}{$s'.\str{context}.\str{ids}$}
  \LetND{$s'.\str{key}$}{$\mi{nonces}$} \label{line:choose-key-2}
  \Let{$\mi{body}$}{$\an{\mi{id},\pub(s'.\str{key}),s'.\str{context}.\str{xsrfToken}}$}
  \LetND{$s'.\str{refXHRcert}$}{$\mi{nonces}\setminus \{s'.\str{key}\}$}
  \Let{$\mi{command}$}{$\an{\tXMLHTTPRequest,\textsf{URL}^\LPO_\str{/certreq},\mPost,\mi{body},s'.\str{refXHRcert}}$}
  \Let{$s'.\str{q}$}{$\str{receiveUC}$}
  \State \textbf{stop} $\an{s',\mi{cookies}',\mi{localStorage}',\mi{sessionStorage},\mi{command}}$
 \EndCase

 \Case{$\str{receiveUC}$}
  \Let{$\mi{input},s'$}{\textsf{CHOOSEINPUT}($s',\mi{scriptinputs}$)}
  \If{$(\proj{1}{\mi{input}} \equiv \tXMLHTTPRequest) \wedge (\proj{3}{\mi{input}} \equiv s'.\str{refXHRcert})$}
   \Let{$\mi{uc}$}{$\proj{2}{\mi{input}}$}
   \Let{$\mi{ia}$}{$\sig{s'.\str{requestOrigin}}{s'.\str{key}}$}
   \Let{$\mi{cap}$}{$\an{\mi{uc},\mi{ia}}$}
   \Let{$\mi{command}$}{}
 \Statex $\an{\tPostMessage,\textsf{OPENERWINDOW}(\mi{tree}, \mi{docnonce}),\an{\str{response},\mi{cap}},s'.\str{requestOrigin}}$ \label{line:create-response-pm}
   \Let{$s'.\str{q}$}{$\str{null}$}
   \State \textbf{stop} $\an{s',\mi{cookies}',\mi{localStorage}',\mi{sessionStorage},\mi{command}}$ \label{line:send-login-pm-2}
  \EndIf
 \EndCase

\EndSwitch

\State \textbf{stop} $\an{\mi{scriptstate},\mi{cookies},\mi{localStorage},\mi{sessionStorage},\an{}}$
\end{algorithmic} \setlength{\parindent}{1em}

\subsubsection{$\mi{script\_RP\_index}$}\label{app:scriptrpindex}

The script $\mi{script\_RP\_index}$ models the script that
is run by an RP. Its formal specification, presented next,
follows the one presented for the other scripts above.

\begin{definition}\label{def:scriptstaterpindex}
  A scriptstate $s$ of $\mi{script\_RP\_index}$ is a term
  of the form $\langle q$, $\mi{CIFindex}$, $\mi{LDindex}$,
  $\mi{dialogRunning}$, $\mi{cap}$, $\mi{handledInputs}$,
  $\mi{refXHRcap} \rangle$ with $q \in \mathbb{S}$,
  $\mi{CIFindex} \in\mathbb{N} \cup \{\bot\}$,
  $\mi{dialogRunning} \in \{\True,\bot\}$, $\mi{cap} \in
  \terms$, $\mi{handledInputs} \subsetPairing \mathbb{N}$,
  $\mi{refXHRcap} \in \nonces \cup \{\bot\}$. 

  We call $s$ the initial scriptstate of
  $\mi{script\_RP\_index}$ iff $s \equiv
  \an{\str{init},\bot,\bot,\bot,\an{},\an{},\bot}$.
\end{definition}

Before we provide the formal specification of the relation
that defines the behavior of $\mi{script\_RP\_index}$, we
present an informal description. The behavior mainly
depends on the state $q$ the script is in.

\begin{description}

\item[$q\equiv\str{init}$] is the initial state. It creates the CIF \iframe and then transitions to $\str{receiveCIFReady}$.

\item[$q\equiv\str{receiveCIFReady}$] expects a
  \texttt{cifready} \pm from the CIF \iframe with origin of
  $\LPO$. It chooses some ID, $\an{}$, or $\bot$ and sends
  this as a \texttt{loaded} \pm to the CIF \iframe with
  receiver's origin set to the origin of
  $\LPO$.\footnote{From the point of view of the real
    scripts running at RP either some ID is considered to
    be logged in (e.g. from some former ``session''), or
    that no one is considered to be logged in ($\an{}$), or
    that $\mi{script\_RP\_index}$ does not know if it
    should consider someone to be logged in ($\bot$). This
    is overapproximated here by allowing
    $\mi{script\_RP\_index}$ to choose
    non-deterministically between these cases.}  It then
  transitions to the state $\str{default}$.

\item[$q\equiv\str{default}$] chooses non-deterministically
  between (1) opening the LD subwindow and then transitions
  to the same state or (2) handling one of the following
  \pms:

  \begin{description}

  \item[\pm \texttt{login}] which has to be sent from the
    CIF with origin of $\LPO$. Handling this \pm stores the CAP
    (contained in the \pm) in the scriptstate and then
    transitions to the $\str{sendCAP}$ state.
  \item[\pm \texttt{logout}] which has to be sent from the
    CIF with origin of $\LPO$. Handling this \pm has no effect
    and results in the same state.
  \item[\pm \texttt{ldready}] which can only be handled
    after the LD has been opened and before a
    \texttt{response} \pm has been received. The
    \texttt{ldready} \pm has to be sent from the origin of
    $\LPO$. The script sends a \texttt{request} \pm to the LD
    and stays in the $\str{default}$ state.
  \item[\pm \texttt{response}] which can only be handled
    after the LD has been opened and before another
    \texttt{response} \pm has been received. The
    \texttt{ldready} \pm has to be sent from the origin of
    $\LPO$. Handling this \pm stores the CAP (contained in
    the \pm) in the scriptstate, closes the LD, and then
    transitions to the $\str{dlgClosed}$ state.
  \end{description}

\item[$q\equiv\str{dlgClosed}$] sends a
  \texttt{loggedInUser} \pm to the CIF and transitions to
  the state $\str{loggedInUser}$.
\item[$q\equiv\str{loggedInUser}$] sends a \texttt{dlgCmplt} \pm
  to the CIF and transitions to the state $\str{sendCAP}$.
\item[$q\equiv\str{sendCAP}$] sends the CAP to RP as a $\mPost$
  \xhr and then transitions to the state
  $\str{receiveServiceToken}$.
\item[$q\equiv\str{receiveServiceToken}$] receives $\an{n,i}$
  from RP, but does not do anything with it. The script then
  transitions to the state $\str{default}$.

\end{description}

We now formally specify the relation $\mi{script\_RP\_index}
\subseteq (\terms \times 2^{\nonces})\times \terms$ of the
RP-Doc's scripting process. Just like in
Appendix~\ref{sec:descr-web-brows}, we describe this
relation by a non-deterministic algorithm. Like all
scripts, the input term given to this script is determined
by the browser and the browser expects a term of a specific
form (see Algorithm~\ref{alg:runscript}). Following
Algorithm~\ref{alg:scriptrpindex}, we provide some more
explanation. 

\captionof{algorithm}{\label{alg:scriptrpindex} Relation of $\mi{script\_RP\_index}$}
\begin{algorithmic}[1]
\Statex[-1] \textbf{Input:} $\an{\mi{tree},\mi{docnonce},\mi{scriptstate},\mi{scriptinputs},\mi{cookies},\mi{localStorage},\mi{sessionStorage},\mi{secret}},\mi{nonces}$
\Let{$s'$}{$\mi{scriptstate}$}
\Let{$\mi{cookies}'$}{$\mi{cookies}$}
\Let{$\mi{localStorage}'$}{$\mi{localStorage}$}

\Switch{$s'.\str{q}$}
 \Case{$\str{init}$}
  \Let{$\mi{command}$}{$\an{\tIframe,\textsf{URL}^\LPO_\str{/cif},\mathsf{GETWINDOW}(\mi{tree},\mi{docnonce})}$}\label{line:cifindex-begin}
  \Let{$s'.\str{q}$}{$\str{receiveCIFReady}$}
  \Let{$\mi{subwindows}$}{$\mathsf{SUBWINDOWS}(\mi{tree},\mi{docnonce})$}
  \Let{$s'.\str{CIFindex}$}{$|\mi{subwindows}|+1$} \Comment{Index of the next subwindow to be created.}
  \State \textbf{stop} $\an{s',\mi{cookies}',\mi{localStorage}',\mi{sessionStorage},\mi{command}}$\label{line:cifindex-end}
 \EndCase

 \Case{$\str{receiveCIFReady}$}
  \Let{$\mi{input},s'$}{\textsf{CHOOSEINPUT}($s',\mi{scriptinputs}$)}
  \If{$\proj{1}{\mi{input}} \equiv
    \tPostMessage$}
   \Let{$\mi{senderWindow}$}{$\proj{2}{\mi{input}}$}
   \Let{$\mi{senderOrigin}$}{$\proj{3}{\mi{input}}$}
   \Let{$m$}{$\proj{4}{\mi{input}}$}
   \Let{$\mi{subwindows}$}{$\mathsf{SUBWINDOWS}(\mi{tree},\mi{docnonce})$}
   \If{$(m \equiv \an{\str{cifready},\an{}}) \wedge (\mi{senderOrigin} \equiv \mathsf{origin_\LPO}) \wedge \breakalgo{5} (\mi{senderWindow} \equiv \proj{s'.\str{CIFindex}}{\mi{subwindows}}.\str{nonce})$}
    \LetND{$\mi{id}$}{$\{\bot,\an{}\} \cup  \IDs $}
    \Let{$\mi{command}$}{$\an{\tPostMessage,\proj{s'.\str{CIFindex}}{\mi{subwindows}},\an{\str{loaded},\mi{id}},\mathsf{origin_\LPO}}$}
    \Let{$s'.\str{q}$}{$\str{default}$}
    \State \textbf{stop} $\an{s',\mi{cookies}',\mi{localStorage}',\mi{sessionStorage},\mi{command}}$
   \EndIf
  \EndIf
 \EndCase

 \Case{$\str{default}$} \label{line:state-default}
  \If{$s'.\str{dialogRunning} \equiv \bot$}
   \LetND{$\mi{choice}$}{$\{\str{openLD},\str{handlePM}\}$}
  \Else
   \Let{$\mi{choice}$}{$\str{handlePM}$}
  \EndIf
  \If{$\mi{choice} \equiv \str{openLD}$} 
   \Let{$s'.\str{dialogRunning}$}{$\True$} \label{line:ldindex-begin}
   \Let{$\mi{command}$}{$\an{\tHref,\textsf{URL}^\LPO_\str{/ld},\wBlank}$}
   \Let{$s'.\str{q}$}{$\str{default}$}
   \State \textbf{stop}
  $\an{s',\mi{cookies}',\mi{localStorage}',\mi{sessionStorage},\mi{command}}$ \label{line:ldindex-end}

  \Else
   \Let{$\mi{input},s'$}{\textsf{CHOOSEINPUT}($s',\mi{scriptinputs}$)}
   \If{$\proj{1}{\mi{input}} \equiv \tPostMessage$}
    \Let{$\mi{senderWindow}$}{$\proj{2}{\mi{input}}$}
    \Let{$\mi{senderOrigin}$}{$\proj{3}{\mi{input}}$}
    \Let{$m$}{$\proj{4}{\mi{input}}$}
    \Let{$\mi{subwindows}$}{$\mathsf{SUBWINDOWS}(\mi{tree},\mi{docnonce})$}
    \If{$\mi{senderOrigin} \equiv \mathsf{origin_\LPO}$}
     \If{$\mi{senderWindow} \equiv \proj{s'.\str{CIFindex}}{\mi{subwindows}}.\str{nonce}$}
      \If{$\proj{1}{m} \equiv \str{login}$}
       \Let{$s'.\str{cap}$}{$\proj{2}{m}$} \label{line:set-cap}
       \Let{$s'.\str{q}$}{$\str{sendCAP}$}
       \State \textbf{stop} $\an{s',\mi{cookies}',\mi{localStorage}',\mi{sessionStorage},\an{}}$
      \ElsIf{$\proj{1}{m} \equiv \str{logout}$}
       \Let{$s'.\str{q}$}{$\str{default}$}
       \State \textbf{stop} $\an{s',\mi{cookies}',\mi{localStorage}',\mi{sessionStorage},\an{}}$
      \EndIf

     \ElsIf{$s'.\str{dialogRunning} \equiv \True$}
      \If{$\proj{1}{m} \equiv \str{ldready}$}                  \Let{$\mi{command}$}{}
\Statex $\an{\tPostMessage, \mathsf{AUXWINDOW}(\mi{tree}, \mi{docnonce}),\an{\str{request},\an{}},\mathsf{origin_\LPO}}$ 
       \Let{$s'.\str{q}$}{$\str{default}$}
       \State \textbf{stop} $\an{s',\mi{cookies}',\mi{localStorage}',\mi{sessionStorage},\mi{command}}$
      \ElsIf{$\proj{1}{m} \equiv \str{response}$}
       \Let{$s'.\str{dialogRunning}$}{$\bot$}
       \Let{$s'.\str{cap}$}{$\proj{2}{m}$} \label{line:set-cap-2}
   \Let{$\mi{command}$}{$\an{\tClose,\mathsf{AUXWINDOW}(\mi{tree},\mi{docnonce})}$} 
       \Let{$s'.\str{q}$}{$\str{dlgClosed}$}
       \State \textbf{stop} $\an{s',\mi{cookies}',\mi{localStorage}',\mi{sessionStorage},\mi{command}}$
      \EndIf
     \EndIf
    \EndIf
   \EndIf
  \EndIf
 \EndCase

 \Case{$\str{dlgClosed}$} 
   \Let{$\mi{subwindows}$}{$\mathsf{SUBWINDOWS}(\mi{tree},\mi{docnonce})$}
   \Let{$\mi{id}$}{$\proj{1}{\unsig{\proj{1}{s'.\str{cap}}}}$} \Comment{Extract ID from CAP.}
\Let{$\mi{command}$}{$\an{\tPostMessage,\proj{s'.\str{CIFindex}}{\mi{subwindows}}.\str{nonce},\an{\str{loggedInUser},\mi{id}},\mathsf{origin_\LPO}}$}
  \Let{$s'.\str{q}$}{$\str{loggedInUser}$}
  \State \textbf{stop} $\an{s',\mi{cookies}',\mi{localStorage}',\mi{sessionStorage},\mi{command}}$
 \EndCase

 \Case{$\str{loggedInUser}$}
   \Let{$\mi{subwindows}$}{$\mathsf{SUBWINDOWS}(\mi{tree},\mi{docnonce})$}
\Let{$\mi{command}$}{$\an{\tPostMessage,\proj{s'.\str{CIFindex}}{\mi{subwindows}}.\str{nonce},\an{\str{dlgCmplt},\an{}},\mathsf{origin_\LPO}}$}
  \Let{$s'.\str{q}$}{$\str{sendCAP}$}
  \State \textbf{stop} $\an{s',\mi{cookies}',\mi{localStorage}',\mi{sessionStorage},\mi{command}}$
 \EndCase

 \Case{$\str{sendCAP}$} \label{line:state-sendcap}
   \LetND{$s'.\str{refXHRcap}$}{$\mi{nonces}$}
   \LetST{$\mi{host}$, $\mi{protocol}$}{$\an{host, protocol} = \mathsf{GETORIGIN}(\mi{tree}, \mi{docnonce})$}{\textbf{stop} $\an{\mi{scriptstate},\mi{cookies},\mi{localStorage},\mi{sessionStorage},\mi{command}}$}
   \Let{$\mi{command}$}{$\an{\tXMLHTTPRequest,\an{\cUrl, \mi{protocol}, \mi{host}, \mathtt{/}, \an{}},\mPost,s'.\str{cap},s'.\str{refXHRcap}}$} \Comment{Relay received CAP to RP.}
   \Let{$s'.\str{q}$}{$\str{receiveServiceToken}$}
   \State \textbf{stop} $\an{s',\mi{cookies}',\mi{localStorage}',\mi{sessionStorage},\mi{command}}$ \label{line:send-cap}
 \EndCase

 \Case{$\str{receiveServiceToken}$}
  \Let{$\mi{input},s'$}{\textsf{CHOOSEINPUT}($s',\mi{scriptinputs}$)}
  \If{$(\proj{1}{\mi{input}} \equiv \tXMLHTTPRequest) \wedge (\proj{3}{\mi{input}} \equiv s'.\str{refXHRcap})$}
   \Let{$s'.\str{q}$}{$\str{default}$}
   \State \textbf{stop} $\an{s',\mi{cookies}',\mi{localStorage}',\mi{sessionStorage},\an{}}$
  \EndIf
 \EndCase

\EndSwitch

\State \textbf{stop} $\an{\mi{scriptstate},\mi{cookies},\mi{localStorage},\mi{sessionStorage},\an{}}$
\end{algorithmic} \setlength{\parindent}{1em}

In Lines~\ref{line:cifindex-begin}--\ref{line:cifindex-end}
and \ref{line:ldindex-begin}--\ref{line:ldindex-end} the
script asks the browser to create iframes. To obtain the
window reference for these iframes, the script first
determines the current number of subwindows and stores it
(incremented by 1) in the scriptstate ($\str{CIFindex}$ and
$\str{LDindex}$, respectively).  When the script is invoked
the next time, the iframe the script asked to be created
will have been added to the sequence of subwindows by the
browser directly following the previously existing
subwindows. The script can therefore access the iframe by
the indexes $\str{CIFindex}$ and $\str{LDindex}$,
respectively.

\renewcommand{\labelenumi}{\arabic{enumi}.}

\section{Security Property}\label{app:securitypropertiesbrowserid}

Formally, the security property for BrowserID is defined as
follows. First note that every RP service token $\an{n,i}$
recorded in RP was created by RP as the result of a unique
HTTPS $\mPost$ request $m$ with a valid CAP for ID $i$. We
refer to $m$ as the \emph{request corresponding to
  $\an{n,i}$}.

\begin{definition}\label{def:security-property} Let $\bidwebsystem$ be a BrowserID web
  system. We say that \emph{$\bidwebsystem$ is secure} if
  for every run $\rho$ of $\bidwebsystem$, every state
  $(S_j, E_j)$ in $\rho$, every $r\in \fAP{RP}$, every RP
  service token of the form $\an{n,i}$ recorded in $r$ in
  the state $S_j(r)$, the following two conditions are
  satisfied:

  \textbf{(A)} If $\an{n,i}$ is derivable from the
  attackers knowledge in $S_j$ (i.e., $\an{n,i} \in
  d_{N^\fAP{attacker}}(S_j(\fAP{attacker}))$), then it
  follows that the browser owning $i$ is fully corrupted in
  $S_j$, i.e., the value of $\mi{isCorrupted}$ is
  $\fullcorrupt$.

  \textbf{(B)} If the request corresponding to $\an{n,i}$
  was sent by some $b\in \fAP{B}$ which is honest in $S_j$,
  then $b$ owns $i$.
\end{definition}

\section{Proof of
  Theorem~\ref{the:securityFixedBrowserID}}\label{app:proofbrowserid}

In order to prove Theorem~\ref{the:securityFixedBrowserID},
we have to prove Conditions A and B of
Definition~\ref{def:security-property}. These are proven
separately in what follows:

\subsection{Condition A}

We assume that Condition A is not satisfied and prove that
this leads to a contradiction. That is, we make the
following assumption (*): There is a run $\rho = s_0,
s_1,\dots$ of $\bidwebsystem$, a state $s_j = (S_j, E_j)$
in $\rho$, an $r \in \fAP{RP}$, an RP service token of the
form $\an{n,i}$ recorded in $r$ in the state $S_j(r)$ such
that $\an{n,i} \in
d_{N^\fAP{attacker}}(S_j(\fAP{attacker}))$ and the browser
owning $i$ is not fully corrupted in $S_j$.

Without loss of generality, we may assume that $\rho$ also
satisfies the following:

(**) Whenever a browser becomes
corrupted (i.e., either $\fullcorrupt$ or $\closecorrupt$)
in a processing step leading to some state $s_l$ in $\rho$,
this browser is triggered immediately afterwards again (in
the processing step leading to $s_{l+1}$) and sends the
full state of the web browser to the attacker process
$\fAP{attacker}$, which then receives this knowledge in
state $s_{l+2}$.  Afterwards, this browser is not triggered
anymore.

(***) For every term $\enc{t}{\pub(k')}$ for some $t \in
\terms$, $k' \in \nonces$ that is a subterm of the output
of a transition of $\Rasp$ but not of the input, i.e.,
$\Rasp$ has created $\enc{t}{k'}$ by itself, $\Rasp$ has
sent an HTTP message containing $t$ (unencrypted) to some
$d \in \mapDomain(\fAP{attacker})$ before.

If there is a run that satisfies (*), it is easy to turn
this run into a run that satisfies both (*) and (**). This
is because an attacker who obtains the state of the browser
can simulate the browser himself.  Moreover, it is easy to
turn a run that satisfies (*) and (**) but not (***) into a
run that satisfies all three properties by adding the
necessary requests from the script $\Rasp$.

Given (*), by definition of RPs, for $\an{n,i}$ there
exists a corresponding HTTPS request received by $r$, which
we call $\mi{req}_\text{cap}$, and a corresponding response
$\mi{resp}_\text{cap}$. The request must contain a valid
CAP $c$ and must have been sent by some atomic process $p$
to $r$. The response must contain $\an{n,i}$ and it must be
encrypted by some symmetric encryption key $k$ sent in
$\mi{req}_\text{cap}$.

In particular, it follows that the request and the response
must be of the following form, where $d_r= \mathsf{dom}(r)$
is the domain of $r$, $n_\text{cap}, k \in \nonces$ are
some nonces, and $c$ is some valid CAP:
\begin{align}
  \label{eq:proofreqcapA} \mi{req}_\text{cap} &=
  \ehreqWithVariable{\hreq{ nonce=n_\text{cap},
      method=\mPost, xhost=d_r, path=/, parameters=\an{},
      headers=[\str{Origin}: \an{d_r, \https}],
      xbody=c}}{k}{\pub(\mapKey(d_r))}\enspace
  ,\\ %
  \label{eq:proofrespcapA} \mi{resp}_\text{cap} &=
  \ehrespWithVariable{\hresp{ nonce=n_\text{cap},
      status=200, headers=\an{},
      xbody=\an{n,i}}}{k}\enspace .
\end{align}
Moreover, there must exist a processing step of the
following form where $m \leq j$, $a_r \in
\mapAddresstoAP(r)$, and $x$ is some address:
\[ s_{m-1} \xrightarrow[r \rightarrow
\{(x{:}a_r{:}\mi{resp}_\text{cap})\}]{(a_r{:}x{:}\mi{req}_\text{cap})
  \rightarrow r} s_{m}\enspace . \]

From the assumption and the definition of RPs it follows
that $c$ is issued for $d_r$ (otherwise, RP would not
accept the CAP, see Line~\ref{line:rp-checksig} of
Algorithm~\ref{alg:rp}). The nonce $n$ in $\an{n,i}$ is
chosen freshly and from RPs nonces $N^r$. It is not used
again by $r$ afterwards.

We assume that $s_j$ is the \emph{first} state in $\rho$
where $\an{n,i} \in
d_{N^\fAP{attacker}}(S_j(\fAP{attacker}))$ (i.e., there is
no $j' < j$, $\an{n,i} \in
d_{N^\fAP{attacker}}(S_{j'}(\fAP{attacker}))$).

We note that, by definition of attacker processes, the
attacker never discards any information, i.e., $t \in
d_{N^\fAP{attacker}}(S_{u}(\fAP{attacker}))$ implies $t \in
d_{N^\fAP{attacker}}(S_{u+1}(\fAP{attacker}))$ for every
term $t$ and $u \in \mathbb{N}$.

To conclude the proof, we now first prove several lemmas.

In what follows, given an atomic process $p$ and a message
$m$, we say that \emph{$p$ emits $m$} in a run
$\rho=s_0,s_1,\ldots$ if there is a processing step of the
form
\[ s_{u-1} \xrightarrow[p \rightarrow E]{} s_{u}\] for some
$u \in \mathbb{N}$, set of events $E$ and some addresses
$x$, $y$ with $(x{:}y{:}m) \in E$.

We say that an atomic process
$p$ \emph{created} a message $m$ (at some point) in a run
if $m$ is (congruent to) a subterm of a message emitted by $p$ in some
processing step and if there is no earlier processing step
where $m$ is a subterm of a message emitted by an atomic
process $p'$.

We say that a
browser $b$ \emph{accepted} a message (as a response to
some request) if the browser decrypted the message (if it
was an HTTPS message) and called the function
$\mathsf{PROCESSRESPONSE}$, passing the message and the
request (see Algorithm~\ref{alg:processresponse}).

We say that an atomic DY process \emph{$p$ knows a term
  $t$} in some state $s=(S,E)$ of a run if it can derive
the term from its knowledge, i.e., $t \in d_{N^p}(S(p))$.

We say that a \emph{script initiated a request
  $r$} if a browser triggered the script (in
Line~\ref{line:trigger-script} of
Algorithm~\ref{alg:runscript}) and the first component of
the $\mi{command}$ output of the script relation is either
$\tHref$, $\tIframe$, $\tForm$, or $\tXMLHTTPRequest$ such
that the browser issues the request $r$ in the same step as
a result.

\begin{lemma}\label{lemma:k-does-not-leak-from-honest-browser}
  If in a run $\rho$ of $\bidwebsystem$ an honest browser
  $b$ emits an HTTPS request of the form

  \[ \ehreqWithVariable{\mi{req}}{k}{\pub(k')} \]
  where $\mi{req}$ is an HTTP request, $k$ is a nonce
  (symmetric key), $k'$ is the private key of an RP or of
  LPO, and if in that run $b$ does not become fully
  corrupted, then all of the following statements are true: 
  \begin{enumerate}
  \item There is no state of $\bidwebsystem$ where any
    party except for $b$ and the owner of the key $k'$
    knows $k'$, thus no one except for $b$ and the owner of
    the key $k'$ can decrypt $\mi{req}$.
    \label{prop:attacker-cannot-decrypt}
  \item There is no state in the run $\rho$ where $k$ is
    known to any atomic process $p$ (i.e., in no state $s =
    (S,E)$ in $\rho$, $k \in d_{N^p}(S(p))$), except for
    the atomic processes $b$ and the owner of the public
    key (some RP or LPO). \label{prop:k-doesnt-leak}
  \item The value of the host header in $\mi{req}$ is a
    domain of the owner of the key
    $k'$. \label{prop:host-header-matches}
  \item Only the owner of the key $k'$ can create a response
    $r$ to this request that is accepted by
    $b$, i.e. the nonce of
    the HTTP request is not known to any atomic process
    $p$, except for the atomic process $b$ and the owner of
    the public key (some RP or LPO).\label{prop:only-owner-answers}
  \end{enumerate}
\end{lemma}

\begin{proof} First, we note that only the intended
  receiver can decrypt the message: The private keys of RPs
  and LPO are per definition only known to the respective
  parties. According to the definition of RPs and LPO, the
  keys do not leak to other parties, i.e., there is no
  state in a run $\rho$ of $\bidwebsystem$ where the keys
  are known to any other parties except their respective
  owners. This proves (\ref{prop:attacker-cannot-decrypt}).

  We can further see from the definition of the receivers
  (some RP or LPO) that they use the key $k$ only to
  encrypt the responses (Algorithms~\ref{alg:lpo}
  and~\ref{alg:rp}). In both definitions, $k$ is extracted
  from the message and discarded after encrypting the
  response. Note that neither RPs nor LPO can be
  corrupted. Hence, neither RPs nor LPO can leak $k$. From
  the definition of the browser $b$, we see that the key is
  always chosen from a fresh set of nonces
  (Line~\ref{line:takenonce-k} of
  Algorithm~\ref{app:mainalgorithmwebbrowserprocess}) that
  are not used anywhere else. Further, the key is stored in
  the browser's state in $\mi{pendingRequests}$, but
  discarded after receiving the response. The information
  from $\mi{pendingRequests}$ is not extracted or used
  anywhere else (in particular it is not accessible by
  scripts). If the browser becomes closecorrupted at some
  point in the run $\rho$, the key cannot be used anymore
  (compare Line~\ref{line:key-not-used-anymore} of
  Algorithm~\ref{alg:browsermain}). Hence, $k$ cannot leak
  from $b$ either. This proves (\ref{prop:k-doesnt-leak}).

  From Line~\ref{line:select-enc-key} of
  Algorithm~\ref{alg:browsermain} we can see that the
  encryption key for the request $\mi{req}$ was actually
  chosen using the host header of the message. The mapping
  from domains to encryption keys in $\bidwebsystem$ is
  always ``correct'', i.e., the owner of $k'$ is the owner
  of the domain that is given in the host header. This
  proves (\ref{prop:host-header-matches}).

  An HTTPS response $r$ that is accepted by $b$ as a
  response to the above request has to be encrypted with
  $k$. (This is checked by the browser using the
  $\mi{pendingRequests}$ state information that is not
  alterable by scripts or other browser actions.) This
  nonce, however, is only known to the owner of the public
  key and $b$. The browser $b$ cannot send responses. This
  proves (\ref{prop:only-owner-answers}).
\end{proof}

\begin{lemma} \label{lemma:https-document-origin} If in a
  run $\rho$ of $\bidwebsystem$ an honest browser $b$ has a
  document $d$ in its state with the origin $\an{\mi{dom},
    \https}$ where $\mi{dom} \in \mapDomain(\fAP{RP}) \cup
  \mapDomain(\fAP{LPO})$, then $b$ extracted (in
  Line~\ref{line:take-script} in
  Algorithm~\ref{alg:processresponse}) the script in that document from an HTTPS
  response that was emitted by the owner of the private key
  belonging to $\mi{dom}$.
\end{lemma}

\begin{proof}
  The origin of the document $d$ is set only once: In
  Line~\ref{line:set-origin-of-document} of
  Algorithm~\ref{alg:processresponse}. The values (domain
  and protocol) used there stem from the information about
  the request (say, $\mi{req}$) that led to loading of
  $d$. These values have been stored in
  $\mi{pendingRequests}$ between the request and the
  response actions. The contents of $\mi{pendingRequests}$
  are indexed by freshly chosen nonces and can never be
  altered or overwritten (only deleted when the response to
  a request arrives). The information about the request
  $\mi{req}$ was added to $\mi{pendingRequests}$ in
  Line~\ref{line:add-to-pendingrequests-https} (or
  Line~\ref{line:add-to-pendingrequests} which we can
  exclude as we will see later) of
  Algorithm~\ref{alg:browsermain}. In particular, the
  request was an HTTPS request iff a (symmetric) key was
  added to the information in $\mi{pendingRequests}$. When
  receiving the response to $\mi{req}$, it is checked
  against that information and accepted only if it is
  encrypted with the proper key and contains the same nonce
  as the request (say, $n$). Only then the protocol part of
  the origin of the newly created document becomes
  $\https$. The domain part of the origin (in our case
  $\mi{dom}$) is taken directly from the
  $\mi{pendingRequests}$ and is thus guaranteed to be
  unaltered.

  From Line~\ref{line:select-enc-key} of
  Algorithm~\ref{alg:browsermain} we can see that the
  encryption key for the request $\mi{req}$ was actually
  chosen using the host header of the message which will
  finally be the value of the origin of the document
  $d$. Since the honest browsers in $\bidwebsystem$ select
  the correct public keys for a domain, we can see that
  $\mi{req}$ was encrypted using the public key belonging
  to $\mi{dom}$. With
  Lemma~\ref{lemma:k-does-not-leak-from-honest-browser} we
  see that the symmetric encryption key for the response,
  $k$, is only known to $b$ and the respective RP or
  LPO. The same holds for the nonce $n$ that was chosen by
  the browser and included in the request. Thus, no other
  party than the owner of the private key belonging to
  $\mi{dom}$ can encrypt a response that is accepted by the
  browser $b$ and which finally defines the script of the
  newly created document.
\end{proof}

\begin{lemma} \label{lemma:https-script-origin} If in a run
  $\rho$ of $\bidwebsystem$ an honest browser $b$ issues an
  HTTP(S) request with the Origin header value
  $\an{\mi{dom}, \https}$ where $\mi{dom} \in
  \mapDomain(\fAP{RP}) \cup \mapDomain(\fAP{LPO})$, then
  that request was initiated by a script that $b$ extracted
  (in Line~\ref{line:take-script} in
  Algorithm~\ref{alg:processresponse}) from an HTTPS
  response that was emitted by the owner of the private key
  belonging to $\mi{dom}$.
\end{lemma}

\begin{proof} First, we can see that the request was
  initiated by a script: As it contains an origin header,
  it must have been a POST request (see the browser
  definition in Appendix~\ref{sec:descr-web-brows}). POST
  requests can only be initiated in
  Lines~\ref{line:send-form}, \ref{line:send-xhr} of
  Algorithm~\ref{alg:runscript} and
  Line~\ref{line:send-redirect} of
  Algorithm~\ref{alg:processresponse}. In the latter
  instance (Location header redirect), the request contains
  at least two different origins, therefore it is
  impossible to create a request with exactly the origin
  $\an{\mi{dom}, \https}$ using a redirect.  In the other
  two cases (FORM and XMLHTTPRequest), the request was
  initiated by a script.

  The Origin header of the request is defined by the origin
  of the script's document. With
  Lemma~\ref{lemma:https-document-origin} we see that the
  content of the document, in particular the script, was
  indeed provided by the owner of the private key belonging
  to $\mi{dom}$.
\end{proof}

\begin{lemma}\label{lemma:cookie-value-doesnt-leak}
  In a run $\rho$ of $\bidwebsystem$, if LPO sends a
  $\str{browserid\_state}$ cookie in a
  $\str{Set{\mhyphen}Cookie}$ header in an HTTPS response
  to an HTTPS request emitted by a browser $b$, there is no
  state in the run where the browser is honest and the
  attacker can derive the cookie value from its own
  knowledge.
\end{lemma}
\begin{proof}
  We can see that the browser is honest when sending the
  request (otherwise, it would not do so, (**)). With
  Lemma~\ref{lemma:k-does-not-leak-from-honest-browser} and
  as in the proof for
  Lemma~\ref{lemma:https-document-origin} we see that the
  sender of the request to LPO (say, $\mi{req}$) is the
  same as the receiver, namely browser $b$. As the message
  is transferred over HTTPS, the attacker cannot read the
  cookie from the response.

  The $\str{browserid\_state}$ cookie is sent to $b$ as an
  httpOnly secure session cookie (compare
  Line~\ref{line:set-session-cookie} in
  Algorithm~\ref{alg:lpo}). When the response arrives at
  $b$, the cookie is transferred to the cookie store
  (Line~\ref{line:set-cookie} of
  Algorithm~\ref{alg:processresponse}) which is indexed by
  domains.  The cookie information can be accessed by
  scripts (Line~\ref{line:assemble-cookies-for-script} of
  Algorithm~\ref{alg:runscript}) and can be added to
  requests (Line~\ref{line:assemble-cookies-for-request} of
  Algorithm~\ref{alg:send}). As the
  $\str{browserid\_state}$ cookie is an httpOnly cookie we
  can rule out the first case.  In the second case, the
  cookie can only be added to requests to the origin
  $\an{\mapDomain(\fAP{LPO}), \https}$, as the cookie is
  marked as secure (as defined in
  Line~\ref{line:cookie-rules-http} in
  Algorithm~\ref{alg:send}). These properties hold as long
  as the browser is not corrupted. 

  As a last step, we have to rule out that LPO or the
  browser use a cookie value that is known to the attacker
  via some other way. We will see that any cookie value was
  initially chosen by LPO.

  First, we can see that the cookie value was either in the
  browser's knowledge before it received the
  $\str{browserid\_state}$ header or that it was chosen
  freshly by LPO. The only line where LPO sets the cookie
  is in Line~\ref{line:set-session-cookie} of
  Algorithm~\ref{alg:lpo}. From the lines before, it is
  easy to see that the session value that finally becomes
  the cookie value was either provided as a cookie in the
  request or is chosen from the set of unused nonces. In
  Lemma~\ref{lemma:cookie-value-cannot-be-overwritten} we
  see that any value that is contained in a request sent by
  an honest browser to LPO was initially chosen by LPO.

  We see that the attacker cannot know the cookie value as
  long as the browser stays honest, which proves the lemma.
\end{proof}

\begin{lemma} \label{lemma:xsrf-token-needs-sessionid} In
  every state $s = (S,E)$ of run $\rho$ of $\bidwebsystem$,
  for every $\mi{xsrfToken}$ of an LPO session and its
  session ID $\mi{sessionid}$, if $\mi{xsrfToken} \in
  d_{N^\fAP{attacker}}(S(\fAP{attacker}))$, then
  $\mi{sessionid} \in
  d_{N^\fAP{attacker}}(S(\fAP{attacker}))$, i.e., an
  attacker can only know an $\mi{xsrfToken}$ value for an
  LPO session if he knows the $\str{browserid\_state}$
  session ID of that session.
\end{lemma}
\begin{proof}
  The $\mi{xsrfToken}$ is chosen by $\fAP{LPO}$
  (Line~\ref{line:chose-xsrftoken} in
  Algorithm~\ref{alg:lpo}). If LPO receives a POST request
  with the path $\str{/ctx}$ that contains a
  $\str{browserid\_state}$ cookie containing a
  $\mi{sessionid}$ that is in its list of valid sessions,
  it returns $\mi{xsrfToken}$ as part of the response.  If
  LPO receives a request to the same URL without a valid
  session ID, it creates a new session and returns
  $\mi{sessionid}$ as well as a freshly chosen
  $\mi{xsrfToken}$ in the response. For other requests (to
  other URLs, etc.) $\mi{xsrfToken}$ is not a part of the
  response at all.

  The $\mi{xsrfToken}$ is only transferred over HTTPS: LPO
  only reacts to HTTPS requests
  (Line~\ref{line:lpo-does-https-only} of
  Algorithm~\ref{alg:lpo}) and the request that is sent
  from the browser to LPO to retrieve $\mi{xsrfToken}$ is
  explicitly sent over HTTPS
  (Line~\ref{line:ctx-over-https} of
  Algorithm~\ref{alg:scriptlpocif} or
  Line~\ref{line:ctx-over-https-2} of
  Algorithm~\ref{alg:scriptlpold}). Thus, if an honest
  browser sends a request to LPO, the attacker cannot read
  the response if the browser stays honest
  (Lemma~\ref{lemma:k-does-not-leak-from-honest-browser}). If
  the browser becomes corrupted later, the attacker learns
  the $\mi{sessionid}$ and the $\mi{xsrfToken}$ at the same
  time. The LPO script that has access to $\mi{xsrfToken}$
  in the browser's state does not sent out this part of the
  state to origins other than LPO's (see
  Algorithm~\ref{alg:scriptlpocif} and
  Algorithm~\ref{alg:scriptlpold}) and the $\mi{xsrfToken}$
  is stored only temporarily in the script's state (as part
  of the \emph{context}, see Line~\ref{line:store-context}
  in Algorithm~\ref{alg:scriptlpocif} and
  Line~\ref{line:store-context-2} in
  Algorithm~\ref{alg:scriptlpold}), such that it is never
  released when the browser is honest or closecorrupted.
  
  We can see that the attacker knows $\mi{sessionid}$
  whenever he knows $\mi{xsrfToken}$, which proves the
  lemma.
\end{proof}

\begin{lemma} \label{lemma:cookie-value-cannot-be-overwritten}
  In a run $\rho$ of $\bidwebsystem$, for any HTTPS request
  $\mi{req}$ that is emitted by an honest browser $b$ and
  that is encrypted with the public key of LPO, if there is
  a $\str{Cookie}$ header in $\mi{req}$ containing a cookie
  with the name $\str{browserid\_state}$, then there is an
  HTTPS response that was emitted by LPO previously in
  the run and that was accepted by $b$. In this response, a
  $\str{Set{\mhyphen}Cookie}$ header was sent with the name
  $\str{browserid\_state}$ and the same value as the
  $\str{browserid\_state}$ cookie in $\mi{req}$.
\end{lemma}
\begin{proof}
  The cookie that is sent in $\mi{req}$ was taken from the
  cookie list that is stored in the browser state (see
  Algorithm~\ref{alg:send}). Cookies are stored per-domain,
  i.e., $\domLPO$ in this case. Adding a cookie to this
  list can be achieved by adding a
  $\str{Set{\mhyphen}Cookie}$ to a response on a request to
  $\domLPO$ or by setting the cookie from a script in a
  document with the origin $\an{\domLPO, x}$ where $x \in
  \{\http, \https\}$. The domain $\domLPO$ is part of the
  $\mi{sts}$ list in honest browsers (see
  Section~\ref{sec:browsers}) thus the browser $b$ never
  contacts the insecure origin $\an{\mapDomain(\fAP{LPO}),
    \http}$. Thus, responses and scripts can only be
  received from the origin $\an{\domLPO, \https}$ (see
  Lemma~\ref{lemma:k-does-not-leak-from-honest-browser}
  Property~(\ref{prop:only-owner-answers}) and
  Lemma~\ref{lemma:https-document-origin}). The LPO scripts
  $\mi{script\_LPO\_cif}$ and $\mi{script\_LPO\_ld}$ do not
  set cookies, thus the only possible way that a cookie can
  be stored in the browser's list of cookies is when LPO
  adds a $\str{Set{\mhyphen}Cookie}$ header to a HTTPS
  response. Obviously, this header has to have the same
  value as the cookie that is finally returned to the
  server. This proves the lemma.
\end{proof}

\begin{lemma} \label{lemma:browser-authenticates-itself-only}
  In a run $\rho$ of $\bidwebsystem$, if an honest browser
  $b$ emits a request $\mi{req}_\text{auth}$ that is
  received by LPO and leads to the
  authentication\footnote{See Section~\ref{sec:lpo} for an
    explanation on the authentication at LPO.} of an LPO
  session identified by the sessionid $\mi{sessionid}$,
  then the identity $i$, for which the session was
  authenticated, is owned by $b$, i.e., $i \in
  \mapIDtoOwner^{-1}(b)$.
\end{lemma}
\begin{proof}
  For authentication, a request of the following form has
  to be received by LPO:
  \begin{align}
    \nonumber \mi{req}_\text{auth} =
    \mathsf{enc}_\mathsf{a}( \langle \langle \cHttpReq,
    n_2, \mPost, \domLPO, \str{/auth}, \an{}, \\ \nonumber
    \langle [ \str{Cookie}: [\str{browserid\_state}: \mi{sessionid}], \dots] \rangle, \\
    {\an{s,\mi{xsrfToken}}}, \rangle, k'' \rangle,
    \pub(\mapKey(\domLPO)) )
  \end{align}
  The request $\mi{req}_\text{auth}$ contains the secret
  $s$ and the $\mi{xsrfToken}$ that, by definition of LPO,
  is stored at LPO along with $\mi{sessionid}$.  In an
  honest browser (which $b$ is), this request can only be
  caused by a script (or through a redirection, which again
  would require a script to initiate the request in the
  first place). There are three scripts that can issue such
  a request: the attacker script and both LPO scripts. In
  the latter case, the LPO scripts will provide the browser
  secret as~$s$ and hence, authenticate for an identity
  owned by the browser. In the former case, the attacker
  script needs to know $\mi{xsrfToken}$. Hence, by
  Lemma~\ref{lemma:xsrf-token-needs-sessionid} he needs to
  know $\mi{sessionid}$. However, the $\mi{sessionid}$
  value does not leak from the honest browser $b$
  (Lemma~\ref{lemma:cookie-value-doesnt-leak}) and cannot
  be set by the attacker
  (Lemma~\ref{lemma:cookie-value-cannot-be-overwritten}). Hence,
  the attacker cannot know $\mi{sessionid}$, and hence, by
  Lemma~\ref{lemma:xsrf-token-needs-sessionid} he cannot
  know $\mi{xsrfToken}$, and hence, $\mi{req}_\text{auth}$
  cannot have been initiated by the attacker script.
\end{proof}

\begin{lemma} \label{lemma:lpo-scripts-issue-only-to-so} In
  a run $\rho$ of $\bidwebsystem$ if either
  $\mi{script\_LPO\_cif}$ or $\mi{script\_LPO\_ld}$ were
  loaded into a document with HTTPS origin and are used to
  create a CAP $c$, i.e., if $c$ is contained in a \pm that
  is sent in Line~\ref{line:send-login-pm} of
  Algorithm~\ref{alg:scriptlpocif} or in
  Line~\ref{line:send-login-pm-2} of
  Algorithm~\ref{alg:scriptlpold}, then the origin for
  which $c$ is issued is the origin of the script that
  receives this \pm. 
\end{lemma}
\begin{proof}
  Looking at the case when $\mi{script\_LPO\_cif}$ issues
  the CAP in Line~\ref{line:create-login-pm} of
  Algorithm~\ref{alg:scriptlpocif}, the origin for which
  the IA is issued in this case is determined by the
  element $\comp{s'}{parentOrigin}$ of the script's
  state. This element is only written to in
  Line~\ref{line:set-parent-origin} of
  Algorithm~\ref{alg:scriptlpocif}. Its value is the sender
  origin of the \pm requesting the CAP. The very same value
  determines the only allowed receiver origin of the \pm
  that returns the CAP
  (Line~\ref{line:create-login-pm}). With a very similar
  argument (different line numbers), we can see that the
  statement for $\mi{script\_LPO\_ld}$ holds true as well.
\end{proof}

\begin{lemma}\label{lemma:caps-from-rp-are-fine}
  In a run $\rho$ of $\bidwebsystem$, if a CAP $c =
  \an{\mi{uc}, \mi{ia}}$ is sent by
  $\mi{script\_RP\_index}$ (Line~\ref{line:send-cap} of
  Algorithm~\ref{alg:scriptrpindex}) running in an honest
  browser $b \in \fAP{B}$ in a document with origin
  $\an{d_r, \https}$ as an HTTP(S) message to an RP $r \in
  \fAP{RP}$, where $d_r=\mapDomain(r)$,
  $\mi{uc}=\sig{\an{i, \pub(k_u)}}{k^\LPO}$, and
  $ia=\sig{o}{k_u'}$, then all of the following statements
  are true:
  \begin{enumerate}
  \item $c$ is a valid CAP. In particular, $k_u =
    k_u'$. \label{prop:c-is-matching}
  \item $\mi{uc}$ was created by LPO and transferred to
    $\mi{script\_RP\_index}$ in a \pm by a script of LPO
    running in $b$ (either $\mi{script\_LPO\_ld}$, \pm sent
    in Line~\ref{line:send-login-pm} of
    Algorithm~\ref{alg:scriptlpold} or
    $\mi{script\_LPO\_cif}$, \pm sent in
    Line~\ref{line:send-login-pm-2} of
    Algorithm~\ref{alg:scriptlpocif}) loaded into a
    document with the origin $\an{\domLPO,
      \https}$. \label{prop:uc-created-by-lpo}
  \item $\mi{ia}$ contains the origin $o = \an{d_r,
      \https}$ \label{prop:ia-origin}.
  \item $k_u$ is not known to any atomic DY process except
    for $b$, as long as $b$ is not
    fullycorrupted. \label{prop:ku-doesnt-leak}
  \item $\mi{uc}$ is issued for an identity $i \in
    \mapIDtoOwner^{-1}(b)$ \label{uc-id-is-browser-id}.
  \end{enumerate}
\end{lemma}

\begin{proof}
  As we know that the message is sent from the origin
  $\an{d_r, \https}$, we know that the script
  $\mi{script\_RP\_index}$ was loaded over an HTTPS origin
  (see Lemma~\ref{lemma:https-script-origin}). Its script
  state cannot be manipulated by scripts loaded from a
  different origin (see Algorithms~\ref{alg:getwindow}
  and~\ref{alg:runscript}).

  The only transitions of the script
  $\mi{script\_RP\_index}$ which can send out a request to
  $r$ are the ones starting out from state
  $\mi{scriptstate}.\str{q} = \str{sendCAP}$
  (Line~\ref{line:state-sendcap} of
  Algorithm~\ref{alg:scriptrpindex}). These transitions
  take the CAP $c$ from $\mi{scriptstate}.\str{cap}$. The
  only transitions before which could have written
  something into this place in the scriptstate are the ones
  where $\mi{scriptstate}.\str{q} = \str{default}$
  (Line~\ref{line:state-default}) when handling a \pm from
  origin LPO (we can overapproximate here by ignoring all
  other side restrictions of this transition, e.g. having
  $\mi{scriptstate}.\str{dialogRunning} = \True$). This
  means that the CAP $c$ was sent by a script with origin
  LPO. Since origin LPO is also an HTTPS origin, the script
  must have been sent by LPO
  (Lemma~\ref{lemma:https-document-origin}).

  The \pm that was received by $\mi{script\_RP\_index}$ is
  checked to be a sequence with the first element being
  $\str{login}$ or $\str{response}$. Such \pms are issued
  by LPO only in Line~\ref{line:create-login-pm} of
  Algorithm~\ref{alg:scriptlpocif}
  ($\mi{script\_LPO\_cif}$) and in
  Line~\ref{line:create-response-pm} of
  Algorithm~\ref{alg:scriptlpold}
  ($\mi{script\_LPO\_ld}$). In both cases, $\mi{ia}$ is
  signed using the private key $k_u$ that is taken from the
  respective script's state. This element of the script's
  state is only written to once, and with a freshly chosen
  nonce (Line~\ref{line:choose-key} of
  Algorithm~\ref{alg:scriptlpocif} and
  Line~\ref{line:choose-key-2} of
  Algorithm~\ref{alg:scriptlpold}, respectively).  In both
  cases, starting in the script state $\str{requestUC}$, an
  \xhr to LPO is sent to have $\pub(k_u)$ signed by
  LPO. From the response to this request, $\mi{uc}$ is
  extracted. The request is always sent over HTTPS to LPO.
  Lemma~\ref{lemma:k-does-not-leak-from-honest-browser}, in
  particular Property~(\ref{prop:only-owner-answers}),
  applies. Therefore, we see that $\mi{uc}$ was actually
  sent by LPO.

  Looking at the LPO definition
  (Line~\ref{line:lpo-issue-uc} of Algorithm~\ref{alg:lpo})
  we see that LPO only sends out freshly created
  $\mi{uc}$'s. LPO only issues valid UCs (if any). Once
  returned to $\mi{script\_LPO\_cif}$ or
  $\mi{script\_LPO\_ld}$, the UC is combined with an IA and
  sent to $\mi{script\_RP\_index}$ (which is determined by
  the sender of the initial CAP request and its
  origin). This script sends the CAP to $r$. Thus, the CAP
  that is sent is always valid, which proves
  (\ref{prop:c-is-matching}). Further, the UC was always
  created by LPO, proving (\ref{prop:uc-created-by-lpo}).
 
  Property~(\ref{prop:ia-origin}) follows immediately from
  the above observations (i.e., $\mi{script\_LPO\_cif}$ or
  $\mi{script\_LPO\_ld}$ were loaded over HTTPS and are
  used to create the CAP) and
  Lemma~\ref{lemma:lpo-scripts-issue-only-to-so}.

  To prove (\ref{prop:ku-doesnt-leak}), we observe that the
  key $k_u$ is always chosen freshly and that it is stored
  only in the script's state. It is not sent to any party,
  not even LPO. The key therefore cannot leak as long as
  the browser is not fullycorrupted (if it becomes
  closecorrupted, the key is removed together with the
  document's state). This proves
  (\ref{prop:ku-doesnt-leak}).

  Property~(\ref{uc-id-is-browser-id}) follows immediately
  with the observations in the proof of
  Property~(\ref{prop:uc-created-by-lpo}) and
  Lemma~\ref{lemma:browser-authenticates-itself-only}.
\end{proof}

\begin{lemma} \label{lemma:ia-does-not-leak} If in a run
  $\rho$ of $\bidwebsystem$ an IA $\mi{ia}$ for an origin
  $\an{d_r, \https}$ where $d_r \in \mapDomain(\fAP{RP})$
  is signed in the scripts $\mi{script\_LPO\_ld}$ or
  $\mi{script\_LPO\_cif}$ in an honest browser $b$, and
  these scripts were loaded over HTTPS from LPO, then at
  most $b$ and the RP $r = \mapDomain^{-1}(d_r)$ know
  $\mi{ia}$.
\end{lemma}
\begin{proof}
  The scripts $\mi{script\_LPO\_ld}$ and
  $\mi{script\_LPO\_cif}$ send the $\mi{ia}$ to the parent or opener
  window (respectively) using  a \pm. For this \pm, the only allowed
  receiver origin is the same as the origin for which
  $\mi{ia}$ was issued, so in our case $\an{d_r, \https}$
  (see proof for Lemma~\ref{lemma:caps-from-rp-are-fine}
  Property~(\ref{prop:ia-origin})). The script
  $\mi{script\_RP\_index}$, which thus must be the
  receiver, sends the complete CAP (containing $\mi{ia}$)
  to RP using HTTPS. The RP discards the CAP after checking
  it. The CAP and especially $\mi{ia}$ therefore cannot
  leak. 
\end{proof}

\begin{lemma} \label{lemma:attacker-cannot-request-uc-for-id}
  In a run $\rho$ of $\bidwebsystem$, if LPO creates a
  message containing a UC $\mi{uc}$ for an identity $i$ of
  a browser $b$, then there is no state in the run $\rho$
  where $b$ is honest and attacker knows the private key
  $k_u$ corresponding to the public key $\pub(k_u)$ that
  was signed in $\mi{uc}$.
\end{lemma}
\begin{proof}
  First, it is easy to see that the
  Lines~\ref{line:start-certreq}--\ref{line:end-certreq} in
  Algorithm~\ref{alg:lpo} have to be used in the transition
  to create $\mi{uc}$: At no other point in the definition
  of LPO $\mi{uc}$ is created or emitted. From
  Line~\ref{line:sign-pubkey} and following it is easy to
  see that a request of the following form has to be sent
  to LPO in order to create $\mi{uc}$:
  \begin{align}
    \nonumber \mi{req}_\text{uc} = \mathsf{enc}_\mathsf{a}(
    \langle \langle \cHttpReq, n_1, \mPost, \domLPO,
    \str{/certreq}, \an{}, \\ \nonumber
    \an{[\str{Cookie}: [\str{browserid\_state}: \mi{sessionid}], \dots]}, \\
    {\an{i,\pub(k_u),\mi{xsrfToken}}} \rangle, k'' \rangle,
    pub(\mapKey(\mapDomain(\fAP{LPO}))) )
  \end{align}

  We can see that this message is encrypted with
  $pub(\mapKey(\mapDomain(\fAP{LPO})))$ and thus the
  attacker cannot decrypt it. There are now two
  cases:

  \begin{itemize}
  \item \textbf{The attacker knows $k''$:} In this case, we
    can see with
    Lemma~\ref{lemma:k-does-not-leak-from-honest-browser}
    that no honest browser has created
    $\mi{req}_\text{uc}$. As RP, LPO, and dishonest
    browsers do not emit requests in general, only the
    attacker can have created this request. For this, he
    needs to know $\mi{xsrfToken}$ and $\mi{sessionid}$.

    The attacker could use a $\mi{sessionid}$ value that
    was first issued to a browser that was honest when
    $\mi{req}_\text{uc}$ was created or to some other party
    (a dishonest browser or himself).

    The first case can be ruled out, as the attacker cannot
    know the $\mi{sessionid}$ value
    (Lemma~\ref{lemma:cookie-value-doesnt-leak}).

    In the second case, he cannot create a message that
    leads to the authentication of the session himself
    (which would require knowledge of the secret for
    identity $i$) and he cannot force the owner browser of
    $i$ to authenticate the session
    (Lemma~\ref{lemma:browser-authenticates-itself-only}).
  \item \textbf{The attacker does not know $k''$:} In this
    case, the request was not created by the attacker. As
    above, RP, LPO and dishonest browsers do not create
    requests. Thus, this request was created by an honest
    browser (``honest'' in the state when
    $\mi{req}_\text{uc}$ was created) and an honest script
    in that browser (with a dishonest script, the attacker
    would need to know $\mi{xsrfToken}$, which he does not
    according to
    Lemma~\ref{lemma:xsrf-token-needs-sessionid} and
    Lemma~\ref{lemma:cookie-value-doesnt-leak}).  If an
    honest script, i.e., $\mi{script\_LPO\_cif}$ or
    $\mi{script\_LPO\_ld}$, is used, the attacker does not
    learn $k_u$ (Lemma~\ref{lemma:caps-from-rp-are-fine}
    Property~(\ref{prop:ku-doesnt-leak})).
  \end{itemize}
  As we can see, in both cases the attacker does not learn
  $k_u$, which proves the statement.
  
\end{proof}

\begin{lemma} \label{lemma:honest-browser-uses-https-rp-script}
  If in a run $\rho$ of $\bidwebsystem$ a browser $b$
  created the request $\mi{req}_\text{cap}$ 
  defined in (\ref{eq:proofreqcapA}), then (i)
  $\mi{req}_\text{cap}$ was sent from the script
  $\mi{script\_RP\_index}$ that was loaded over an HTTPS
  origin from $r$ while the browser was honest or (ii)
  $\mi{req}_\text{cap}$ was encrypted by the attacker
  script while the browser was honest and the attacker
  knows the CAP $\mi{c}$ and the symmetric key $k$.
\end{lemma}
\begin{proof}

  We can see that if the browser is dishonest, it did
  encrypt $\mi{req}_\text{cap}$ while it was still honest:
  With assumption (**) we can see that dishonest browsers
  only send their state to the attacker. Thus, the
  encrypted message must have been in the state of the
  browser before corruption. So for both (i) and (ii) we
  know that the browser was honest.

  In an honest browser, the browser itself can create
  encrypted requests (when an HTTPS request is sent) and
  scripts can create encrypted requests (by assembling and
  encrypting the message in the script relation).  

  In the former case (HTTPS request), which corresponds to
  statement (i) of the lemma, we see by
  Lemma~\ref{lemma:https-script-origin} that the script
  that initiated $\mi{req}_\text{cap}$ was actually loaded
  from $r$ using HTTPS ($r$ is the owner of $d_r$), and
  that it was not altered by any other party.  

  In the latter case (script encrypted request), which
  corresponds to statement (ii) in the lemma, we see that
  the honest scripts do not encrypt messages and thus, the
  attacker script is the only script that can do so. To do
  so, the script needs to know every component of
  $\mi{req}_\text{cap}$ before the encryption, in
  particular $k$ and $c$. These have been sent to the
  attacker before the encryption according to (***). Thus,
  the attacker must know $k$ and $c$ before.
\end{proof}

Let $m$ be the message that was passed to $\fAP{attacker}$
leading to $s_j$ for some addresses $x$ and $y$ (with $s_j$
as defined in (*)). That is:
\[ s_{j-1} \xrightarrow{(x{:}y{:}m) \rightarrow \fAP{attacker}}
s_j \enspace .\]
By our assumption, we know that $\an{n,i} \not\in
d_{N^\fAP{attacker}}(S_{j-1}(\fAP{attacker}))$ and that
$\an{n,i} \in d_{N^\fAP{attacker}}(S_{j-1}(\fAP{attacker}),
m)$.

We now distinguish two cases: (i) The attacker does not
know $k$ in $s_j$ (i.e., cannot derive $k$ in state
$s_j$). (ii) The attacker can derive $k$ in $s_j$. In both
cases we lead (*) to a contradiction.

\subsubsection{The attacker does not know $k$ in $s_j$}

We now assume that $k \not\in
d_{N^\fAP{attacker}}(S_j(\fAP{attacker}))$, i.e., the
attacker does not know $k$ in $s_j$. In particular, we have
that $k \not\in
d_{N^\fAP{attacker}}(S_{j-1}(\fAP{attacker}))$.

We distinguish between the kind of atomic processes that
potentially have created $\mi{req}_\text{cap}$. In all
cases, we arrive at a contradiction.

\begin{itemize}
\item \textbf{The browser that owns $i$} created
  $\mi{req}_\text{cap}$: By
  Lemma~\ref{lemma:honest-browser-uses-https-rp-script} it
  follows that the browser was honest when encrypting
  $\mi{req}_\text{cap}$, and $\mi{req}_\text{cap}$ was
  initiated by $\mi{script\_RP\_index}$, which was
  delivered over HTTPS from $r$.  Note that we can rule out
  case (ii) in the lemma, as the attacker does not know
  $k$.

  This script initiated $\mi{req}_\text{cap}$ and it is
  easy to see that this script (or no script at all)
  receives the corresponding response: From the browser
  definition, we see that \xhr responses are delivered to
  the document with the same nonce as the document that
  initiated the request
  (Line~\ref{line:process-xhr-response} in
  Algorithm~\ref{alg:processresponse}). Other documents
  have no access to the data from this document, except for
  same-origin documents (this is ensured by the
  $\mathsf{Clean}$ function that is used in
  Line~\ref{line:clean-tree} of
  Algorithm~\ref{alg:runscript} and by the
  $\mathsf{GETWINDOW}$ function
  (Algorithm~\ref{alg:getwindow}) that determines the
  windows which can be manipulated by other
  scripts). However, other same-origin documents can only
  contain the script $\mi{script\_RP\_index}$ (this is the
  only script that RP sends, and with
  Lemma~\ref{lemma:https-document-origin} we see that other
  same-origin documents cannot have been sent by the
  attacker). Other manipulations to the window of the
  document (e.g., navigating the window away) change the
  active document in the window
  (Algorithm~\ref{alg:processresponse}) and could only
  prevent the script from receiving the response.

  From Algorithm~\ref{alg:scriptrpindex} it is easy to see
  that after $\an{n,i}$ is delivered back to
  $\mi{script\_RP\_index}$ after $\mi{resp}_\text{cap}$ was
  received, nothing happens with $\an{n,i}$: If the browser
  is uncorrupted, only same-origin scripts have access to
  it (as shown above), but there are no scripts which use
  the information. The information can therefore not leak
  to the attacker. If the browser is closecorrupted before
  receiving $\mi{resp}_\text{cap}$, the attacker cannot
  derive $\an{n,i}$ from its information, as the decryption
  key is lost. If the browser is closecorrupted after
  receiving $\mi{resp}_\text{cap}$, by definition of
  close-corruption, $\an{n,i}$ is removed from the
  browser's state before the browser can be controlled by
  the adversary. By the assumption in (*), the browser
  cannot be fullycorrupted at any point in the run. Hence,
  in contradiction to (*), the attacker cannot obtain
  $\an{n,i}$.
\item A \textbf{browser that does not own $i$} created
  $\mi{req}_\text{cap}$: In this case, it still holds that
  the browser was honest when encrypting
  $\mi{req}_\text{cap}$ and the script
  $\mi{script\_RP\_index}$ created the request and was
  loaded over HTTPS
  (Lemma~\ref{lemma:honest-browser-uses-https-rp-script}). With
  Lemma~\ref{lemma:https-script-origin} and
  Lemma~\ref{lemma:caps-from-rp-are-fine}, in particular
  Properties~(\ref{prop:uc-created-by-lpo}) and
  (\ref{uc-id-is-browser-id}), we see that the RP script
  only initiates HTTPS requests containing CAPs that have
  been created by LPO and for an identity of the browser.
  This is in contradiction to the fact that $i$ is not
  owned by the browser but $\mi{req}_\text{cap}$ contains a
  CAP for $i$. Hence, $\mi{req}_\text{cap}$ cannot have
  been created by this browser.
\item \textbf{RPs or LPO} created $\mi{req}_\text{cap}$: As
  per their definitions (Algorithms~\ref{alg:lpo} and
  \ref{alg:rp}), they do not initiate or create HTTP(S)
  requests.
\item \textbf{The attacker process} created
  $\mi{req}_\text{cap}$: It is clear that any atomic
  process that created $\mi{req}_\text{cap}$ needs to know
  $k$. It follows, by our assumption that the attacker
  cannot derive $k$, that the attacker has not created
  $\mi{req}_\text{cap}$.
\end{itemize}

\subsubsection{The attacker knows $k$ in $s_j$}

As above, we distinguish between the kind of atomic
processes that potentially have created the request
$\mi{req}_\text{cap}$. We will see that the attacker needs
to know the CAP $c$ to learn $\an{n,i}$.
\begin{itemize}
\item \textbf{The browser that owns $i$} created
  $\mi{req}_\text{cap}$: By our assumption (*), this
  browser cannot be fully corrupted in the run. By
  Lemma~\ref{lemma:honest-browser-uses-https-rp-script}, it
  follows in the case (i) that $\mi{script\_RP\_index}$
  sent the request and that $k$ cannot be known by the
  attacker (with
  Lemma~\ref{lemma:k-does-not-leak-from-honest-browser},
  Property~(\ref{prop:k-doesnt-leak})) and hence, the
  browser cannot have created $\mi{req}_\text{cap}$.  In
  the case (ii) it follows that the attacker needs to know
  the CAP $c$ in order to create the request.
\item \textbf{A browser not owning $i$} created
  $\mi{req}_\text{cap}$: By
  Lemma~\ref{lemma:honest-browser-uses-https-rp-script}, we
  see that the browser was honest while encrypting
  $\mi{req}_\text{cap}$ and (i) that
  $\mi{script\_RP\_index}$ sent the request. With
  Lemma~\ref{lemma:caps-from-rp-are-fine}
  Property~(\ref{uc-id-is-browser-id}) we see that the
  browser cannot have created $\mi{req}_\text{cap}$ because
  it only creates requests for its own identities. In the
  case of (ii) we see that, again, the attacker has to know
  the cap $c$ in order to create the request.
\item \textbf{RPs or LPO} created $\mi{req}_\text{cap}$: As
  per their definitions (Algorithms~\ref{alg:lpo} and
  \ref{alg:rp}), they do not emit HTTP requests.
\item \textbf{The attacker process} created
  $\mi{req}_\text{cap}$: It is clear that any atomic
  process that created $\mi{req}_\text{cap}$ needs to know
  $c$. 
\end{itemize}

As we can see, the attacker needs to know $c = \an{\mi{uc},
  \mi{ia}}$ before he is able to create
$\mi{req}_\text{cap}$.  With
Lemma~\ref{lemma:attacker-cannot-request-uc-for-id} we know
that the attacker cannot request $\mi{uc}$ itself with the
identity $i$, and thus, he cannot know the key $k_u$ that
was signed in $\mi{uc}$. Neither can any browser other than
$b$ know $k_u$, otherwise the attacker could corrupt this
browser and learn $k_u$.  The key $k_u$ is needed to create
$\mi{ia}$, therefore only in $b$ the identity assertion
$\mi{ia}$ can be created, and it can only be created by
$\mi{script\_LPO\_cif}$ or $\mi{script\_LPO\_ld}$ (LPO
checks the origin of the request for $\mi{uc}$, and only
the script that sends $\mi{uc}$ knows $k_u$).  With
Lemma~\ref{lemma:ia-does-not-leak} we see that the attacker
cannot learn $\mi{ia}$.

Hence we can see that the attacker cannot know the CAP $c$
that he needs in order to to create
$\mi{req}_\text{cap}$. In particular, he cannot know the
key $k$ that was used to encrypt the response
$\mi{resp}_\text{cap}$, in contradiction to the assumption
that the attacker knows $k$.  \hfill\ensuremath{\square}

\subsection{Condition B}

We assume that Condition B is not satisfied and prove that
this leads to a contradiction. That is, we make the
following assumption: There is a run $\rho$ of
$\bidwebsystem$, a state $(S_j, E_j)$ in $\rho$, an $r\in
\fAP{RP}$, an RP service token of the form $\an{n,i}$
recorded in $r$ in the state $S_j(r)$, the request
corresponding to $\an{n,i}$ was sent by some $b\in \fAP{B}$
which is honest in $S_j$ and $b$ does not own~$i$.

Again, without loss of generality, we may assume that
$\rho$ satisfies (**) and (***) as above.

As above, the request $\mi{req}_\text{cap}$ corresponding
to $\an{n,i}$ is of the following form
\begin{align}
  \mi{req}_\text{cap} &= \ehreqWithVariable{\hreq{
      nonce=n_\text{cap}, method=\mPost, xhost=d_r,
      path=/, parameters=\an{}, headers=[\str{Origin}:
      \an{d_r, \https}], xbody=c}}{k}{\pub(k_r)}
\end{align}
with $d_r$, $n_\text{cap}$, $\pub(k_r)$, $k$, $c$ as
before.

By Lemma~\ref{lemma:https-script-origin}, this request was
actually initiated by a script of $r$, which can only be the
script
$\mi{script\_RP\_index}$. Lemma~\ref{lemma:caps-from-rp-are-fine}
says that for any such request, $b$ must be the owner of
$i$ which is a contradiction to the
assumption.\hfill\ensuremath{\square}

\end{document}

